\documentclass[10pt]{amsart}
\usepackage{amssymb,amsfonts,latexsym,amscd}
\usepackage{amsmath}
 \newtheorem{thm}{Theorem}[section]
 \newtheorem{cor}[thm]{Corollary}
 \newtheorem{lem}[thm]{Lemma}
 \newtheorem{prop}[thm]{Proposition}
 \theoremstyle{definition}
 \newtheorem{defn}[thm]{Definition}
 \theoremstyle{remark}
 \newtheorem{rem}[thm]{Remark}
 \newtheorem*{ex}{Example}
 \numberwithin{equation}{section}

\def\H{{\mathcal H}}

\def\pst{\Psi^{{\mbox{\tiny trial}}}} 

\def\pgs{\Psi^{{\mbox{\tiny GS}}}}    

\def\gvar{\phi} 

\def\su{\!\!\uparrow} 
\def\zo{\!\!\begin{pmatrix} 1 \\0 \end{pmatrix}}

\def\G{G}
\def\F{F}
\def\L{L}
\def\cL{\mathcal{L}} 
\def\Fu{F_{1,\alpha}} 

\def\Gu{\Gamma_1}   
\def\Gd{\Gamma_2}   
\def\Guab{\Gamma_1^{(a,b)}}
\def\Gutab{\Gamma_1^{(\tilde a,\tilde b)}}
\def\Gdab{\Gamma_2^{(a,b)}}

\def\rest{\mathcal{R}} 

\def\mab{\begin{pmatrix}a\\b\end{pmatrix}}
\def\mtab{\begin{pmatrix}\tilde{a}\\ \tilde{b}\end{pmatrix}}
\def\mabx{\begin{pmatrix}a(x)\\b(x)\end{pmatrix}}

\def\Sgs{u}
\def\Sgsa{{\Sgs_\alpha}} 

\def\dvar{e}

\def\vac{\Omega_f}  

\def\gr{\Phi_\alpha} 
\def\tet{\Theta}
\def\gs{\theta_{\mbox{\tiny GS}}}     
\def\ups{\Upsilon_\alpha} 

\def\bra{\langle}       
\def\Bra{\Big\langle}   
\def\ket{\rangle}       
\def\Ket{\Big\rangle}   

\def\Aa{A^-} 
\def\Ac{A^+} 

\def\Ba{B^-} 
\def\Bc{B^+} 

\def\Proj{\Pi}

\def\ttH{K}  
\def\tH{\tilde H} 

\newcommand{\gH}{{\mathfrak H}}  
\newcommand{\gF}{{\mathfrak F}}  

\def\C{{\mathbb{C}}} 
\def\N{{\mathbb{N}}} 
\def\R{{\mathbb{R}}} 

\def\cno{c_{\mathrm{n.o.}}} 

\def\1{{\mathbf{1}}}
\renewcommand\d{\mathrm{d}}
\renewcommand\Re{\mathrm{Re}\,}

\begin{document}

\title[Hydrogen Binding Energy in QED. The spin case.]
 {Quantitative Estimates on the Binding \\Energy for Hydrogen in Non-Relativistic\\ QED. II. The spin case.}

\author[J.-M. Barbaroux]{Jean-Marie Barbaroux}

\address{Unit\'e Mixte de Recherche
(UMR 6207) du CNRS et de l'Universit\'e d'Aix-Marseille et
de l'Universit\'e du Sud Toulon-Var.\\ Unit\'e affili\'ee \`a la
FRUMAM F\'ed\'eration de Recherche 2291 du CNRS.}
\email{Jean-Marie.Barbaroux@cpt.univ-mrs.fr}

\author[S. Vugalter]{Semjon Vugalter}
\address{Karlsruhe Institute of Technology,
Kaiserstra\ss e 89-93,\\
76133, Karlsruhe, Germany\\}
\email{Semjon.Wugalter@kit.edu}



\begin{abstract}

 The hydrogen binding energy in the Pauli-Fierz model with the spin Zeeman term  is determined up to the order $\alpha^3$, where $\alpha$ denotes the fine-structure constant.
\end{abstract}

\maketitle


\section{Introduction}

We continue the work started in \cite{BCVV2} by Th.Chen, V.Vougalter and the present authors,
with the study of the binding energy of a hydrogen-like atom in non-relativistic QED.
This time we consider the Pauli-Fierz model with the spin Zeeman term.

For an atom consisting of an electron interacting
with a static nucleus of charge $eZ$ described by the
Schr\"odinger-Coulomb Hamiltonian $-\Delta -\alpha Z /|x|$, the quantity
$$
 \inf\mathrm{spec}(-\Delta) - \inf\mathrm{spec}(-\Delta -\frac{\alpha Z }{|x|} ) \, = \,
 \frac{(Z \alpha)^2}{4}\ ,
$$
corresponds to the binding energy necessary to remove the electron
to spatial infinity.

The interaction of the electron with the quantized electromagnetic
field is accounted for by adding to $-\Delta -\alpha
Z/|x|$ the photon field energy operator $H_f$, and an operator
$I(\alpha)$ which describes the coupling of the electron to the
quantized electromagnetic field, yielding the so-called
Pauli-Fierz operator (see details in Section~\ref{model}).

In this case, the binding energy is given by
\begin{equation}\label{eq:binding-1}
 \Sigma_0 - \Sigma := \inf\mathrm{spec}\big(\, -\Delta + H_f + I(\alpha)\,\big)\
 - \ \inf\mathrm{spec}\big(\,-\Delta - \frac{\alpha Z}{|x|} + H_f + I(\alpha)\,\big)\
\end{equation}

The free infraparticle binds a larger quantity of low-energetic
photons than the confined particle and thus possesses a larger
effective mass. In order for the particle to leave the potential
well, an additional energetic effort is therefore necessitated
compared to the situation without coupling to the quantized
electromagnetic field.

It remains however a difficult task to determine the binding energy.
There are mainly two difficulties. The first is that the ground state energy
is not an isolated eigenvalue of the
Hamiltonian, and can not be determined with ordinary perturbation
theory. The second is due to the infamous infrared problem in quantum
electrodynamics whose origin is in the photon form factor in the quantized
electromagnetic vector potential occurring in the interaction term
$I(\alpha)$, that contains a critical frequency space singularity.

The systematic study of the Pauli-Fierz operator, in a more general case
involving more than one electron, was initiated by Bach,
Fr\"ohlich and Sigal \cite{Bachetal1995, Bachetal1999,
Bachetal1999bis}.

Later on, several rigorous results \cite{GLL, HVV, Chenetal2003, HS, CattoHainzl2004,
BCV, Hainzletal2005, Bachetal2006, BCVV1, BCVV2, BV3, HH, BV4}
have been obtained addressing both qualitative and quantitative estimates on
the binding energy and the ground state energies $\Sigma_0$ and $\Sigma$ occurring in \eqref{eq:binding-1}.

The case of a spinless particle attracted most of the attention since the additional spin-Zeeman term $\sqrt{\alpha}\sigma\!\cdot\!B(x)$ in the case of a particle with spin induces additional technical difficulties for quantitative estimates. There are three major problems, which substantially complicate the computation of the binding energy for the operator with the spin Zeeman term.

First, in the spinless case, the self-energy $\Sigma_0$, the Pauli-Fierz ground state energy $\Sigma$ and the Schr\"odinger ground state energy are of the same order $\alpha^2$. The contribution of the quantized radiation field into the binding energy in this model is of the order $\alpha^3.$ In contrast, in the model with the spin Zeeman term,  both $\Sigma_0$ and $\Sigma$ are of the order $\alpha$. At the same time the ground state energy of the corresponding Schr\"odinger operator and the contribution of the radiation field into the binding energy, which is the main subject of our interest, have the same orders as in the spinless case, namely $\alpha^2$ and $\alpha^3$ respectively.

Second, the expected  photon number for the ground state of the Pauli-Fierz operator in the spinless case is of order $\alpha^2$.  For the operator with spin Zeeman term, this quantity is of order $ \alpha$ (see section~\ref{section:S4.1}). Due to this difference, to find the infimum of the spectrum of the Pauli-Fierz operator for a particle with spin, we have to consider a much wider class of possible minimizer of the energy functional compared to the case of a spinless particle.

Finally, the field energy $H_f$  of the ground state of the Pauli-Fierz operator with spin has order $\alpha$ and not $\alpha^2$ as in the spinless case. Consequently the estimates of the energy on states orthogonal to the ground state of the Schr\"odinger operator (for detailed definition and results see \cite{BCVV2}), which played a crucial role in \cite{BCVV2}, are not valid any more for the Pauli-Fierz operator with spin.

In the work at hand we follow the same strategy as in \cite{BCVV2}. This strategy consists mainly in an iteration procedure based on variational estimates. At the same time due to the differences between the Pauli-Fierz operators with and without spin, which were mentioned above, the technical implementation of this strategy is very different from the one in \cite{BCVV2}. In particular we have to do several additional steps to compute the ground state energy $\Sigma$. As the first step we estimate $\Sigma$ up to the order $\alpha$ with the error of the order $\alpha^2$ (Recall that in \cite{BCVV2}, at the first iteration, the ground state energy was obtained up to the order $\alpha^3$ with an error of the order $\alpha^4$ ). Then we use the obtained information regarding the approximate ground state of the Pauli-Fierz operator with spin to improve estimate on $\Sigma$ up to the order $\alpha^2$. Only at the third step we are able to get the binding energy up to the order $\alpha^3$, which is the first order containing a contribution of the quantized radiation field.

The binding energy of a Hydrogen-like atom with nuclear charge $Z$ for a particle with spin was studied earlier in \cite{HS}. To avoid the difficulties mentioned above the authors of [18] made an important simplifying assumption that the product $\alpha Z$ remains constant as $\alpha$ tends to zero, which implies that their result is valid only for large values of the atomic number. The main achievement of our work is in the development of a technique which allows to work without this restriction.


The paper is organized as follows. In section 2 we give the main definitions and the statement of the result. In section 3 we construct a trial state to get the lower bound on the binding energy. Section 4 contains the statements of two auxiliary results: an estimate on the expected photon number and an energy estimate on states orthogonal to the ground state of the Schr\"odinger operator. Although the results are different from the spinless case, the proofs of these two estimates follow the same ideas and the same steps as the proofs of the corresponding estimates in \cite{BCVV2}. In section 5 we get the lower bound on the energy of the Pauli-Fierz operator up to the order $\alpha$ and show that the binding energy can not be greater then $C\alpha^2$. In section 6 we refine the photon number estimate using the same ideas as in \cite{BCVV2}. In section 7 we do the second iteration for the ground state energy of the Pauli-Fierz operator with spin. In section 8 we prove the main theorem. Finally, the Appendix contains the proofs of the statements given in section 4, short review of the results on the ground state of the translation invariant operator proved in \cite{BV4} and a large number of technical estimates.


\section{Model and main result}\label{model}

We study an electron, i.e., a spin $1/2$ particle, interacting with the quantized
electromagnetic field in the Coulomb gauge, and with the
electrostatic potential generated by a nucleus.

The Hilbert space accounting for the Schr\"odinger electron is given by
$\gH_{el}=L^2(\R^3)\otimes\C^2$. Here $\R^3$ is
the configuration space of the particle, while $\C^2$ accommodates its spin.

The Fock space of photon states is given by
 $$
   \gF \; = \; \bigoplus_{n \in \N} \gF_n ,
 $$
where the $0$-photon space is $\gF_0=\C$, and for $n\geq 1$ the $n$-photon space $\gF_n =
\bigotimes_s^n\left(L^2(\R^3)\otimes\C^2\right)$ is the symmetric
tensor product of $n$ copies of one-photon Hilbert space
$L^2(\R^3)\otimes\C^2$. The factor $\C^2$ accounts for the two
independent transversal polarizations of the photon.

On $\gF$, we
introduce creation and annihilation operators $a_\lambda^*(k)$,
$a_\lambda(k)$ satisfying the distributional commutation relations
 $$
  [ \, a_{\lambda}(k)  , \,  a^\ast_{\lambda'}(k') \, ]  =
  \delta_{\lambda, \lambda'} \, \delta (k-k')
  \; \;   ,
  \quad [ \, a_\lambda(k)  ,\,
   a_{\lambda'}(k') \, ]  =  [ \, a_\lambda^\ast(k)  ,\,
   a_{\lambda'}^\ast(k') \, ]  =  0 \, .
 $$
There exists a unique unit ray $\vac\in\gF$, the
Fock vacuum, which satisfies $a_\lambda(k) \, \vac=0$ for all
$k\in\R^3$ and $\lambda\in\{1,2\}$.

The Hilbert space of states of the system consisting of both the
electron and the radiation field is given by
$$
    \gH \; := \; \gH_{el} \, \otimes \, \gF .
$$
We use atomic units such that $\hbar = c = 1$, and where the mass of the
electron equals $m=1/2$. The electron charge is then given by
$e=\sqrt{\alpha}$, where the fine structure constant $\alpha$ has physical value about $1/137$ and will
here be considered as a small parameter.


The Pauli-Fierz Hamiltonian  for a charged particle with spin in the external electrostatic potential due to a pointwise static nucleus of atomic number $Z=1$,
coupled to the quantized electromagnetic radiation field is defined by
\begin{equation}\label{rpf}
\begin{split}
 \ : \left(- i\nabla_{x}\otimes I_f + \sqrt{\alpha}
A(x)\right)^2 :\ + \ \sqrt{\alpha}\,\sigma\!\cdot\! B(x)
+ V(x)\!\otimes\! I_f + I_{el}\!\otimes\! H_f\, .
\end{split}
\end{equation}

The operator that couples a particle to the quantized vector
potential is
\begin{equation}\nonumber
\begin{split}
 A(x) = & \sum_{\lambda = 1,2} \int_{\R^3}
  \frac{\zeta(|k|)}{2\pi|k|^{1/2}}
  \varepsilon_\lambda(k)\Big[ e^{ikx} \otimes a_\lambda(k)  +
  e^{-ikx} \otimes a_\lambda^\ast
  (k) \Big] \d k \\
  =: & \ \Aa(x) + \Ac(x)\, ,
\end{split}
\end{equation}
where ${\rm div}A =0$ by the Coulomb gauge condition.

The vectors $\varepsilon_\lambda(k)\in\R^3$ ($\lambda=1,\, 2$), are the two
orthonormal polarization vectors perpendicular to $k$,
\begin{equation}\nonumber
  \varepsilon_1(k) = \frac{(k_2, -k_1, 0)}{\sqrt{k_1^2 + k_2^2}}\qquad
 {\rm and} \qquad
   \varepsilon_2(k) = \frac{k}{|k|}\wedge \varepsilon_1(k).
\end{equation}

The function $\zeta(|k|)$ implements an {\it ultraviolet cutoff}, independent of $\alpha$, on the photon momentum $k$. We assume $\zeta$ to be of class $C^1$ and to have a compact support.

The symbol $: ... :$ denotes normal ordering and is applied to the operator $A(x){}^2$. It corresponds here to the subtraction of a constant operator $\cno\, \alpha$, with $\cno = [\Aa(x),\, \Ac(x)] = ({2}/{\pi})
\int_0^\infty r|\zeta(r)|^2 {\rm d}r$.

The operator that couples a particle to the magnetic field $B =
{\rm curl}A$ is given by
\begin{equation}\nonumber
\begin{split}
B (x)  = &\displaystyle\sum_{\lambda=1,2}\! \int_{\R^3}\!
\frac{\zeta(|k|)}{2\pi|k|^{1/2}} (k\times i\varepsilon_\lambda(k))
\Big[ e^{ikx}\otimes a_\lambda(k)  - e^{-ikx}\otimes
a_\lambda^\ast(k)\Big] \d k
\\
 =: &
\ \Ba(x) + \Bc(x).
\end{split}
\end{equation}
In Equation~\eqref{rpf}, $\sigma = (\sigma_1, \sigma_2, \sigma_3)$
is the 3-component vector of Pauli matrices
\begin{eqnarray*}
    \sigma_1 =
    \begin{pmatrix}
      0 & 1 \\
      1 & 0
    \end{pmatrix}\, , \ \
    \sigma_2 =
    \begin{pmatrix}
      0 & -i \\
      i & 0
    \end{pmatrix}\, , \ \
    \sigma_3 =
    \begin{pmatrix}
      1 & 0 \\
      0 & -1
    \end{pmatrix}\, .
\end{eqnarray*}

The Coulomb potential is the operator of multiplication by
$$
 V(x) = -\frac{\alpha}{|x|}\, .
$$

The photon field energy operator $H_f$ is given by
\begin{equation}\nonumber
H_f = \sum_{\lambda= 1,2} \int_{\R^3} |k| a_\lambda^\ast (k)
a_\lambda (k) \, \d k.
\end{equation}

In the sequel, instead of the operator \eqref{rpf}, we shall proceed to a change of variables, and study the unitarily equivalent
Hamiltonian
\begin{equation}\label{eq-H-Utrsf-1}
  H = U\Big( : \left(i\nabla_{\! x}\otimes I_f -
 \sqrt{\alpha} A(x)\right)^2\! : +  \sqrt{\alpha} \sigma\!\cdot\! B(x)\ +\ V(x) \otimes I_f\
 + \ I_{el}\otimes H_f  \Big)U^*\ ,
\end{equation}
where the unitary transform $U$ is defined by
\begin{equation}\label{def:unitary-transform}
  U = \mathrm{e}^{i P_f.x} ,
\end{equation}
and
\begin{equation}\nonumber
 P_f = \sum_{\lambda =1,2} \int  k \,
 a^\ast_\lambda(k) a_\lambda(k) \d k
\end{equation}
is the photon momentum operator.
We have
 $$
  U i\nabla_x U^* = i\nabla_x + P_f \, , \quad U A(x)U^* = A(0)\, , \quad\mbox{and}\quad
  U B(x)U^* = B(0)\ .
 $$
In addition, the Coulomb operator $V$, the photon field energy
$H_f$, and the photon momentum $P_f$ remain unchanged under the
action of $U$. Therefore, in this new system of variables, and omitting by abuse of notations the operators
$I_{el}$ and $I_f$, the Hamiltonian \eqref{eq-H-Utrsf-1} reads
\begin{equation}\label{ast-ast}
\begin{split}
  H = \ : \left( (i\nabla_x -  P_f)
 - \sqrt{\alpha} A(0)\right)^2 :
 \ +\   \sqrt{\alpha} \sigma\!\cdot\! B(0) \  -\  \frac{\alpha}{|x|} \ +\   H_f\, ,
\end{split}
\end{equation}
where $:...:$ denotes again the normal ordering.



The Hamiltonian for a free electron with spin coupled to the quantized radiation
field is given by the self-energy operator $T$,
\begin{equation}\nonumber
\begin{split}
 T & = H + \frac{\alpha}{|x|}
   \ =  \  :\left( (i\nabla_x -  P_f)
 - \sqrt{\alpha} A(0)\right)^2 : \
 + \ \sqrt{\alpha} \sigma\!\cdot\! B(0) \ +\  H_f\, ,
\end{split}
\end{equation}
where we omit again the operators $I_{el}$ and $I_f$.

This system is translationally invariant, that is,
$T$ commutes with the operator of total momentum
 $$
 P_{tot} = p_{el} + P_f ,
 $$
where $p_{el}$ and $P_f$ denote respectively the electron and the
photon momentum operators.

Therefore, for fixed value $p\in\R^3$ of the total momentum, the
restriction of $T$ to the fibre space $\C^2 \otimes\mathfrak{F}$ is
given by (see e.g. \cite{Chen2008, BV4})
\begin{equation}\label{def:T(P)}
  T(p) = \ \ :(p - P_f  - \sqrt{\alpha} A(0))^2: \ +\ \sqrt{\alpha}\sigma\!\cdot\!B(0)  + H_f\, .
\end{equation}
Henceforth,
we will write
 $$
 A^\pm:= A^\pm(0)\quad\mbox{and}\quad B^\pm:=B^\pm(0) .
 $$

The ground state energies of $T$ and $H$ are respectively denoted by
 $$
  \Sigma_0 :=\inf\mathrm{spec}(T) \quad\mbox{and}\quad\Sigma : = \inf\mathrm{spec}(H)\, ,
 $$
and the binding energy is defined by
$$
 \Sigma_0\ -\ \Sigma\, .
$$

It is proven in \cite{BCFS, Chen2008} that
 $$
 \Sigma_0 = \inf\mathrm{spec} (T(0))\, ,
 \mbox{ and $\Sigma_0$ is an eigenvalue of the operator $T(0)$}\, .
 $$

%

Our main result is the following,
\begin{thm}\label{thm:main}
The binding energy fulfills the following inequality
\begin{equation}\label{eq:main1}
 \begin{split}
   \Sigma_0 -\Sigma =  \frac14\alpha^2 \ +\  \left(\dvar^{(1)}
   +  {\dvar}_{\mbox{\scriptsize{Zeeman}}}^{(1)}\right) \alpha^3 \  + \  \mathcal{O}(\alpha^{3+\frac16})\, ,
 \end{split}
\end{equation}
where
 $$
  \dvar^{(1)} = \frac2{3\pi} \int_0^\infty
  \frac{\zeta^2(t)}{1+t} \, \mathrm{d}t
  \quad\mbox{and}\quad
  {\dvar}_{\mbox{\scriptsize{Zeeman}}}^{(1)} = \frac{2}{3\pi} \int_0^\infty
  \frac{t^2\, \zeta^2(t)}{(1+t)^3} \, \mathrm{d}t \, .
 $$
\end{thm}
\begin{rem}
i) Recall that for the spinless Pauli-Fierz operator  it is known (see \cite{BCVV2} and references therein) that the binding energy is
$$
 \Sigma_0 - \Sigma = \frac14\alpha^2 + \dvar^{(1)} \alpha^3 +\mathcal{O}(\alpha^4)\, .
$$

The above Theorem~\ref{thm:main} thus shows that the spin-Zeeman term yields an additional contribution to the term of order $\alpha^3$.

ii) In \cite{BCVV2}, a factor $1/3$ was missing in the definition of $\dvar^{(1)}$.

iii) Since the self-energy $\Sigma_0$ for the operator $T(0)$ given by \eqref{def:T(P)} is of the order $\alpha$ and not of the order $\alpha^2$ as in the spinless case, it turns out that it would to be a much harder task to compute higher order terms of the binding energy, and in particular to recover a $\log\alpha$ divergent term in the expansion, as in \cite{BCVV2}.
\end{rem}

In the remainder, we will need the following notations. For $n\in\N$, let $\Proj_{n}$
be the orthogonal projection onto $\gH_{el}\otimes\gF_n$, $\Proj_{\geq n}$ the orthogonal projection onto $\gH_{el}\otimes\left(\bigoplus_{k\geq n}\gF_k\right)$, and $\Proj_{\leq n} = 1 -\Proj_{\geq n+1}$.

On $\gH_{el}\otimes\gF$, we define the positive bilinear form
\begin{equation}\label{bilinear}
 \langle v,\, w\rangle_* := \langle v,\, (H_f + P_f^2) w\rangle_{\gH_{el}\otimes\gF}\, ,
\end{equation}
and its associated semi-norm $\|v\|_* = \langle v,\, v\rangle_*^{1/2}$.

We also define the positive bilinear form on $\gF\times\gF$
$$
 \langle v,\, w\rangle_{*,\gF} := \langle v,\, (H_f + P_f^2) w\rangle_\gF\, ,
$$
and its associated semi-norm $\|v\|_{*,\gF} = \langle v,\, v\rangle_{*,\gF}^{1/2}$.

In the sequel we shall omit the index $\gF$ in $\|.\|_{*,\gF}$, $\bra.,\, .\ket_{\gH_{el}\otimes\gF}$ and $\bra.,\, .\ket_{\gF}$ if no possible confusion can occur.

\section{Lower bound on the binding energy}\label{S3}

In this section, we shall prove the inequality
\begin{equation}\label{eq:lower-bound-main}
 \begin{split}
   \Sigma_0 -\Sigma \geq  \frac14\alpha^2 \ +\  \left(\dvar^{(1)}
   +  {\dvar}_{\mbox{\scriptsize{Zeeman}}}^{(1)}\right) \alpha^3 \  + \  \mathcal{O}(\alpha^{4}\log\alpha)\, ,
 \end{split}
\end{equation}

To prove this inequality, we construct a trial function $\pst$ such that holds $\bra\pst,\, H\pst\ket / \|\pst\|^2 \leq \Sigma_0 -\alpha^2/4 -  (\dvar^{(1)}+ \ {\dvar}_{\mbox{\scriptsize{Zeeman}}}^{(1)}) \alpha^3 +\mathcal{O}(\alpha^4|\log\alpha|)$.

Let
$$
 P:= i\nabla_x\, .
$$

We denote by $\gs$ the ground state of $T(0)$ with the normalization condition $\Proj_0 \gs = \vac \zo$.
(see\eqref{eq:def-phi0-1}-\eqref{eq:def-phi0-2} in Theorem~\ref{thm:bv4} for detailed
definiton and properties of $\gs$), and let
 $$
  \tet := \gs \, \Sgsa\, ,
 $$
where $\Sgsa$ is the normalized ground state of the Schr\"odinger operator $-\Delta -\alpha/|x|$
\begin{equation}\label{eq:def-u-alpha}
  (-\Delta - \frac{\alpha}{|x|}) \Sgsa = -\frac{\alpha^2}{4} \Sgsa\, .
\end{equation}

For $\Gu$ defined as in \eqref{eq:def-gamma1} by
\begin{equation}\label{eq:3.0}
 \Gu := \ -(H_f + P_f^2)^{-1} \sigma.B^+ \vac\zo \, ,
\end{equation}
we set
\begin{equation}\label{def:gr}
 \gr := \ 2\, P . P_f\, (H_f + P_f^2)^{-1} \Gu\, \Sgsa \, ,
\end{equation}
and
\begin{equation}\label{def:ups}
 \ups := \  2\,  \chi_{(\alpha,\, \infty)}(H_f)(H_f + P_f^2)^{-1} \, P. \Ac  \vac \Sgsa \zo \, ,
\end{equation}
where  $\chi_{(\alpha,\, \infty)}(H_f)$ is an infrared cutoff and $\chi_{(\alpha,\, \infty)}$ is the
characteristic function of $(\alpha,\, \infty)$.

Let us define the following trial function
\begin{equation}\nonumber
 \pst = \tet +  \alpha^\frac12 \gr + \alpha^\frac12 \ups\, .
\end{equation}
In comparison with the trial function used in the spinless case \cite{BCVV2} to recover the estimate up to the order $\alpha^3$,  with error $\alpha^4$, the function $\pst$ differs in two points. First we pick now the state $\gs$ as the ground state of the translation invariant operator $T(0)$ with spin. Second, we have an additional vector $\gr$ which is responsible for the $ {\dvar}_{\mbox{\scriptsize{Zeeman}}}^{(1)}\, \alpha^3$ term in \eqref{eq:main1}.

From the definition of $H$, expanding \eqref{ast-ast} and taking into account the normal ordering, we obtain
\begin{equation}\label{expand-H}
\begin{split}
 H =
 & (-\Delta -\frac{\alpha}{|x|}) + (H_f + P_f^2) - 2\,P. P_f
 - 4\alpha^\frac12 \Re P . \Aa  \\
 & + 4 \alpha^\frac12 P_f . \Aa
 + 2\alpha\Ac\Aa + 2\alpha (\Re\Aa)^2 + 2\alpha^\frac12\Re \sigma. \Ba\, .
\end{split}
\end{equation}
We shall use this expression to estimate all terms occurring in
\begin{equation}\label{eq:aaa}
\begin{split}
& \bra \pst,\, H \, \pst\ket \\
& =  \left\bra \tet +  \alpha^\frac12 \gr + \alpha^\frac12\ups\, , \ H\ \, \left(\tet +  \alpha^\frac12 \gr + \alpha^\frac12 \ups\right)\right\ket\, .
\end{split}
\end{equation}

\underline{Step 1.} We first compute the three direct terms $\bra \tet, H\tet \ket$, $\bra \alpha^\frac12 \gr, H \alpha^\frac12 \gr \ket$ and $\bra \alpha^\frac12 \ups, H \alpha^\frac12 \ups \ket$.

Since $\gs$ is a ground state vector of $T(0)$, and using orthogonality between the components of $P \Sgsa$ and $\Sgsa$, we have
\begin{equation}\label{eq:trial1}
\begin{split}
\bra \tet,\, H\tet\ket & = \bra \Sgsa \gs,\, H\Sgsa\gs\ket \\ &
= \|\Sgsa\|^2 \bra \gs,\, T(0) \gs\ket + \|\gs\|^2 \bra\Sgsa,\, (-\Delta -\frac{\alpha}{|x|})\Sgsa\ket\\
& = (\Sigma_0 -\frac{\alpha^2}{4}) \|\tet\|^2 \, .
\end{split}
\end{equation}

Using that for $i,\, j,\, k\in\{1, \, 2\, , 3\}$, $\partial \Sgsa/\partial x_i$ and $\partial^2 \Sgsa/(\partial x_j\partial x_k)$ are orthogonal, and the fact that non particle conserving operators have mean value zero in the state $\gr$, yields
\begin{equation}\label{eq:trial2}
\begin{split}
 & \bra \alpha^\frac12\gr, H \alpha^\frac12\gr \ket  \\
 & = \left\bra \alpha^\frac12\gr, \left( -\Delta-\frac{\alpha}{|x|}\ +\ (H_f+P_f^2)\ +\ \alpha\Aa\!\cdot\!\Ac\right) \alpha^\frac12\gr\right\ket \\
 & = \alpha\|\gr\|_*^2\ +\ \mathcal{O}(\alpha^4)\, ,
\end{split}
\end{equation}
where the last inequality holds since $\|P\Sgsa\| =\mathcal{O}(\alpha)$, $\| (-\Delta-\frac{\alpha}{|x|}) P \Sgsa\| =\mathcal{O}(\alpha^3)$.

Using the same arguments as above, and the fact that
$$
 \|(H_f+P_f^2)^{-1}
 \chi_{(\alpha,\,\infty)}(H_f) ({\Ac})_{j} \vac\zo\|
 = \mathcal{O}(|\log\alpha|^\frac12)\, ,
$$
the last direct term can be estimated as
\begin{equation}\label{eq:trial3}
\begin{split}
    & \bra \alpha^\frac12 \ups,\, H \alpha^\frac12 \ups\ket  \\
    & = \left\bra \alpha^\frac12\ups, \left( -\Delta-\frac{\alpha}{|x|}
    + (H_f+P_f^2) + \alpha\Aa\!\cdot\!\Ac\right) \alpha^\frac12\ups \right\ket \\
    & = \mathcal{O}(\alpha^5\log\alpha) +
    \alpha \left\bra \ups,  (H_f+P_f^2)  \ups \right\ket + \mathcal{O}(\alpha^4)\, \\
    & =
    \alpha \| \ups\|_*^2\ +\ \mathcal{O}(\alpha^4)\, .
\end{split}
\end{equation}

\underline{Step 2.} We compute in \eqref{eq:aaa} the cross terms with $\gr$ and $\ups$. Using as above the estimates $\|(H_f+P_f^2)^{-1} \chi_{(\alpha,\,\infty)}(H_f) ({\Ac})_{j} \vac\zo\|
= \mathcal{O}(|\log\alpha|^\frac12)$, $\|P\Sgsa\| =\mathcal{O}(\alpha)$, and $\| (-\Delta-\frac{\alpha}{|x|}) P \Sgsa\|
=\mathcal{O}(\alpha^3)$ yields
\begin{equation}\label{eq:cross-1.1}
   \bra \alpha^\frac12\gr, (-\Delta -\frac{\alpha}{|x|}) \alpha^\frac12\ups\ket
 + \bra \alpha^\frac12\ups, (-\Delta -\frac{\alpha}{|x|}) \alpha^\frac12\gr\ket
 = \mathcal{O}(\alpha^5|\log\alpha|^\frac12)\, .
\end{equation}

Due to Lemma~\ref{lem:appendix-1} (see Appendix~\ref{appendix2}) holds
\begin{equation}\nonumber
\begin{split}
&  \bra \alpha^\frac12\gr, (H_f+ P_f^2) \alpha^\frac12\ups\ket
 + \bra  \alpha^\frac12\ups, (H_f+ P_f^2)  \alpha^\frac12\gr\ket = 0\, .
\end{split}
\end{equation}

Furthermore, $\|\Aa \gr\| \leq c \|\gr\| = \mathcal{O}(\alpha)$ and
$\|\Aa \ups\| \leq c \|H_f^\frac12 \ups\| = \mathcal{O}(\alpha)$ implies
\begin{equation}\label{eq:cross-1.2}
\begin{split}
&  \bra \alpha^\frac12\gr,\, 2\alpha\Ac . \Aa \alpha^\frac12\ups\ket
 + \bra \alpha^\frac12 \ups,\, 2\alpha\Ac . \Aa \alpha^\frac12\gr\ket \\
& \leq 4 \alpha^2 \|\Aa \gr\|\, \|\Aa \ \ups\| =\mathcal{O}(\alpha^4)\, .
\end{split}
\end{equation}

In addition, due either to the symmetry of $\Sgsa$ or the occurrence of non-particle conserving terms,
all other cross terms with $\gr$ and $\ups$
in \eqref{eq:aaa} are equal to zero. Therefore, collecting \eqref{eq:cross-1.1}-\eqref{eq:cross-1.2} we get
\begin{equation}\label{eq:trial4}
   \bra \alpha^\frac12\gr,\, H\, \alpha^\frac12\ups\ket
 + \bra\alpha^\frac12\ups,\, H\, \alpha^\frac12\gr\ket
 = \mathcal{O}(\alpha^4)\, .
\end{equation}


\underline{Step 3.}
We estimate in \eqref{eq:aaa} the cross terms involving $\gr$ and $\tet=\gs \Sgsa$. Only terms coming from
$-2\Re P . P_f$ and $-4\Re \alpha^\frac12 P.\Aa$ can a priori contribute since other terms are
zero due to the symmetry of $\Sgsa$.

The contribution of $-2\Re P . P_f$ is
$$
 -2 \Re \left\bra \Proj_1\tet,\, P.P_f \alpha^\frac12\gr \right\ket - 2\Re \left\bra \alpha^\frac12\gr ,\, P.P_f \Proj_1\tet \right\ket\, .
$$

We write $\Proj_1\tet$ as $(\alpha^\frac12 \gamma_1\Gu + \Proj_1 R)\Sgsa$, where $R$ and $\gamma_1$ are defined by \eqref{eq:def-phi0-1} and \eqref{eq:def-phi0-2} in Theorem~\ref{thm:bv4}. This implies
\begin{equation}\label{eq:cross-2.1}
\begin{split}
& -2 \Re \left\bra \Proj_1\tet,\, P.P_f \alpha^\frac12\gr \right\ket - 2\Re \left\bra \alpha^\frac12\gr ,\, P.P_f \Proj_1\tet \right\ket\\
& = -4 \Re \left\bra (\alpha^\frac12 \gamma_1\Gu + \Proj_1 R) \Sgsa,\, P.P_f 2\alpha^\frac12
(H_f+P_f^2)^{-1} P. P_f \Gu\Sgsa\right\ket \\
& = - 8 \alpha \Re \gamma_1 \bra P. P_f \Gu\Sgsa,\,
(H_f+P_f^2)^{-1}  P. P_f\Gamma_1 \Sgsa\ket \\
& - 8\alpha^\frac12\Re \bra P\Sgsa . P_f \Proj_1 R, \, P\Sgsa . (H_f+P_f^2)^{-1} P_f\Gu\ket \\
& \geq - 2  \alpha \Re \gamma_1 \| \gr\|_*^2 - c \alpha \|P\Sgsa\|^2 \|P_f \Proj_1\Gu\| \\
& = - 2\alpha \| \gr\|_*^2 + \mathcal{O}(\alpha^4)\, ,
\end{split}
\end{equation}
where in the last equality, we used $\|P\Sgsa\| = \mathcal{O}(\alpha)$ and from \eqref{eq:est-1} of Theorem~\ref{thm:bv4} that $|\gamma_1 -1| = \mathcal{O}(\alpha)$ and
$\| \Proj_1 R\|_* \leq \| R\|_* = \mathcal{O}(\alpha^\frac32)$.

The contribution of $-4\alpha^\frac12\Re P.\Aa$ is
\begin{equation}\label{eq:cross-2.2}
\begin{split}
&  -4 \alpha^\frac12\Re \left\bra \Proj_0 \tet,\,  P.\Aa \alpha^\frac12\gr \right\ket
   -4 \alpha^\frac12\Re \left\bra \alpha^\frac12\gr,\, P.\Aa \Proj_2 \tet \right\ket\\
& = -4 \alpha^\frac12 \Re \left\bra P.\Ac \Proj_0 \tet,\, \alpha^\frac12\gr \right\ket \\
&  - 4 \alpha^\frac12\Re \left\bra 2 \alpha^\frac12P.P_f (H_f+P_f^2)^{-1}\Gu\Sgsa,\, P.\Aa(\alpha \gamma_2 \Gd + \Proj_2 R)\Sgsa \right\ket\\
& = -4 \alpha^\frac12 \Re \left\bra P.\Ac \vac\Sgsa\su,\, \alpha^\frac12\gr \right\ket \\
&    - 8 \alpha^2 \Re\bar\gamma_2 \bra P\Sgsa . P_f (H_f+P_f^2)^{-1}\Gu,\, P\Sgsa . \Aa\Gd\ket\\
&   - 8\alpha  \Re \bra P\Sgsa . P_f (H_f+P_f^2)^{-1}\Gu,\, P\Sgsa . \Aa \Proj_2 R\ket   = \mathcal{O}(\alpha^4)\, ,
\end{split}
\end{equation}
where in the second inequality we used $\| P\Sgsa\|=\mathcal{O}(\alpha)$,
 $\|\Aa \Proj_2 R\| \leq c \|\Proj_2 R \|_* =
\mathcal{O}(\alpha^\frac32)$ from \eqref{eq:est-1} in Theorem~\ref{thm:bv4}, and $ \bra P.\Ac \vac\Sgsa\zo,\, \gr\ket = 0$ from Lemma~\ref{lem:appendix-1}.

The estimates \eqref{eq:cross-2.1} and \eqref{eq:cross-2.2} yields
\begin{equation}\label{eq:trial5}
 \bra \alpha^\frac12 \gr,\, H\tet \ket + \bra \tet,\, H\alpha^\frac12\gr\ket
 = -2\alpha\|\gr\|_*^2 + \mathcal{O}(\alpha^4)\, .
\end{equation}


\underline{Step 4.}
We next estimate in \eqref{eq:aaa} the cross terms involving $\ups$ and $\tet=\gs \Sgsa$. As in the previous step, only terms coming from
$-2\Re P . P_f$ and $-4\Re \alpha^\frac12 P.\Aa$ can a priori contribute since other terms are
zero due to the symmetry of $\Sgsa$.

The contribution of $-2\Re P . P_f$ is
\begin{equation}\label{eq:cross-3.1}
\begin{split}
& -2\Re \bra \Proj_1 \tet,\, P.P_f \alpha^\frac12\ups\ket -2\Re \bra \alpha^\frac12\ups,\, P.P_f \tet\ket  \\
& = -4\Re \bra (\alpha^\frac12\gamma_1 \Gu + \Proj_1 R)\Sgsa ,\, P.P_f \alpha^\frac12\ups\ket \\
& = - 2 \alpha \Re \gamma_1 \bra 2\, P.P_f (H_f + P_f^2)^{-1}\Gu\Sgsa,\, \ups\ket_*
- 4\alpha^\frac12\Re \bra P\Sgsa . P_f \Proj_1 R,\, \ups\ket \\
&  =\mathcal{O}(\alpha^4)
\end{split}
\end{equation}
where we used $ \bra 2\, P.P_f (H_f + P_f^2)^{-1}\Gu\Sgsa,\, \ups\ket_* =
\bra\gr,\, \ups\ket_* = 0$ due to Lemma~\ref{lem:appendix-1}, and $\|P\Sgsa\| =
\mathcal{O}(\alpha)$, $\|\ups\| = \mathcal{O}(\alpha)$ and $\| P_f \Proj_1 R\|\leq \|R\|_* = \mathcal{O}(\alpha^\frac32)$
(Theorem~\ref{thm:bv4}).

The contribution of $-4\alpha^\frac12 \Re P.\Aa$ is
\begin{equation}\label{eq:cross-3.2}
\begin{split}
 & - 4 \alpha^\frac12 \Re\bra \Proj_0 \tet,\, P.\Aa \alpha^\frac12 \ups\ket
   - 4 \alpha^\frac12 \Re \bra \alpha^\frac12 \ups,\, P.\Aa \Proj_2 \tet\ket \\
 & = - 2\alpha \Re \bra 2 P.\Ac \vac\Sgsa\zo,\,  \ups\ket - 4\alpha\Re\bra \ups,\, P.\Aa (\alpha\gamma_2\Gd + \Proj_2 R)\Sgsa\ket \\
 & = -2\alpha \|\ups\|_*^2 + \mathcal{O}(\alpha^4)\, ,
\end{split}
\end{equation}
where we applied in the last equality $\|P\Sgsa\| = \mathcal{O}(\alpha)$,  $\|\ups\|= \mathcal{O}(\alpha)$ and
$\|\Aa \Proj_2 R\| \leq c \|R\|_* = \mathcal{O}(\alpha^\frac32)$.

Equations \eqref{eq:cross-3.1} and \eqref{eq:cross-3.2} implies
\begin{equation}\label{eq:trial6}
 \bra \alpha^\frac12 \ups,\, H\, \tet \ket + \bra \tet,\, H \, \alpha^\frac12 \ups\ket =
 -2 \alpha \|\ups\|_*^2 + \mathcal{O}(\alpha^4)\, .
\end{equation}

\underline{Step 5.}
Collecting all above estimates \eqref{eq:trial1}, \eqref{eq:trial2}, \eqref{eq:trial3}, \eqref{eq:trial4},
\eqref{eq:trial5} and \eqref{eq:trial6} yields
\begin{equation}\nonumber
 \bra \pst,\, H\pst \ket = (\Sigma_0 -\frac{\alpha^2}{4}) \| \tet \|^2
  - \alpha \|\gr\|_*^2  -\alpha \|\ups\|_*^2
  + \mathcal{O}(\alpha^4\log\alpha)
\end{equation}

To conclude the proof, we need to normalize the above expression. First note that $\bra\tet,\, \alpha^\frac12(\gr+\ups)\ket = 0$ due to orthogonality of $\Sgsa$ and $\partial\Sgsa/\partial x_j$ ($j=1,\, 2,\, 3$). Therefore
\begin{equation}\nonumber
\begin{split}
 \|\pst\|^2
     & = \|\tet\|^2 + \alpha\|\gr+\ups\|^2 \\
     & = \|\tet\|^2 + \mathcal{O}(\alpha^3 |\log\alpha|)\, ,
\end{split}
\end{equation}
since $\|\gr + \ups\| = \mathcal{O}(\alpha|\log\alpha|^\frac12)$.

In addition, we have $\|\tet\|^2 = 1 + \mathcal{O}(\alpha)$ (see Theorem~\ref{thm:bv4}). Therefore, with
the estimates $\Sigma_0 = \mathcal{O}(\alpha)$, $\|\gr\|_* = \mathcal{O}(\alpha)$, and
$\|\ups\|_* = \mathcal{O}(\alpha)$, we obtain
$$
 \frac{\bra \pst,\, \, H \, \pst\ket}{\|\pst\|^2} = (\Sigma_0 -\frac{\alpha^2}{4}) - \alpha \|\gr\|_*^2
 - \alpha \|\ups\|_*^2 + \mathcal{O}(\alpha^4 |\log\alpha|)\, .
$$
To conclude the proof, it suffices to note that $\Sigma \leq  {\bra \pst,\, \, H \, \pst\ket}/{\|\pst\|^2}$ and
to replace $\|\gr\|_*$ and $\|\ups\|_*$ by their expressions in Lemma~\ref{lem:appendix-2}.


\section{Two auxiliary estimates}

In this section, following the strategy of \cite{BCVV2}, we derive two auxiliary estimates.
The first one is a photon number bound for the ground states
of $H$. The second one is an estimate on the energy of the states orthogonal to the ground state of the Schr\"odinger operator.

\subsection{Photon number bound}\label{section:S4.1}
The photon number bound is obtained by arguments similar to the one obtained
in the case of a spinless electron (see \cite{BCVV2}).
However, in the present case, due to the spin-Zeeman
term $\sqrt{\alpha}\sigma.B(x)$, it happens that the photon number estimate is one order larger
in powers of the fine structure constant, yielding a bound proportional to
$\alpha$ instead of a bound proportional to $\alpha^2|\log\alpha|$ as in \cite{BCVV2}.
As it can be seen from the estimates on the ground states on the translation invariant
operator found in \cite{BV4}, this estimate can not be improved.

\begin{prop}\label{Nf-exp-lemma-1}
Let
\begin{equation}\label{eq:def-op-K}
    \ttH \; = \; {(i\nabla-\sqrt\alpha A(x))^2}\, +\sqrt{\alpha}\sigma\cdot B(x)
    \, + \, H_f \, - \, \frac{\alpha}{|x|}\; ,
\end{equation}
be the Pauli-Fierz operator defined without normal ordering. Let $\pgs   \in \H$ be a
ground state of $K$,
\begin{equation}\nonumber
     \ttH \, \pgs   \; = \; E \, \pgs   \;,
\end{equation}
normalized by
\begin{equation}\nonumber
    \| \, \pgs   \, \| \; = \; 1 \;.
\end{equation}
Let
\begin{equation}\nonumber
    N_f \; := \; \sum_{\lambda=1,2}\int
    \, a_\lambda^*(k) \, a_\lambda(k) \, \d k
\end{equation}
denote the photon number operator. Then, there exists a constant
$c$ independent of $\alpha$,  such that for any sufficiently small
$\alpha>0$, the estimate
\begin{equation}\label{eq:photon-number-estimate-1}
     \bra\pgs   \, , \, N_f \,\pgs  \ket \; \leq \; c \, \alpha
\end{equation}
is satisfied.
\end{prop}

The proof of this proposition is given in Appendix~\ref{appendix-photon-number-bound}.

\subsection{A priori estimates for states orthogonal to $u_\alpha$}

As in \cite{BCVV2}, we shall need a specific notion of orthogonality to the ground state of the Schr\"odinger operator, defined as follows.
\begin{defn}\label{def:def-orthogonality}
We say that a state $\varphi$ in $\gH$ is orthogonal to the ground state $u_\alpha$ of the Schr\"odinger operator $(-\Delta-\frac{\alpha}{|x|})$ if for all $n\geq 0$,
for all $j\in\{1,\ldots ,n\}$, for all $\lambda_j\in\{1,2\}$,
for a.e.  $k=(k_1,\ldots,k_n)\in\R^{3n}$, for all $\mtab\in\C^2$
\begin{equation}\label{eq:def-orthogonality}
\begin{split}
 & \quad \bra \Proj_n \varphi(k_1,\lambda_1;\, k_2,\lambda_2;\cdots;k_n,\lambda_n;\, .\, ) ,\ u_\alpha(\,.\,)\mtab \ket_{L^2(\R^3, \C^2)} = 0\, .
\end{split}
\end{equation}
\end{defn}

For states orthogonal to the ground state of the Schr\"odinger-Coulomb operator, we will show that the energy is larger than the sum of Coulomb energy $-\alpha^2/4$ and self-energy $\Sigma_0$. Even more, due to this orthogonality, even if we subtract from this energy term a part of the Schr\"odinger-Coulomb energy, we can prove that this property remains correct. However, unlike the case of spinless electron (\cite{BCVV2}), we can not subtract a part of the field energy, since in the spin case, the field energy of the ground state is of the order $\alpha$ whereas the difference between the first and the second level for the Schr\"odinger-Coulomb operator is of order $\alpha^2$. This difference between the present case and the spinless case \cite{BCVV2} together with the weaker photon number bound, force us to add two more iterations to estimate the radiation field corrections to the binding energy.

\begin{prop}\label{prop:orthogonality}
Assume that $\varphi\in\gH$ fulfils \eqref{eq:def-orthogonality}

Then there exists $\alpha_0$ such that for all $0<\alpha<\alpha_0$ we have
\begin{equation}\label{eq:orth-u-alpha-0}
 \bra \varphi,\, H\varphi\ket \geq (\Sigma_0 - \frac{\alpha^2}{4}) \|\varphi\|^2 + \frac{3}{32}\alpha^2 \| \varphi \|^2\, .
\end{equation}
\end{prop}

The proof of this result is postponed to Appendix~\ref{appendix-proof-state-orthogonal-u-alpha}


\section{First estimates on the binding energy and the ground states of Hydrogen atom}\label{S6}


Using the photon number bound (Proposition~\ref{Nf-exp-lemma-1}), we prove that although the self-energy and the ground state energy are of the order of $\alpha$ (see respectively \cite{BV4} and section~\ref{S3}), the binding energy is of the order $\alpha^2$. This result is stated in Proposition~\ref{proposition:4}. In addition, we obtain some estimates on the approximate ground states which will be used in Sections~\ref{S-new-7} and \ref{S4} for further improvement of the binding energy estimate.


Given a ground state $\pgs$ of $H$ such that $\|\Proj_0 \pgs\|=1$, we decompose it into a part parallel to $u_\alpha$ and a part orthogonal to $u_\alpha$, where $u_\alpha$ is defined in \eqref{eq:def-u-alpha} as the ground state of the Schr\"odinger operator $-\Delta -\alpha/|x|$  . Namely, we define $\phi\in\C^2\otimes\gF$ and $\G\in\gH$ by
\begin{equation}\label{*-0}
 \pgs = u_\alpha\phi + \G \, ,
\end{equation}
with $\G\in\gH$ orthogonal to the ground state $u_\alpha$  in the sense of \eqref{eq:def-orthogonality} of Definition~\ref{def:def-orthogonality}.


Next we define a splitting for the state $\phi$.

Given $\mtab\in\C^2$, let
\begin{equation}\label{def:Guab}
\begin{split}
 \Gutab:= -(H_f+P_f^2)^{-1}\sigma.\Bc \vac\mtab\, .
\end{split}
\end{equation}
With this definition, $\Gu^{(1,0)}$ equals the one photon component $\Gu$ of an approximate ground state of $T(0)$ as defined by \eqref{eq:def-gamma1}.

Then we consider the following decomposition for $\phi$:

\noindent Let $a$, $b$ and $\gamma_0$ be defined by
\begin{equation}\label{def:Guab-0}
 \Proj_0 \phi = \gamma_0 \mab \vac , \quad\mbox{with}\quad |a|^2 + |b|^2  =1\, ,
\end{equation}
and let $\gamma_1$ and $R_1$ de defined by
\begin{equation}\label{def:Guab-2}
 \Proj_1\phi = (\sqrt{\alpha} \gamma_1 \Guab + R_1),\quad \bra R_1,\, \Guab\ket_*=0\, .
\end{equation}
Here the bilinear form  $\bra.,.\ket_*$ acts on $(\C^2\otimes\gF)^2$.

For $\G$ given by \eqref{*-0}, we define the following decomposition:
\begin{equation}\label{def:g}
 \Proj_0 \G =:g\, ,
\end{equation}
and for
\begin{equation}\label{def:Gamma-1-g}
 \Gu(g) := - (H_f+P_f^2)^{-1} \sigma.\Bc g\, ,
\end{equation}
similarly to \eqref{def:Guab}, we split $\Proj_1 \G$ as
\begin{equation}\label{eq:decomposition-again-5-bis}
 \Proj_1 \G=: \sqrt{\alpha}\beta_1 \Gu(g) + L_1
\end{equation}
where $\beta_1$ and $L_1$ are uniquely defined by the condition
\begin{equation}\label{eq:decomposition-again-6}
 \bra \Gu(g),\, L_1\ket_*=0\, .
\end{equation}

For
$$
 \Lambda : = \sup_{\zeta(r)\neq 0} |r|\,
$$
where $\zeta(.)$ is the ultraviolet cutoff, we define
\begin{equation}\label{eq:def-M}
  M(\Proj_1 \G ) =\left\{
   \begin{array}{lll}
     \|P \L_1 \|^2 & \mbox{ if } & |\beta_1|< 8(1+\Lambda) \\
     \|(P-P_f)\Proj_1 \G\|^2 & \mbox{ if } & |\beta_1| \geq 8(1+\Lambda)
   \end{array}
   \right. ,
\end{equation}

\begin{prop}\label{proposition:4} $\ $

i) For some $c>0$ we have
\begin{equation}\label{eq:estimate-gs-energy-1}
\begin{split}
 \Sigma_0 - \Sigma \leq c\alpha^2\, .
\end{split}
\end{equation}

ii) For $\pgs$ a ground state of $H$ such that $\|\Proj_0\pgs\|=1$ we have

\begin{equation}\label{eq:estimate-gs-energy-2}
\begin{split}
 & \| H_f^\frac12 \Proj_{\geq 2}\pgs\| = \mathcal{O}(\alpha)\, ,\\
 & \| \nabla g\| =  \mathcal{O}(\alpha)\, , \\
 & \|R_1\|_*=\mathcal{O}(\alpha),
 \quad \| L_1 \|_* = \mathcal{O}(\alpha)\, ,
\end{split}
\end{equation}
\begin{equation}\label{eq:estimate-gs-energy-3}
\begin{split}
 & (\gamma_1-\gamma_0) = \mathcal{O}(\alpha^\frac12)\, ,\\
 & (\beta_1-1) \| g\| =  \mathcal{O}(\alpha^\frac12)\, ,\\
 & \|(P-P_f) \Proj_{\geq 2}\pgs\| = \mathcal{O}(\alpha)\, .
\end{split}
\end{equation}
and
\begin{equation}\label{eq:estimate-gs-energy-4}
 M(\Proj_1 \G) = \mathcal{O}(\alpha^2)\, .
\end{equation}
\end{prop}
The proof of Proposition~\ref{proposition:4} requires several lemmata given below, and is postponed to the end of this section.

\begin{lem}\label{lem:1}
There exists a constant $c>0$ such that for all $\alpha$ and all $\psi\in\mathfrak{D}(H)$ we have
$$
 \bra \psi,\, H\psi\ket \geq -c\alpha \|\psi\|^2 +\frac12 \|H_f^\frac12\psi\|^2 + \frac14\|(P-P_f)\psi\|^2\, .
$$
\end{lem}

Since $\Sigma\leq 0$, there is an immediate consequence of this result for any ground state of $H$.
\begin{cor}
Let $\pgs$ be a ground state of $H$, then
$$
 \|H_f^\frac12\pgs\|^2 \leq c\alpha \|\pgs\|^2\, ,\quad\mbox{and}\quad \|(P-P_f)^2\pgs\|^2 \leq c\alpha\|\pgs\|^2\, .
$$
\end{cor}
%
%
\begin{proof} We first prove Lemma~\ref{lem:1} for the Hamiltonian $H$ without normal ordering. Reintroducing the normal ordering will shift the estimates by $\cno\alpha$, and will therefore not change the result.

For $\psi\in\mathfrak{D}(H)$, we have
\begin{equation}\label{eq:87}
\begin{split}
 & \bra \left[ (P-P_f -\sqrt{\alpha} A(0))^2 -\frac{\alpha}{|x|} +\sqrt{\alpha}\sigma.B(0) + H_f\right] \psi,\, \psi\ket \\
 & \geq \bra (\frac{1-\epsilon}{2} (P-P_f)^2 - \frac{\alpha}{|x|})\psi,\, \psi\ket
       +\bra \frac{1-\epsilon}{2} (P-P_f)^2\psi,\,\psi\ket \\
 &      + \alpha(1-\epsilon^{-1}) \bra A(0)^2\psi,\, \psi\ket + \sqrt{\alpha}\bra \sigma . B(0)\psi,\, \psi\ket + \bra H_f \psi,\, \psi\ket \\
 & \geq - c\alpha^2 \|\psi\|^2 +\frac{1-\epsilon}{2} \bra (P-P_f)^2\psi,\, \psi\ket
        - c \alpha \epsilon^{-1} \|\psi\|^2 - c\epsilon^{-1}\alpha \|H_f^\frac12\psi\|^2 \\
 &      - c\sqrt{\alpha} \|H_f^\frac12\psi\| - c\sqrt{\alpha} \|H_f^\frac12\psi\| + \|H_f^\frac12\psi\| \\
 & \geq - c\alpha \|\psi\|^2 + \frac12 \|H_f^\frac12 \psi\|^2 +\frac14 \|(P-P_f)\psi\|^2 \, ,
\end{split}
\end{equation}
for $\epsilon<1/2$ and $\alpha$ small enough, and where in the second inequality we used \cite[Lemma~A.4]{GLL} to bound from below the terms $A(0)^2$ and $\sigma.B(0)$.
\end{proof}

\begin{lem}\label{lem:2}
There exist $c>0$ such that
\begin{equation}\label{eq:itc}
\begin{split}
 & \bra [(P-P_f^2) - \frac{\alpha}{|x|} + H_f] \Proj_1\G,\, \Proj_1\G\ket \\
 & \geq
 \alpha \frac{1}{2(1 + \Lambda)}  |\beta_1|^2 \|\Gu (g)\|_*^2 + \frac{1}{2(1 + \Lambda)}
 \|\L_1\|_*^2 +
 \frac12 \|(P-P_f)\Proj_1 \G\|^2 -c\alpha^3\, .
\end{split}
\end{equation}
In addition, if $|\beta_1| < 8(1+\Lambda)$, then for any $\xi>0$,  there exist
$\frac{1}{4(1 + \Lambda)}>\delta_3>0$, $c>0$ and  $\frac12> \delta_4>0$ such that
\begin{equation}\label{eq:itc-2}
\begin{split}
 & \bra [(P-P_f^2) - \frac{\alpha}{|x|} + H_f] \Proj_1\G,\, \Proj_1\G\ket \\
 & \geq \alpha |\beta_1|^2 \|\Gu (g)\|_*^2 + \delta_3\|\L_1\|_*^2
 - \frac{\xi}{8} \|\nabla g\|^2 + \delta_4 \| P L_1\|^2 -c\alpha^3\, .
\end{split}
\end{equation}

\end{lem}
%
%
\begin{proof}
Let us denote
$$
 T_1:= \sqrt{\alpha}\beta_1 \Gu(g)\, ,
$$

With the photon number estimate \eqref{eq:photon-number-estimate-1} in Proposition~\ref{Nf-exp-lemma-1}, we have
\begin{equation*}
 \bra (\frac12 (P-P_f)^2 -\frac{\alpha}{|x|}) \Proj_1\G, \, \Proj_1\G\ket
 \geq -\frac12 \alpha^2 \|\Proj_1 \G\|^2 \geq -c\alpha^3\, .
\end{equation*}
Taking into account that
$$
 \|H_f^\frac12 \Proj_1 \G \|^2 \geq \frac{1}{2(1+\Lambda)} \|\Proj_1 \G \|_*^2
 = \frac{1}{2(1+\Lambda)} (\| T_1 \|_*^2 + \| \L_1 \|_*^2)\, ,
$$
yields \eqref{eq:itc}.

In the rest of the proof, we consider the case $|\beta_1|< 8(1+\Lambda)$.

$\bullet$ If $\|P\L_1\| > \tilde{c} \|P T_1\|$, for $\tilde{c}>1$, we have
\begin{equation}\label{eq:star-1}
\begin{split}
 & \bra (P-P_f)^2(T_1 + L_1),\, (T_1 + L_1)\ket + \|H_f^\frac12 (T_1 + L_1)\|^2 \\
 & = \bra P^2 (T_1 + L_1), T_1+L_1\ket + \|T_1\|_*^2 + \|L_1\|_*^2
 - 2\Re\bra P.P_f (T_1 + \L_1),\, (T_1+\L_1)\ket
\end{split}
\end{equation}
The last term in the right hand side of \eqref{eq:star-1} is estimated as
\begin{equation}\label{eq:star-2}
\begin{split}
 & - 2\Re \bra P.P_f (T_1+ \L_1), \, (T_1+ \L_1)\ket \\
 & = - 2\Re \bra P.P_f T_1,\, T_1\ket
 - 4\Re \bra P.P_f T_1,\, \L_1\ket
 - 2\Re \bra P. P_f L_1,\, \L_1\ket \\
 & \geq - 4\|P_f \L_1 \|\, \|P T_1\| - 2 \| P_f \L_1\|\, \|P \L_1\| \\
 & \geq -c\sqrt{\alpha} \|\nabla g\|^2 - \sqrt{\alpha} \|\L_1\|_*^2
 - (1-\epsilon)^{-1} \|P_f \L_1\|^2 - (1-\epsilon) \|P \L_1\|^2\, .
\end{split}
\end{equation}
In the first inequality, we used from Lemma~\ref{lem:appendix-3} that $\bra P.P_f T_1,\, T_1\ket=0$. In the second inequality we used $\|P_f \L_1\| \leq \|\L_1\|_*$. In addition, in the second inequality, we used Lemma~\ref{lem:appendix-6} stating that $\|\nabla \Gu(g)\| \leq c\|\nabla g\|$ which yields, since $|\beta_1|$ is bounded by $8(1+\Lambda)$, that  $\|P T_1\| \leq c \sqrt{\alpha} \|\nabla g\|$.

Since we have imposed a fixed ultraviolet cutoff, we can find $\epsilon>0$ in \eqref{eq:star-2} depending only on $\Lambda$, such that for some $\delta_3>0$ we have
\begin{equation}\label{eq:star-3}
 (1-2\delta_3)\|L_1\|_*^2 > (1-\epsilon)^{-1} \|P_f \L_1\|^2\, .
\end{equation}
Inequalities \eqref{eq:star-2} and \eqref{eq:star-3} in \eqref{eq:star-1} yield
\begin{equation}\label{eq:star-4}
\begin{split}
 & \bra (P-P_f)^2 (T_1 + \L_1),\, (T_1 + \L_1)\ket + \|H_f^\frac12 (T_1 + \L_1)\|^2 \\
 & \geq \bra P^2 (T_1 + L_1),\, (T_1 + L_1)\ket + \|T_1\|_*^2 +\delta_3 \|L_1\|_*^2
 -\frac{\xi}{8} \|\nabla g\|^2 - (1-\epsilon) \|P \L_1\|^2\, .
\end{split}
\end{equation}
The first term in the right hand side of \eqref{eq:star-4} can be estimated as
\begin{equation*}
\begin{split}
 & \bra P^2 (T_1 + \L_1),\, (T_1 + \L_1)\ket \\
 & \geq (\|P \L_1\| - \|P T_1\|)^2 \geq (1-\frac{1}{\tilde c})^2 \|P\L_1\|^2 \geq (1 -\frac{\epsilon}{2})^2 \|P \L_1\|^2\, ,
\end{split}
\end{equation*}
for $\tilde{c} = 4/\epsilon$. Remind that since $\epsilon$ depends only on $\Lambda$, then $\tilde{c}$ also depends only on $\Lambda$. Therefore, for the sum of the first term and the last term in the
right hand side of \eqref{eq:star-4} holds
\begin{equation*}
\begin{split}
 & \bra P^2 (T_1 + \L_1),\, (T_1 + \L_1)\ket - (1-\epsilon) \|P \L_1\|^2 \\
 & \geq \frac{\epsilon^2}{4} \| P \L_1\|^2 \geq \frac{\epsilon^2}{8} \| P \L_1\|^2 + \frac{\epsilon^2}{8} \| P T_1\|^2\, .\\
\end{split}
\end{equation*}
Since one half of the last two terms combined with $\bra -\frac{\alpha}{|x|} \Proj_1 \G,\, \Proj_1 \G\ket $ is bounded below by $-c\alpha^2\|\Proj_1\G\|^2 = -c\alpha^3$, it gives together with \eqref{eq:star-4} that there exists $\alpha_0$ depending on $\xi$ such that for all $\alpha_0>\alpha>0$ we have
\begin{equation*}
\begin{split}
 & \bra [(P-P_f)^2 -\frac{\alpha}{|x|} + H_f]\, (T_1 + \L_1),\,  (T_1 + \L_1 )\ket \\
 & \geq
 \|T_1\|_*^2 +\delta_3 \|\L_1 \|_*^2 -\frac{\xi}{8} \|\nabla g\|^2 + \frac{\epsilon^2}{16} \|P \L_1\|^2 - c\alpha^3 \, .
\end{split}
\end{equation*}
and thus proves \eqref{eq:itc-2} in the case $|\beta_1| < 8(1+\Lambda)$ and $\|P \L_1\| > \tilde{c} \|P T_1\|$.

$\bullet$ We now consider the case $\|P\L_1\| \leq \tilde{c} \|P T_1\|$.

Since $\bra P.P_f T_1,\, T_1\ket = 0$ (see Lemma~\ref{lem:appendix-3}) we obtain
\begin{equation}\label{eq:5.4-7.8}
\begin{split}
  & \Re \bra 2 P.P_f (T_1 + \L_1),\, (T_1 + \L_1)\ket = 4\Re \bra P.P_f T_1,\, \L_1\ket
  + 2\Re \bra P.P_f \L_1,\, \L_1\ket \\
  & \geq 4 \|P T_1\|\, \|P_f \L_1\| + 2 \| P_f \L_1\| \, \| P \L_1\|\, \\
  & \geq  4 \|P T_1\|\, \|P_f \L_1\| + 2 \| P_f \L_1\| \, \tilde{c} \| P T_1\|\, \\
  & \geq - c\sqrt{\alpha} \|\nabla g\|^2 - c\sqrt{\alpha} \|L_1\|_*^2\, ,
\end{split}
\end{equation}
where we used $\|P_f \L_1\| \leq \|L_1\|_*$ and $\|PT_1\| \leq c\sqrt{\alpha} \|\nabla g\|$ since $|\beta_1| < 8(1+\Lambda)$ and $\| \nabla \Gu(g) \| \leq c\| \nabla g\|$ from Lemma~\ref{lem:appendix-6}.

Therefore, we get
\begin{equation}\nonumber
\begin{split}
 & \bra [ (P-P_f)^2 -\frac{\alpha}{|x|} + H_f] (T_1 + \L_1),\, (T_1 + \L_1) \ket \\
 & \geq \bra [ (1-\delta_4) P^2 -\frac{\alpha}{|x|}] (T_1 + \L_1),\, (T_1 + \L_1) \ket
 + \|T_1\|_*^2 + \|L_1\|_*^2 \\
 & \ \ \ \ - c\sqrt{\alpha}\|\nabla g\|^2 - c\sqrt{\alpha} \|L_1\|_*^2 + \delta_4 \|P \L_1\|^2 \\
 &    \geq \|T_1\|_*^2 + \delta_3\| \L_1 \|_*^2 -\frac{\xi}{8} \|\nabla g\|^2 + \delta_4 \| P \L_1\|^2-c\alpha^3\, ,
\end{split}
\end{equation}
which proves \eqref{eq:itc-2} in the case $|\beta_1| < 8(1+\Lambda)$ and $\|P \L_1\| < \tilde{c} \|P T_1\|$, and thus concludes the proof.
\end{proof}

We have the following a priori estimate
\begin{lem}\label{lem:3} For $M(\Proj_1 \G )$ defined by \eqref{eq:def-M}, $\delta_3$ and $\delta_4$ defined as in Lemma~\ref{lem:2}, there exist $\delta_1>0$, $c_0>0$ and $c>0$ such that for all $\alpha>0$,
\begin{equation}\label{eq:star-9.5}
\begin{split}
 & \bra \Proj_{\leq 1} \pgs ,\, H\, \Proj_{\leq1} \pgs\ket \\
 & \geq
 \frac{\delta_1}{2} \|\nabla g\|^2 + \delta_4 M(\Proj_1 \G)
 -\alpha |\gamma_0|^2  \|\Guab \|_*^2 + \alpha|\gamma_0 - \gamma_1|^2  \|\Guab \|_*^2 + \| R_1\|_*^2\\
 &
 -\alpha\|\Gu (g)\|_*^2 + c_0\alpha |\beta_1-1|^2 \|g\|^2 + \delta_3 \|L_1\|_*^2 - c\alpha^2\, .
\end{split}
\end{equation}
\end{lem}
%
%
\begin{proof}
Substituting the expression for $\Proj_{\leq 1} \pgs$ according to the
decomposition \eqref{*-0}-\eqref{def:Guab-2} yields
\begin{equation}\label{eq:star-10}
\begin{split}
 \bra H \Proj_{\leq 1}\pgs,\, \Proj_{\leq 1}\pgs\ket =
 & \ \bra H \Proj_{\leq 1} u_\alpha\phi,\, \Proj_{\leq 1} u_\alpha\phi\ket
 + \bra H \Proj_{\leq 1} \G,\, \Proj_{\leq 1} \G\ket \\
 &
 + \bra H \Proj_{\leq 1} u_\alpha \phi,\, \Proj_{\leq 1} \G\ket
 + \bra H \Proj_{\leq 1} \G ,\, \Proj_{\leq 1} u_\alpha \phi\ket
\end{split}
\end{equation}
According to \eqref{expand-H} and the splitting \eqref{def:Guab-0}-\eqref{def:Guab-2} we have
\begin{equation}\label{eq:star-11}
\begin{split}
 & \ \bra H \Proj_{\leq 1} u_\alpha\phi,\, \Proj_{\leq 1} u_\alpha\phi\ket \\
 & \geq \alpha \|\Guab\|_*^2 |\gamma_1|^2 + \bra (-\Delta -\frac{\alpha}{|x|}) u_\alpha,\, u_\alpha\ket |\gamma_0|^2 \\
 & \ \ \ - 2 \alpha \Re \gamma_1\bar\gamma_0 \|\Guab\|_*^2 + \|R_1\|_*^2
 + 2 \sqrt{\alpha}\Re \bra \sigma.\Ba R_1, g_0\ket \\
 & \geq \|R_1\|_*^2 + (|\gamma_1|^2 -2\Re\gamma_1\bar\gamma_0) \alpha \|\Guab\|_*^2 - |\gamma_0|^2 \frac{\alpha^2}{4} + 2\sqrt{\alpha} \bra R_1,\, \Guab\ket_* \\
 & \geq \|R_1\|_*^2 - \alpha |\gamma_0|^2  \|\Guab\|_*^2 +\alpha |\gamma_0 - \gamma_1|^2 \|\Guab\|_*^2 - |\gamma_0|^2\frac{\alpha^2}{4}\,
\end{split}
\end{equation}
where we used $\bra R_1,\, \Guab\ket_*=0$ and $P_f.A(0)\vac=0$.

In the case $|\beta_1| \geq 8(1+\Lambda)$, using Inequality \eqref{eq:itc} of Lemma~\ref{lem:2} implies
\begin{equation}\label{eq:star-12-0}
\begin{split}
 & \bra H \Proj_{\leq 1} \G,\, \Proj_{\leq 1} \G\ket \\
 & \geq \alpha \frac{1}{2(1+\Lambda)} |\beta_1|^2 \|\Gu(g)\|_*^2 + 2 \delta_3 \|\L_1\|_*^2
   + \frac12 \|(P-P_f)\Proj_1 \G\|^2 -c\alpha^3 \\
 & + \bra (-\Delta -\frac{\alpha}{|x|})g,\, g\ket
 - 4\sqrt{\alpha} \Re \bra P.\Aa \Proj_1 \G, \Proj_0 \G\ket
   + 2\sqrt{\alpha}\Re \bra \sigma.\Ba \Proj_1 \G,\, \Proj_0 \G\ket.
\end{split}
\end{equation}
For the last three terms in the right hand side of \eqref{eq:star-12-0}, we use the following three inequalities \eqref{eq:ineq-1}-\eqref{eq:ineq-3}

 \begin{equation}\label{eq:ineq-1}
 \bra(-\Delta-\frac{\alpha}{|x|})g,\, g\ket \geq
-\frac{\alpha^2}{4} \|g\|^2 + \delta_1 \|\nabla g\|^2 \, .
 \end{equation}
Since $g\perp u_\alpha$ and since from Lemma~\ref{lem:appendix-4} we have $\bra\Gu(g),\, P.\Ac \Proj_0 \G \ket =0$, we get
 \begin{equation}\label{eq:diese-5}
 \begin{split}
  & -4 \sqrt{\alpha}\Re \bra P.\Aa \Proj_1 \G,\, \Proj_0 \G \ket \\
  & = - 4 \sqrt{\alpha} \Re \bra P.\Aa \sqrt{\alpha} \beta_1\Gu(g),\, \Proj_0 \G\ket
      - 4 \sqrt{\alpha}\bra P.\Aa \L_1,\, \Proj_0 \G\ket \\
  & \geq -c \sqrt{\alpha} \|\nabla g\|^2 - c\sqrt{\alpha} \|\L_1\|_*^2
  \geq -\frac{\delta_1}{4}\|\nabla g\|^2 - c\sqrt{\alpha} \|\L_1\|_*^2 \, .
  \end{split}
\end{equation}
And finally, we have
\begin{equation}\label{eq:ineq-3}
 \begin{split}
  & \Re\bra \sigma.\Ba \Proj_1 \G,\, \Proj_0 \G \ket = -\sqrt{\alpha}\Re \bra \beta_1 \Gu(g) ,\, \Gu(g)\ket_* +
  \bra \L_1,\, \Gu(g)\ket_* \\
  & = -\sqrt{\alpha}\Re \beta_1 \|\Gu(g)\|_*^2\, .
  \end{split}
\end{equation}
This three inequalities \eqref{eq:ineq-1}-\eqref{eq:ineq-3} together with \eqref{eq:star-12-0} gives
\begin{equation}\label{eq:star-12}
\begin{split}
 & \bra H \Proj_{\leq 1} \G,\, \Proj_{\leq 1} \G\ket \\
 &
 \geq (\frac{1}{2(1+\Lambda)}|\beta_1|^2 -2\Re\beta_1)\alpha \|\Gu(g)\|_*^2 + \delta_3 \|\L_1\|_*^2
 + \frac12 \|(P-P_f)\Proj_1 \G\|^2 \\
 &\ \ \  +\frac{3}{4}\delta_1 \|\nabla g\|^2 - c\alpha^2
 \\
 & \geq (\frac{1}{2(1+\Lambda)}|\beta_1|^2 -2\Re\beta_1)\alpha \|\Gu(g)\|_*^2 + \delta_3 \|\L_1\|_*^2
 + \delta_4 M(\Proj_1 \G)\\
 & +\frac{3}{4}\delta_1 \|\nabla g\|^2 - c\alpha^2   \, .
\end{split}
\end{equation}
In the last inequality we used $\frac12 \|(P-P_f) \Proj_1 \G\|^2 \geq \delta_4 M(\Proj_1 \G)$.

Similarly, in the case $|\beta_1| < 8(1+\Lambda)$, using this time Inequality\eqref{eq:itc-2} of Lemma~\ref{lem:2}, the fact that $\delta_4 \|P\L_1\|^2 \geq \delta_4 M(\Proj_1 \G)$ and again the three inequalities \eqref{eq:ineq-1}-\eqref{eq:ineq-3}, we obtain
\begin{equation}\label{eq:star-12-bis}
\begin{split}
 & \bra H \Proj_{\leq 1} \G,\, \Proj_{\leq 1} \G\ket \\
 & \geq (|\beta_1|^2 -2\Re\beta_1)\alpha \|\Gu(g)\|_*^2 + \delta_3 \|\L_1\|_*^2
   -\frac{\xi}{2}\|\nabla g\|^2
 + \delta_4 \|P \L_1\|^2 \\
 & +\frac34 \delta_1 \|\nabla g\|^2 - c\alpha^2   \\
 & \geq (|\beta_1|^2 -2\Re\beta_1)\alpha \|\Gu(g)\|_*^2 + \delta_3 \|\L_1\|_*^2
 + \delta_4 M(\Proj_1 \G) +\frac12 \delta_1 \|\nabla g\|^2 - c\alpha^2   \, .
\end{split}
\end{equation}

To pursue the estimate, we consider two cases.

If $|\beta_1| \geq 8(1+\Lambda)$, we have $\frac{1}{2(1+\Lambda)}|\beta_1|^2 - 2\Re\beta_1 \geq
\frac{1}{4(1+\Lambda)}|\beta_1|^2$, and thus, for some $c_0>0$,
\begin{equation}\label{eq:star-12.1-geq}
\begin{split}
& (\frac{1}{2(1+\Lambda)}|\beta_1|^2 - 2\Re\beta_1)\alpha \|\Gu(g)\|_*^2 \\
& \geq \alpha c_0 |\beta_1|^2 \|\Gu(g)\|_*^2
\geq - \alpha \|\Gu(g)\|_*^2 +\alpha c_0|\beta_1 -1|^2 \, \|g\|^2\, .
\end{split}
\end{equation}
Therefore, using \eqref{eq:star-12} yields
\begin{equation}\label{eq:star-13}
\begin{split}
 \bra H \Proj_{\leq 1} \G,\, \Proj_{\leq 1} \G \ket \geq
 & -\alpha \|\Gu(g)\|_*^2 + c_0\alpha |\beta_1-1|^2 \|g\|^2 + \delta_3 \|L_1\|_*^2 \\
 & + \frac{\delta_1}{2} \|\nabla g\|^2 + \delta_4 M(\Proj_1 \G ) - c\alpha^2\, .
\end{split}
\end{equation}

If $|\beta_1| < 8(1+\Lambda)$, then we have
\begin{equation}\label{eq:star-12.1-<}
 (|\beta_1|^2 - 2\Re\beta_1)\alpha \|\Gu(g)\|_*^2 \geq -\alpha \|\Gu(g)\|_*^2 + \alpha|1-\beta_1|^2 \|\Gu(g)\|_*^2\, .
\end{equation}
Therefore, using \eqref{eq:star-12-bis}, it also yields \eqref{eq:star-13}

Eventually, wee need to estimate the cross terms between $\Proj_{\leq 1} \G$
and $\Proj_{\leq 1} u_\alpha \phi$. The photon number bound \eqref{eq:photon-number-estimate-1} in Proposition~\ref{Nf-exp-lemma-1} implies
\begin{equation}\label{eq:star-14}
\begin{split}
 & - 4\Re \bra P.P_f \Proj_1 u_\alpha\phi,\, \Proj_1 \G\ket \geq - c\alpha \|\Proj_1 \phi\|\,
 \|P_f\Proj_1 \G\| \geq -c\alpha^2\, .
\end{split}
\end{equation}
Similarly we have
\begin{equation}\label{eq:star-15}
 -2\sqrt{\alpha} \Re \bra P.\Aa \Proj_1 u_\alpha\phi,\, g\ket - 2 \sqrt{\alpha}\Re\bra P.\Aa \Proj_1 \G,\, \Proj_0 u_\alpha\phi\ket \geq -c\alpha^2\, .
\end{equation}
Collecting \eqref{eq:star-10}, \eqref{eq:star-11}, \eqref{eq:star-13}, \eqref{eq:star-14} and \eqref{eq:star-15} proves \eqref{eq:star-9.5}.
\end{proof}

We are now ready to give the proof of Proposition~\ref{proposition:4}.

%
%
\noindent\textit{Proof of Proposition~\ref{proposition:4}.}
We first compute the terms involving only more than one-photon sectors. Lemma~\ref{lem:1} and Proposition~\ref{Nf-exp-lemma-1} yield
\begin{equation}\label{eq:five-prime}
\begin{split}
 \bra H\,  \Proj_{\geq 2}\pgs,\,  \Proj_{\geq 2} \pgs \ket \geq -c\alpha^2 +\frac12 \|H_f^\frac12 \Proj_{\geq 2} \pgs\|^2
 + \frac14 \| (P-P_f)^2 \Proj_{\geq 2} \pgs\|^2
\end{split}
\end{equation}

We next estimate cross terms between more than one-photon sector and less than one-photon sector
\begin{equation}\label{eq:five-star}
\begin{split}
 & \bra H \Proj_{\geq 2}\pgs,\, \Proj_{\leq 1} \pgs \ket + \bra H \Proj_{\leq 1}\pgs,\, \Proj_{\geq 2} \pgs \ket\\
 & = 2\sqrt{\alpha} \Re\bra \sigma.\Ba \Proj_2 \pgs,\Proj_1\pgs\ket
 - 4\sqrt{\alpha} \Re \bra (P-P_f).\Aa\Proj_2 \pgs,\, \Proj_1\pgs\ket \\
 & + 2\sqrt{\alpha} \Re \bra \Aa.\Aa \Proj_2 \pgs,\, \Proj_0\pgs\ket +
 2\alpha \Re \bra \Aa.\Aa \Proj_3 \pgs,\, \Proj_1 \pgs\ket\, .
\end{split}
\end{equation}
The first term in the right hand side of \eqref{eq:five-star} is bounded by
\begin{equation}\label{eq:six-star}
\begin{split}
 2\sqrt{\alpha}\Re \bra \sigma.\Ba \Proj_2 \pgs,\, \Proj_1\pgs\ket &\geq -c\alpha\|\Proj_1 \pgs\|^2 - \epsilon \|\Ba \Proj_2 \pgs\|^2 \\
 & \geq -c\alpha^2- \epsilon\|H_f^\frac12 \Proj_2 \pgs\|^2\, ,
\end{split}
\end{equation}
using from \cite[Lemma~A4]{GLL} that $\|\Ba \pgs\| \leq c\|H_f^\frac12 \pgs\|$.

The second term in the right hand side of \eqref{eq:five-star} is bounded by
\begin{equation}\label{eq:seven-star}
\begin{split}
&  -4\sqrt{\alpha} \Re \bra (P-P_f).\Aa \Proj_2 \pgs,\, \Proj_1\pgs \ket \\
 & \geq -\epsilon \|(P-P_f)\Proj_2 \pgs\|^2 - c\alpha \|\Ac \Proj_1 \pgs\|^2 \\
 & \geq -\epsilon \|(P-P_f)\Proj_2 \pgs\|^2 - c\alpha^2\, ,
\end{split}
\end{equation}
using that $\|\Ac\Proj_1\pgs\| \leq c\|\Proj_1\pgs\|$, and $\|\Proj_1\pgs\| = \mathcal{O}(\alpha)$ from the photon number bound \eqref{eq:photon-number-estimate-1}.

The third term in the right hand side of \eqref{eq:five-star} is estimated as
\begin{equation}\label{eq:eight-star}
 2\alpha\Re \bra \Aa.\Aa\Proj_2 \pgs, \Proj_0\pgs\ket \geq - c\alpha^2 \|\Proj_0\pgs\|^2 - \epsilon\|H_f^\frac12 \Proj_2 \pgs\|^2\, .
\end{equation}

Similarly, the last term in the right hand side of \eqref{eq:five-star} is
\begin{equation}\label{eq:nine-star}
 2\alpha\Re\bra \Aa.\Aa \Proj_3\pgs,\Proj_1\pgs\ket \geq - c\alpha^2 \|\Proj_1 \pgs\|^2
 - \epsilon\|H_f^\frac12\Proj_3\pgs \|^2\, .
\end{equation}
Inequalities \eqref{eq:five-prime}-\eqref{eq:nine-star}, together with Lemma~\ref{lem:3} yield
\begin{equation}\nonumber
\begin{split}
 &\bra \pgs,\, H\, \pgs\ket  \\
 & \geq
 -\alpha |\gamma_0|^2 \|\Guab\|_*^2
 + \alpha|\gamma_1-\gamma_0 |^2 \|\Guab\|_*^2 + \|R_1\|_*^2 \\
 &
 -\alpha \|\Gu(g)\|_*^2 + c_0\alpha |\beta_1 -1|^2 \|g\|^2 +\delta_3 \|L_1\|_*^2
 + \frac{\delta_1}{2} \|\nabla g\|^2 + \delta_4 M(\Proj_1 \G ) \\
 &
 +\frac14 \|H_f^\frac12 \Proj_{\geq 2} \pgs\|^2
 + \frac18 \|(P - P_f)^2 \Proj_{\geq 2} \pgs\|^2 +\mathcal{O}(\alpha^2)\, \\
\end{split}
\end{equation}
Since from \eqref{eq:app-2} of Lemma~\ref{lem:appendix-6} we have
\begin{equation}\nonumber
\begin{split}
 & |\gamma_0|^2 \|\Guab\|_*^2 + \|\Gu(g)\|_*^2 = \|\Gu^{(1,0)}\|_*^2 (|\gamma_0|^2 + \|g\|^2) \\
 & = \|\Gu^{(1,0)}\|_*^2 \|\Proj_0 \pgs\|^2 = \|\Gu(\Proj_0 \pgs)\|_*^2\, ,
\end{split}
\end{equation}
we thus get
\begin{equation}\label{eq:kti-1}
\begin{split}
 &\bra \pgs,\, H\, \pgs\ket  \\
 &\geq -\alpha\|\Gu\|_*^2 + \alpha|\gamma_1-\gamma_0 |^2 \|\Gu\|_*^2 + \|R_1\|_*^2 \\
 & + c_0\alpha |\beta_1 -1|^2 \|g\|^2 +\delta_3 \|L_1\|_*^2
 + \frac{\delta_1}{2} \|\nabla g\|^2 + \delta_4 M(\Proj_1 \G) \\
 &
 +\frac14 \|H_f^\frac12 \Proj_{\geq 2} \pgs\|^2
 + \frac18 \|(P - P_f)^2 \Proj_{\geq 2} \pgs\|^2 +\mathcal{O}(\alpha^2)\, .
\end{split}
\end{equation}

On the other hand, according to \eqref{eq:lower-bound-main} and \eqref{eq:sigma-zero} holds
\begin{equation}\label{eq:eq:5.40}
 \Sigma \leq -\alpha\|\Gu\|_*^2 + \mathcal{O}(\alpha^2)\, .
\end{equation}
Since $\| \pgs\|^2 = 1+\mathcal{O}(\alpha)$,
we have with \eqref{eq:eq:5.40}
\begin{equation}\nonumber
\begin{split}
 \bra \pgs,\, H\, \pgs\ket =\Sigma \|\pgs\|^2 \leq -\alpha\|\Gu\|_*^2
 +\mathcal{O}(\alpha^2)\, .
\end{split}
\end{equation}
Together with \eqref{eq:kti-1}, this yields \eqref{eq:estimate-gs-energy-1}-\eqref{eq:estimate-gs-energy-4}.



\section{Refined photon number estimates}\label{S7}

The photon number estimate obtained in subsection~\ref{section:S4.1} enabled us to reduce the set of states among which we looked for $\pgs$, a minimizer of the functional $\bra \psi,\, H\psi\ket/\|\psi\|^2$.

As a next step we shall use the result of Proposition~\ref{proposition:4} to get refined estimates on the expected photon number for different parts of $\pgs$, which in its turn will be used in Section~\ref{S-new-7} for further reduction of the set of possible minimizers, and as a result for improving the estimate of $\Sigma$.


\begin{prop}\label{IPNE}[Refined photon number bound]
We have
$$
 \| \Proj_{\geq 2} \pgs\| ,\, \| R_1\| ,\, \|L_1\| =\mathcal{O}(\alpha^\frac56)\, .
$$
\end{prop}
%
%
\begin{proof}
The assumption $\|\Proj_0\pgs\|=1$ and the decomposition \eqref{*-0}-\eqref{def:Guab-0} imply that $|\gamma_0|$ is bounded by $1$. Therefore, from \eqref{eq:estimate-gs-energy-3} of Proposition~\ref{proposition:4} we obtain that $|\gamma_1|$ is bounded. Similarly, we have $\beta_1\|g\|$ bounded.

For $\rest:= \pgs -\sqrt{\alpha}\gamma_1\Guab u_\alpha - \sqrt{\alpha}\beta_1 \Gu(g)$, we get
\begin{equation}\label{eq:estimate-ipne-1}
\begin{split}
& \int_{|k|\leq \alpha^\frac13} \| a_\lambda(k)\rest\|^2 \mathrm{d} k
  \leq 3 \int_{|k|\leq\alpha^\frac13} \|a_\lambda(k) \pgs\|^2  \mathrm{d} k \\
& + 3 \alpha|\gamma_1|^2 \int_{|k|\leq\alpha^\frac13} \|a_\lambda(k) \Guab\|^2  \mathrm{d} k
  + 3 \alpha |\beta_1|^2 \int_{|k|\leq\alpha^\frac13} \|a_\lambda(k) \Gu(g) \|^2  \mathrm{d} k \mathrm{d} x
\end{split}
\end{equation}
Moreover, from Lemma~\ref{lem:IPNE} we have
$$
 \sum_{\lambda=1,2} \int_{|k|\leq \alpha^\frac13} \| a_\lambda(k) \Guab \|^2 \mathrm{d}k \leq c\alpha^\frac23\, ,
$$
and similarly
$$
 \sum_{\lambda=1,2} \int_{|k|\leq \alpha^\frac13} \| a_\lambda(k) \Gu(g) \|^2 \mathrm{d}k \, \mathrm{d} x\leq c \|g\|^2\alpha^\frac23\, .
$$
Therefore \eqref{eq:estimate-ipne-1} and boundedness of $|\gamma_1|$ and $|\beta_1|\, \|g\|$ yields
\begin{equation}\nonumber
\begin{split}
 & \int_{|k|\leq \alpha^\frac13} \| a_\lambda(k)\rest\|^2 \mathrm{d} k
 \leq 3 \left(\int_{|k|\leq\alpha^\frac13} \|a_\lambda(k) \pgs\|^2  \mathrm{d} k\right)
 + c\alpha^\frac53 \, .
\end{split}
\end{equation}
Thus, using the bounds \eqref{a-psi-fund-1} and \eqref{a-psi-fund-2} given in the proof of the photon number estimate \eqref{eq:photon-number-estimate-1}, we obtain
\begin{equation}\nonumber
\begin{split}
 & \int_{|k|\leq \alpha^\frac13} \| a_\lambda(k)\rest\|^2 \mathrm{d} k  \\
 & \leq 3 \left( \int_{|k|\leq\alpha^\frac74} \frac{c\alpha^{-\frac32}}{|k|} \, \mathrm{d} k
 + \int_{\alpha^\frac74 \leq |k|\leq \alpha^\frac13}  \frac{\alpha^2}{|k|^3} + \frac{\alpha}{|k|}\, \mathrm{d} k \right) + c\alpha^\frac53
 \leq c\alpha^\frac53 \, .
\end{split}
\end{equation}
This implies
\begin{equation}\label{eq:IPNE-1}
\begin{split}
 \bra \rest,\, N_f\, \rest\ket
 & = \sum_{\lambda=1,2}\int_{|k|\leq \alpha^\frac13} \|a_\lambda(k) \rest\|^2 \mathrm{d} k
 + \sum_{\lambda=1,2}\int_{|k| >\alpha^\frac13} \|a_\lambda(k) \rest\|^2 \mathrm{d} k \\
 & \leq c\alpha^\frac53 + \sum_{\lambda=1,2}\int_{|k| >\alpha^\frac13}
 |k|\, \alpha^{-\frac13}\|a_\lambda(k) \rest\|^2 \mathrm{d} k \\
 &\leq c\alpha^\frac53 + \alpha^{-\frac13} \|H_f^\frac12 \rest\|^2
  \leq c\alpha^\frac53\, ,
\end{split}
\end{equation}
where in the last inequality we used
\begin{equation}\nonumber
\begin{split}
 \|H_f^\frac12 \rest\|^2
 & = \|H_f^\frac12 \Proj_{\geq 2} \pgs\|^2 + \|H_f^\frac12 u_\alpha R_1\|^2 + \|\ H_f^\frac12 L_1 \|^2 \\
 & \leq \|H_f^\frac12 \Proj_{\geq 2} \pgs\|^2 + \| R_1\|_*^2 + \|\ L_1 \|_*^2 \, ,
\end{split}
\end{equation}
and the estimates \eqref{eq:estimate-gs-energy-2} of Proposition~\ref{proposition:4}.

The identity
 $$
  \bra \rest,\, N_f\, \rest\ket = \bra \Proj_{\geq 2}\pgs, \, N_f\, \Proj_{\geq 2}\pgs\ket
  + \bra u_\alpha R_1,\, N_f\, u_\alpha R_1\ket + \bra \L_1,\, N_f\, \L_1\ket\, ,
 $$
and the inequality \eqref{eq:IPNE-1} conclude the proof.
\end{proof}


\section{Upper bound on the binding energy up to the order $\alpha^2$ with error $\mathcal{O}(\alpha^{\frac83})$}\label{S-new-7}


In this section, we apply the results of Section~\ref{S7} to improve the upper bound on the binding energy (see Proposition~\ref{prop:main-cmain}) and to derive additional information regarding the ground state $\pgs$ (see Corollary~\ref{Cmain}). To this end, we introduce a splitting of $\pgs$ in the two-photon sector. We then state the main results, Proposition~\ref{prop:main-cmain} and Corollary~\ref{Cmain}. We give both the statement and the proof of Lemmata used to establish the result of Proposition~\ref{prop:main-cmain}, and then prove Proposition~\ref{prop:main-cmain}.

Recall the decomposition
$$
 \pgs = u_\alpha\phi + \G\,
$$
given by \eqref{*-0}, with the normalization $\|\Proj_0 \pgs\|=1$.

As in \eqref{def:Guab-0}, let $a$, $b$ and $\gamma_0$ be defined by
\begin{equation}\nonumber
 \Proj_0 \phi = \gamma_0 \mab \vac , \quad\mbox{with}\quad |a|^2 + |b|^2  =1\, ,
\end{equation}

We then decompose $\Proj_1\phi$ as in \eqref{def:Guab-0}-\eqref{def:Guab-2}
$$
 \Proj_1 \phi = \sqrt{\alpha}\gamma_1 \Guab +R_1, \quad \bra\Guab,\, R_1\ket_*=0\, .
$$

 Similarly to the two-photon component of an approximate ground state of $T(0)$ (see \eqref{eq:def-gamma2}), we define $\Gdab$ by
\begin{equation}\label{def:pi-2-phi-1}
\begin{split}
 \Gdab : = -(H_f+P_f^2)^{-1}\left(\sigma.\Bc\Guab + 2 \Ac.P_f \Guab + \Ac.\Ac\vac\mab\right)\, ,
\end{split}
\end{equation}
and set
\begin{equation}\label{def:pi-2-phi-2}
 \Proj_2 \phi =: \alpha\gamma_2\Gdab + R_2\, ,
\end{equation}
where $\gamma_2$ and $R_2$ are defined by the condition
\begin{equation}\label{def:pi-2-phi-3}
\bra R_2,\, \Gdab\ket_* = 0\, .
\end{equation}

The part $\Proj_1 \G$ is splitted as in \eqref{def:g}-\eqref{eq:decomposition-again-6}:
$$
 \Proj_1 \G= \sqrt{\alpha}\beta_1 \Gu(g) + L_1 \quad ,
 \bra \Gu(g),\, L_1\ket_*=0\, .
$$

In addition we define the following splitting of $\Proj_2 \G$.
Let $g:=\Proj_0 \G$, and let
\begin{equation}\label{def:Gamma-deux-g}
\begin{split}
 \Gd(g) : = -(H_f+P_f^2)^{-1}\left(\sigma.\Bc\Gu(g) + 2 \Ac.P_f \Gu(g) + \Ac.\Ac g \right)\, .
\end{split}
\end{equation}
We write
\begin{equation}\nonumber
 \Proj_2\G = \alpha \beta_2\Gd(g) + \L_2\, ,
\end{equation}
where $\beta_2$ and $\L_2$ are defined by
\begin{equation}\nonumber
 \bra \L_2,\, \Gd(g)\ket_* = 0\, .
\end{equation}

\begin{prop}\label{prop:main-cmain}
There exist $\alpha_0$, $\delta_1$, $\delta_2$, $\delta_3$, $\delta_5$, $c_0$, $c_1$ and $c$, strictly positive constants, such that for all $0<\alpha<\alpha_0$ we have
\begin{equation}\nonumber
\begin{split}
  & \bra \pgs,\, H\pgs\ket \geq (\Sigma_0 - \frac{\alpha^2}{4}) \|\pgs\|^2 + S[\pgs]\, ,
\end{split}
\end{equation}
where
\begin{equation*}
\begin{split}
  S[\pgs] =
  & \frac12(\|R_1\|_*^2 + \|R_2\|_*^2 + \frac{\delta_3}{2} \| \L_1 \|_*^2 + \frac{\delta_5}{2} \| \L_2 \|_*^2)
  + \frac14 \|H_f^\frac12 \Proj_{\geq 3} \pgs\|^2 \\
  &  + c_1 \alpha|\gamma_1 - \gamma_0|^2 + c_1\alpha^2 |\gamma_2-\gamma_1|^2  + c_0 \alpha |1-\beta_1|^2 \|g\|^2\\
  &  +c_0 \alpha^2 |1-\beta_2|^2 \|g\|^2
  + \frac{\delta_1}{2} \|\nabla g\|^2 + \frac{\delta_2}{2} \alpha^2 \|g\|^2 - c\alpha^\frac83 \, .
\end{split}
\end{equation*}
\end{prop}
The proof of the proposition is postponed to the end of this section.

\begin{rem}
Notice that for the trial state $\pst$ used in section~\ref{S3} to get a lower bound on the binding energy, a splitting similar to the one we are using here give $\gamma_1=1$, $\gamma_2=1$, $R_1=0$, $R_2=0$, $L_1=0$, $L_2=0$, $g=0$ and $\Proj_{\geq 2} \pst =0$. Our goal is to show that this trial state is close to the true ground state $\pgs$. Corollary~\ref{Cmain} gives estimate on the difference between $\pst$ and $\pgs$.
\end{rem}
\begin{cor}\label{Cmain}
We have following estimates
\begin{equation}\label{eq:nn1}
\begin{split}
 & \|H_f^\frac12 \Pi_{\geq 3} \pgs\| = \mathcal{O}(\alpha^\frac43)\, ,\\
 & \| g \| = \mathcal{O}(\alpha^\frac13)\, ,\quad
 \| \nabla g\| = \mathcal{O}(\alpha^\frac43),\quad \|\Proj_1 \G\| = \mathcal{O}(\alpha^\frac56)
 \, ,\\
 & \| R_1\|_*,\, \|L_1\|_*,\, \|R_2\|_*,\, \|L_2\|_*  = \mathcal{O}(\alpha^\frac43)\, ,\\
 & |1-\beta_1|\, \| g \|  = \mathcal{O}(\alpha^\frac56)\, ,\
 |\beta_1|\, \| g \|   = \mathcal{O}(\alpha^\frac13)\, ,\\
 & |1-\beta_2|\, \| g \|  = \mathcal{O}(\alpha^\frac13)\, ,\
 |\beta_2|\, \| g \|   = \mathcal{O}(\alpha^\frac13)\, ,\
\end{split}
\end{equation}
\begin{equation}\label{eq:nn2}
|\gamma_0 - 1| = \mathcal{O}(\alpha^\frac23),\, \ | \gamma_2 - 1 |=  \mathcal{O}(\alpha^\frac13)\, ,
| \gamma_1 - 1 |= \mathcal{O}(\alpha^\frac23)\, ,
\end{equation}
and
\begin{equation}\label{eq:nn3}
 \| H_f^\frac12 \Proj_2 \G \| = \mathcal{O}(\alpha^\frac43)\, .
\end{equation}
\end{cor}
\begin{proof}
Proposition~\ref{prop:main-cmain} and the same arguments as in the proof of Proposition~\ref{proposition:4}~ii) yields
$$
 S[\pgs] \leq 0\, ,
$$
which proves \eqref{eq:nn1}. This also yields
\begin{equation}\label{eq:nn4}
 |\gamma_2 - \gamma_1| = \mathcal{O}(\alpha^\frac13)\quad\mbox{and}\quad |\gamma_1 - \gamma_0| = \mathcal{O}(\alpha^\frac56)\, .
\end{equation}
Since $\|\Proj_0 \pgs\| = 1 = |\gamma_0|^2 + \|g\|^2$, and $\|g\| = \mathcal{O}(\alpha^\frac13)$, we obtain
$$
 |\gamma_0 -1| = \mathcal{O}(\alpha^\frac23)\,
$$
which, together with \eqref{eq:nn4}, implies \eqref{eq:nn2}. To prove \eqref{eq:nn3}, it suffices to apply \eqref{eq:app-3} and the fact that $\|g\| = \mathcal{O}(\alpha^\frac13)$ and $\|\L_2\|_* = \mathcal{O}(\alpha^\frac43)$.
\end{proof}
\begin{rem}
The estimates of Corollary~\ref{Cmain} allow to improve the estimates on the norm of the functions, similarly to what is done in Proposition~\ref{IPNE}. However, we shall not need this improved estimates in the present article.
\end{rem}

In the rest of this section, we state and prove the Lemmata that enable us to prove Proposition~\ref{prop:main-cmain} and then prove the proposition itself. We follow the same strategy as in Section~\ref{S6}. We first estimate the quadratic form of $H$ on the part of $\pgs$ corresponding to $n$-photon sectors with $n\geq 3$ (Lemma~\ref{lem:5}). The contribution of this part is small because $\|\Proj_{\geq 3} \pgs\| = \mathcal{O}(\alpha^\frac56)$. Then in Lemma~\ref{lem:6} we estimate the cross term $\bra \Proj_{\geq 3}\pgs,\, H\, \Proj_{\leq 2}\pgs\ket$. The most involved part is the estimate of the quadratic form on the projection of $\pgs$ onto the sectors with $n\leq2$ photons. It is done in Lemmata~\ref{lem:7} to \ref{lem:12}. Here we use the splitting described at the beginning of the current section. For the convenience of the reader, we recall here the expression for $H$.
\begin{equation}\nonumber
\begin{split}
 H =
 & (-\Delta -\frac{\alpha}{|x|}) + (H_f + P_f^2) - 2\,P. P_f
 - 4\alpha^\frac12 \Re P . \Aa  \\
 & + 4 \alpha^\frac12 P_f . \Aa
 + 2\alpha\Ac\Aa + 2\alpha (\Re\Aa)^2 + 2\alpha^\frac12\Re \sigma. \Ba\, .
\end{split}
\end{equation}

\begin{lem}\label{lem:5}
 We have
 $$
 \bra \Proj_{\geq 3} \pgs,\, H\, \Proj_{\geq 3}\pgs \ket \geq \frac12 \|H_f^\frac12 \Proj_{\geq 3} \pgs\|^2
 + \frac14 \|(P-P_f)\Proj_{\geq 3} \pgs\|^2 - c\alpha^\frac83\, .
 $$
\end{lem}
%
%
\begin{proof}
Picking $\psi= \Proj_{\geq 3} \pgs$  in \eqref{eq:87}
and using from Proposition~\ref{IPNE} that $\|\Proj_{\geq 3}\pgs\|^2 = \mathcal{O}(\alpha^\frac53)$ concludes the proof.
\end{proof}

\begin{lem}\label{lem:6}
 The following estimate holds
 $$
 \bra \Proj_{\geq 3} \pgs,\, H\, \Proj_{\leq 2}\pgs \ket \geq -\frac{1}{8} \|H_f^\frac12 \Proj_{\geq 3} \pgs\|^2
 - \frac{1}{16} \|(P-P_f)\Proj_{3} \pgs\|^2 - c\alpha^\frac83\, .
 $$
\end{lem}
%
%
\begin{proof}
The estimate $\|\Proj_{\geq 2} \pgs\| = \mathcal{O}(\alpha^\frac56)$ in Proposition~\ref{IPNE} yields
\begin{equation}\nonumber
\begin{split}
&  4\sqrt{\alpha} \Re\bra (P-P_f). \Aa \Proj_3\pgs,\, \Proj_2\pgs\ket \\
 & \geq -\frac{1}{16} \|(P-P_f) \Proj_3 \pgs\|^2 - c\alpha \|\Proj_2 \pgs\|^2 \\
 & \geq -\frac{1}{16} \|(P-P_f) \Proj_3 \pgs\|^2 -c\alpha^\frac83\, .
\end{split}
\end{equation}
Proposition~\ref{Nf-exp-lemma-1} implies $\|\Proj_1\pgs\|^2=\mathcal{O}(\alpha)$, therefore, with \cite[Lemma~A4]{GLL} we obtain
\begin{equation}\nonumber
\begin{split}
 & 2\alpha \Re \bra \Aa. \Aa \Proj_4 \pgs,\, \Proj_2 \pgs\ket
 + 2\alpha \Re \bra \Aa . \Aa \Proj_3 \pgs,\, \Proj_1\, \pgs\ket \\
 & \geq -\frac{1}{16} \|H_f^\frac12 \Proj_4 \pgs\|^2 -\frac{1}{16} \|H_f^\frac12 \Proj_3\pgs\|^2
 - c\alpha^2 \left( \|\Proj_2 \pgs\|^2 + \|\Proj_1 \pgs\|^2\right) \\
 & \geq  -\frac{1}{16} \|H_f^\frac12 \Proj_{\geq 3} \pgs\|^2 - c\alpha^3\, .
\end{split}
\end{equation}
Similarly we obtain
\begin{equation}\nonumber
\begin{split}
 2\sqrt{\alpha} \Re \bra\sigma .\Ba \Proj_3 \pgs,\, \Proj_2 \pgs\ket
 &  \geq -\frac{1}{16} \|H_f^\frac12 \Proj_3 \pgs\|^2 - c\alpha \|\Proj_2 \pgs\|^2 \\
 & \geq -\frac{1}{16} \|H_f^\frac12 \Proj_3 \pgs\|^2 -c\alpha^\frac83\, ,
\end{split}
\end{equation}

Since particle conserving terms $\bra \Proj_{\geq3} \pgs,\, (-\Delta -\frac{\alpha}{|x|} + H_f+P_f^2) \Proj_{\leq 2}\pgs\ket$ are zero, collecting the above inequalities yields the result.
\end{proof}

\begin{lem}\label{lem:7}
There exist $c$ and $c_1$ strictly positive such that
 \begin{equation*}
 \begin{split}
 & \bra \Proj_{\leq 2} u_\alpha\phi,\, H\, \Proj_{\leq 2}u_\alpha\phi \ket \\
 & \geq - \alpha |\gamma_0|^2\, \| \Guab \|_*^2 - \alpha^2 |\gamma_0|^2\, \|\Gdab\|_*^2
 + 2\alpha^2 |\gamma_0|^2 \|\Aa \Guab \|^2 -\frac{\alpha^2}{4} |\gamma_0|^2 \\
 & +\frac12 \| R_1\|_*^2 + \frac12\| R_2\|_*^2 + c_1 \alpha^2 |\gamma_2 - \gamma_1|^2
 + c_1 \alpha |\gamma_1 - \gamma_0|^2 - c\alpha^\frac83\, .
 \end{split}
 \end{equation*}
\end{lem}
%
%
\begin{proof}
Let
$$
 \bra \Proj_{\leq 2} u_\alpha\phi,\, H\, \Proj_{\leq 2} u_\alpha\phi\ket =: \sum_{i=1}^{10} Q_i\, ,
$$
where
\begin{flalign*}
 Q_1
 & : = \bra \Proj_{\leq 2} u_\alpha\phi, (-\Delta -\frac{\alpha}{|x|} + H_f + P_f^2) \,
 \Proj_{\leq 2} u_\alpha\phi \ket \\
 & = -\frac{\alpha^2}{4} |\gamma_0|^2 + \|\Proj_1 \phi\|_*^2 + \|\Proj_2\phi\|_*^2
 -\frac{\alpha^2}{4} (\|\Proj_1 \phi\|^2 + \|\Proj_2 \phi\|^2 ) &\\
 & \geq  -\frac{\alpha^2}{4} |\gamma_0|^2 + \|R_1\|_*^2 +\| R_2\|_*^2 +\alpha |\gamma_1|^2 \|\Guab\|_*^2 + \alpha^2 |\gamma_2|^2 \|\Guab\|_*^2 - c\alpha^3\, ,&
\end{flalign*}
\begin{flalign*}
 Q_2
 : = & 4\Re \bra \alpha \gamma_2 \Gdab,\, \sqrt{\alpha} P_f.\Ac \sqrt{\alpha}\gamma_1 \Guab\ket & \\
 & + 2\Re \bra \alpha \gamma_2 \Gdab,\, \sqrt{\alpha} \sigma .\Bc \sqrt{\alpha}\gamma_1 \Guab\ket & \\
 & + 2\Re \bra \alpha \gamma_2 \Gdab,\, \alpha \Ac . \Ac \gamma_0 \mab \vac \ket & \\
 = & 2\alpha^2 \Re \gamma_2 \bar\gamma_1 \bra \Gdab,\, (H_f+P_f^2)^{-1} 2 P_f.\Ac \Guab\ket_* &\\
 & + 2\alpha^2 \Re \gamma_2 \bar\gamma_1 \bra \Gdab,\, (H_f+P_f^2)^{-1}  \sigma.\Bc  \Guab\ket_* & \\
 & + 2\alpha^2 \Re \gamma_2\bar\gamma_0 \bra \Gdab,\, (H_f+P_f^2)^{-1} \Ac . \Ac \mab\vac\ket_*\, ,&
\end{flalign*}
\begin{flalign*}
 Q_3
 & : =  2\Re \bra \sqrt{\alpha}\gamma_1\Guab,\, \sqrt{\alpha}\sigma . \Bc \gamma_0 \mab \vac \ket \\
 & = -2\alpha \Re \gamma_1 \bar\gamma_0 \|\Guab\|_*^2\, ,&
\end{flalign*}
\begin{flalign*}
 Q_4
 & :=
 4\Re \bra R_2,\, \sqrt{\alpha} P_f.\Ac \sqrt{\alpha}\gamma_1 \Guab\ket
 + 2\Re\bra R_2,\, \sqrt{\alpha}\sigma.\Bc \sqrt{\alpha}\gamma_1\Guab\ket \\
 & + 2\Re\bra R_2,\, \alpha \Ac.\Ac \gamma_0 \mab \vac\ket\\
 & = 2\alpha \Re \bar\gamma_1 \bra R_2,\, (H_f+P_f^2)^{-1} 2 P_f.\Ac \Guab\ket_* &\\
 & \ \ \ + 2\alpha\Re \bar\gamma_1 \bra R_2,\, (H_f+P_f^2)^{-1}  \sigma.\Bc \Guab\ket_* &\\
 & \ \ \ + 2\alpha\Re\bar\gamma_0 \bra R_2, \, (H_f+P_f^2)^{-1} \Ac .\Ac \mab\vac\ket_*\, ,&
\end{flalign*}
\begin{flalign*}
 & Q_5 : = 4\sqrt{\alpha}\Re \bra P_f.\Aa R_2,\, R_1\ket + 2\sqrt{\alpha} \Re \bra \sigma\cdot \Ba R_2,\, R_1\ket \, ,&
\end{flalign*}
\begin{flalign*}
 & Q_6 : = 4 \Re\bra\alpha \gamma_2 \Gdab,\, \sqrt{\alpha}P_f. \Ac R_1 \ket
 + 2 \Re \bra \alpha \gamma_2 \Gdab,\, \sqrt{\alpha} \sigma . \Bc R_1 \ket\, , &
\end{flalign*}
\begin{flalign*}
 Q_7 &: = 2 \bra \alpha \Ac. \Aa \sqrt{\alpha} \gamma_1 \Guab,\sqrt{\alpha} \gamma_1 \Guab\ket &\\
 & = 2\alpha^2 |\gamma_1|^2 \| \Aa\Guab\|^2\, &
\end{flalign*}
\begin{flalign*}
 Q_8 & := 2 \Re \bra R_1,\, \sqrt{\alpha}\sigma\cdot \Bc \gamma_0\mab\vac \ket \, ,&
\end{flalign*}
\begin{flalign*}
 Q_9 & := 2 \bra \alpha \Ac.\Aa R_1,\, R_1\ket \!+\! 2 \bra \alpha\Ac.\Aa R_2,\, R_2\ket
 \!+\! 2 \bra \alpha \Ac.\Aa \alpha\gamma_2\Gd,\, \alpha\gamma_2\Gd\ket\, , &
\end{flalign*}
and
\begin{flalign*}
 Q_{10}  & := 4 \Re \bra \alpha \Ac.\Aa R_1, \sqrt{\alpha} \gamma_1 \Guab\ket + 4 \Re \bra \alpha \Ac.\Aa R_2,\, \alpha \gamma_2 \Gdab\ket &\\
 & \geq  - c\alpha^{\frac52} - c\alpha^3 |\gamma_2| \geq -c\alpha^\frac52 - c\alpha^3 |\gamma_2 - \gamma_1|^2\, , &
\end{flalign*}
where in the last inequalities we used that $\gamma_1$ is bounded and $\| R_1 \|_* = \mathcal{O}(\alpha)$ from Proposition~\ref{proposition:4}.

We have
\begin{flalign*}
& Q_2 \geq -2\alpha^2 \Re \gamma_2\bar\gamma_1 \|\Gdab\|_*^2 - c\alpha^2 |\gamma_0 - \gamma_1|\, .&
\end{flalign*}
Using $|\gamma_0 - \gamma_1|^2=\mathcal{O}(\alpha)$ from Proposition~\ref{proposition:4} and
$\bra R_2,\, \Gdab\ket_* =0$ yields
\begin{flalign*}
 Q_4 & = -2 \alpha\Re \bar\gamma_1 \bra R_2,\, \Gdab\ket_*  + 2 \alpha\Re (\bar\gamma_0 -\gamma_1)
\bra \Aa R_2,\, \Ac \mab\vac \ket &\\
& \geq -\epsilon \|H_f^\frac12 R_2\|^2 - c\alpha^2 |\gamma_0 - \gamma_1|^2
\geq -\epsilon \|R_2\|_*^2 - c\alpha^3 \, .&
\end{flalign*}
Proposition~\ref{IPNE} implies
\begin{flalign*}
 Q_5
 & \geq -c\sqrt{\alpha} \|H_f^\frac12 R_2\|^2 - c\sqrt{\alpha} \|H_f^\frac12 R_1\|^2
 -\epsilon \|H_f^\frac12 R_2\|^2 - \alpha \|R_1\|^2 & \\
 & \geq -2\epsilon \|R_2\|_*^2 - \epsilon \|R_1\|_*^2 - c^{2+\frac23}\, . &
\end{flalign*}
We also have
\begin{flalign*}
 Q_6 & \geq -4\alpha^\frac32 |\gamma_2|\, \|\Aa \Gdab\|\, \|P_f R_1\| - 2\alpha^\frac32 |\gamma_2|
 \, \|H_f^{-\frac12} \sigma . \Ba \Gdab\|\, \| H_f^\frac12 R_1\| & \\
 & \geq -\epsilon \|R_1\|_*^2 - c\alpha^3 |\gamma_2|^2\, .&
\end{flalign*}
Since $\bra \Guab,\, R_1\ket_*=0$, we get
\begin{flalign*}
& Q_8 = 0\, . &
\end{flalign*}
Therefore, collecting all above estimates implies
\begin{equation}\label{eq:rhs-22}
\begin{split}
 & \bra \Proj_{\leq 2} u_\alpha\phi,\, H\, \Proj_{\leq 2} u_\alpha\phi\ket \\
 & \geq \alpha (|\gamma_1|^2 - 2\Re \gamma_1\bar\gamma_0) \|\Guab\|_*^2
 + \alpha^2(|\gamma_2|^2 - 2\Re \gamma_2\bar\gamma_1) \|\Gdab\|_*^2
  -c\alpha^3 |\gamma_2|^2\\
 &  + 2\alpha^2 |\gamma_1|^2
 \| \Aa \Guab\|^2 -\frac{\alpha^2}{4} |\gamma_0|^2
 +\frac12 \|R_1\|_*^2 +\frac12\|R_2\|_*^2 - c\alpha^2|\gamma_0 - \gamma_1|
 - c\alpha^\frac83 \\
 & \geq - \alpha |\gamma_0|^2 \|\Guab\|_*^2 -\alpha^2 |\gamma_1|^2 \|\Gdab\|_*^2
 + 2\alpha^2 |\gamma_1|^2 \|\Aa\Guab\|^2 +\frac12 \|R_1\|_*^2 + \frac12\| R_2\|_*^2 \\
 &  +c|\gamma_2 - \gamma_1|^2 \alpha^2 + c|\gamma_0 -\gamma_1|^2\alpha
 - c\alpha^2|\gamma_0 - \gamma_1| - c\alpha^3 |\gamma_2|^2 -\frac{\alpha^2}{4}|\gamma_0|^2- c\alpha^\frac83\, .
\end{split}
\end{equation}
Since $\| \Proj_0 \G\| \leq \|\Proj_0 \pgs\| =1$, we have $|\gamma_0|\leq 1$ and thus we can write
$$
 -c\alpha^3 |\gamma_2|^2 \geq -c\alpha^3 |\gamma_2\!-\!\gamma_1|^2 - c\alpha^3|\gamma_1|^2
 \geq -c\alpha^3 |\gamma_2\!-\!\gamma_1|^2 - c\alpha^3|\gamma_1\!-\!\gamma_0| -c\alpha^3|\gamma_0|^2\, ,
$$
In additon, using the estimate $\gamma_0-\gamma_1=\mathcal{O}(\alpha^\frac12)$ (see proposition~\ref{proposition:4}) and estimating from below the second term in the right hand side of \eqref{eq:rhs-22} by
$$
 -\alpha^2 |\gamma_0|^2 \|\Gdab\|_*^2 - c\alpha^2|\gamma_0-\gamma_1|\, ,
$$
yields that for some constant $c_1>0$ independent of $\alpha$ we have
\begin{equation}\nonumber
\begin{split}
 & \bra \Proj_{\leq 2} u_\alpha\phi,\, H\, \Proj_{\leq 2} u_\alpha\phi\ket \\
 & \geq - \alpha |\gamma_0|^2 \|\Guab\|_*^2 -\alpha^2 |\gamma_0|^2 \|\Gdab\|_*^2
 + 2\alpha^2 |\gamma_0|^2 \|\Aa\Guab\|^2 -\frac{\alpha^2}{4}|\gamma_0|^2 \\
 & \ \ \ +\frac12 \|R_1\|_*^2 + \frac12\| R_2\|_*^2
  \ \ \ +c_1 \alpha^2 |\gamma_2 - \gamma_1|^2  + c_1 \alpha |\gamma_0 -\gamma_1|^2
 - c\alpha^\frac83\, .
\end{split}
\end{equation}
which concludes the proof.
\end{proof}

\begin{lem}\label{lem:8}
There exist $c>0$ such that
 \begin{equation}\label{eq:statement-a-lem:8}
 \begin{split}
 &   \bra [(P-P_f)^2 -\frac{\alpha}{|x|} + H_f] \Proj_2 \G, \, \Proj_2 \G \ket \\
 &  \geq \frac{1}{2(1+\Lambda)} \alpha^2|\beta_2|^2\, \|\Gd (g)\|_*^2
   + \frac{1}{2(1 + \Lambda)} \|L_2\|_*^2
   +\frac12 \| (P-P_f) \Proj_2 \G \|^2 - c\alpha^3\, .
 \end{split}
 \end{equation}
For $\Lambda := \sup_{\zeta(r)\neq 0} |r|$, where $\zeta$ is the ultraviolet cutoff, if $|\beta_2| < 8(1+\Lambda)$, then for any $\upsilon>0$, there exists $\delta_5$,
$\frac{1}{2(1 + \Lambda)}>\delta_5>0$, $c>0$ and $\delta_6>0$ such that
\begin{equation}\label{eq:statement-b-lem:8}
 \begin{split}
 &  \bra [(P-P_f)^2 -\frac{\alpha}{|x|} + H_f] \Proj_2 \G, \,  \Proj_2 \G \ket \\
 &  \geq \alpha^2|\beta_2|^2\, \|\Gd (g)\|_*^2 + \delta_5 \|L_2\|_*^2 - \frac{\upsilon}{2} \|\nabla g\|^2
   +\delta_6 \|P \L_2\|^2  - c\alpha^3\, .
 \end{split}
 \end{equation}
\end{lem}
%
%
\begin{proof}
According to \eqref{eq:app-3} and \eqref{eq:app-6}
we have
\begin{equation}\nonumber
\begin{split}
 & |\bra P.P_f \alpha \beta_2 \Gd(g),\, \alpha \beta_2 \Gd(g)\ket|
  \leq c \alpha^2 | \beta_2 |^2 \|g\|\, \|\nabla g\|\, .
\end{split}
\end{equation}
Since $\|\nabla g\| \leq c\alpha$ (Proposition~\ref{proposition:4}) we thus obtain in the case $|\beta_2| \leq 8(1+\Lambda)$
\begin{equation}\label{eq:eq:eq:eq}
 |\bra P.P_f \alpha \beta_2 \Gd(g),\, \alpha \beta_2 \Gd(g)\ket| \leq c \alpha^3\, .
\end{equation}
The rest of the proof follows exactly the proof of Lemma~\ref{lem:2}, with indices $2$ instead of $1$ everywhere, and using \eqref{eq:eq:eq:eq} instead of $\bra P.P_f T_1,\, T_1\ket=0$ in the inequalities \eqref{eq:star-2} and \eqref{eq:5.4-7.8}
\end{proof}


\begin{lem}\label{lem:9}
Pick $\delta_4$ as given in Lemma~\ref{lem:2},
then for all $\epsilon>0$, there exists $\alpha_0>0$ and $c>0$ such that for all $0<\alpha<\alpha_0$,
\begin{equation*}
\begin{split}
 & - 4\Re \sqrt{\alpha} \bra \Proj_{\leq 2} \G,\, P.\Aa \Proj_{\leq 2} \G \ket  \\
 & \geq -\epsilon\|\nabla g\|^2 - \epsilon\|\L_2 \|_*^2 - \epsilon\|L_1\|_*^2 - \epsilon \delta_4 M(\Proj_1 \G)\\
 & \ \ \ - c\alpha^3 |\beta_2|^2 \|g\|^2
 - c\alpha^\frac52 |\beta_2|\, \|g\| -c\alpha^3\, ,
\end{split}
\end{equation*}
where $M(\Proj_1 \G )$ is defined in \eqref{eq:def-M}.
\end{lem}
%
%
\begin{proof}
We have two terms to estimate in $-4\Re\sqrt{\alpha} \bra P.\Aa \Proj_{\leq 2} \G,\, \Proj_{\leq 2} \G  \ket$. The term $-4\Re \sqrt{\alpha} \bra P.\Aa \Proj_1 \G,\, \Proj_0 \G\ket $ was already estimated in \eqref{eq:diese-5}, and we obtained
\begin{equation}\label{eq:diese-0}
 -4\Re \sqrt{\alpha} \bra P.\Aa \Proj_1 \G,\, \Proj_0 \G\ket
 \geq
 - c\sqrt{\alpha} \|\nabla g\|^2 - c\sqrt{\alpha} \|L_1\|_*^2\, .
\end{equation}

We next bound the term $- 4\Re \sqrt{\alpha} \bra P.\Aa \Proj_2 \G,\, \Proj_1 \G \ket $.

We consider two cases

$\bullet$ If $|\beta_1|\geq 8(1+\Lambda)$, we write
\begin{equation}\label{eq:lem7.9-0}
\begin{split}
 & - 4\Re \sqrt{\alpha} \bra P.\Aa \Proj_2 \G,\, \Proj_1 \G \ket \\
 & = - 4\Re \sqrt{\alpha} \bra (P-P_f).\Aa \Proj_2 \G,\, \Proj_1\G\ket
 - 4 \Re \sqrt{\alpha} \bra P_f.\Aa \Proj_2 G,\, \Proj_1 G\ket \\
 & = -4\Re \sqrt{\alpha} \bra \Aa \Proj_2 G,\, (P-P_f) \Proj_1 G\ket
  - 4\Re \sqrt{\alpha} \bra \Aa \Proj_2 G,\, P_f \L_1\ket \\
 & \ \ \
 - 4\Re\sqrt{\alpha} \bra \Aa \alpha \beta_2 \Gd(g),\, P_f \sqrt{\alpha} \beta_1 \Gu(g) \ket
 - 4\Re\sqrt{\alpha} \bra \Aa L_2,\, P_f\sqrt{\alpha} \beta_1 \Gu(g)\ket \\
 & \geq -c\alpha^\frac12 \|H_f^\frac12 \Proj_2 \G\|\, \|(P-P_f) \Proj_1 \G\|
        -c\alpha^\frac12 \|H_f^\frac12 \Proj_2 \G\|\, \| P_f \L_1\| \\
 & \ \ \ - c\alpha^2 |\beta_2|\, \|g\| \ |\beta_1|\, \|g\| - c\alpha \|\Aa L_2\| \ |\beta_1|\, \|g\|\, .
\end{split}
\end{equation}
First we have, according to Proposition~\ref{proposition:4}
\begin{equation}\label{eq:lem7.9-1}
 \| H_f^\frac12 \Proj_2 \G \| \leq \|H_f^\frac12 \Proj_{\geq 2} \pgs \| \leq c\alpha\, .
\end{equation}
Second, according also to Proposition~\ref{proposition:4}, we have $|\beta_1 -1| \|g\| \leq c \alpha^\frac12$, for $c$ independent of $\beta_1$ and $\alpha$. Since we are in the case where $|\beta_1|\geq 8(1+\Lambda)$, it implies that there exists $c>0$ independent of $\alpha$ and $\beta_1$ such that
\begin{equation}\label{eq:lem7.9-2}
 |\beta_1| \, \|g\| \leq c\alpha^\frac12\, .
\end{equation}
Thus, \eqref{eq:lem7.9-1} and \eqref{eq:lem7.9-2} give with \eqref{eq:lem7.9-0}, that for all $\epsilon>0$, there exists $c$ and $\alpha_0>0$ such that for all $0<\alpha<\alpha_0$ we have
\begin{equation}\label{eq:lem7.9-3}
\begin{split}
 & -4\Re\sqrt{\alpha} \bra P.\Aa \Proj_2\G,\, \Proj_1\G  \\
 & \geq - c\alpha^\frac32 M(\Proj_1 \G) - c\alpha^\frac32 \|L_1\|_*^2
        - c\alpha^\frac52 |\beta_2|\, \|g\| - c\alpha^2 \|L_2\|_*^2 \\
 & \geq - c\alpha^3 - \epsilon \delta_4 M(\Proj_1 \G) - \epsilon \|L_1\|_*^2 - \epsilon\|L_2\|_*^2 - c\alpha^\frac52 |\beta_2|\, \|g\|\, .
\end{split}
\end{equation}

$\bullet$ In the case $| \beta_1| < 8(1+\Lambda)$, we have with \eqref{eq:app-5} of Lemma~\ref{lem:appendix-6}
\begin{equation}\label{eq:lem7.9-4}
\begin{split}
 & -4\sqrt{\alpha}\Re \bra P. \Aa \Proj_2 \G,\, \Proj_1 \G\ket \\
 &
 = - 4\sqrt{\alpha} \Re \bra \Aa \Proj_2 \G, \, P L_1
 \ket
   - 4\sqrt{\alpha} \Re \bra \Aa \Proj_2 \G, \, P \sqrt{\alpha}\beta_1 \Gu(g)\ket \\
 & \geq - c \alpha^\frac12 \|H_f^\frac12 \Proj_2 \G \| \, \|P \L_1\|
        - c \alpha \|H_f^\frac12 \Proj_2 \G \| |\beta_1| \|\nabla g\| \\
\end{split}
\end{equation}
With \eqref{eq:lem7.9-1} and the fact that $|\beta_1|$ is bounded in this case, \eqref{eq:lem7.9-4} implies that for any $\epsilon>0$, there exists $c$ and $\alpha_0>0$ such that for all $0<\alpha<\alpha_0$, we have
\begin{equation}\label{eq:lem7.9-5}
\begin{split}
 & -4\sqrt{\alpha}\Re \bra P. \Aa \Proj_2 \G,\, \Proj_1 \G\ket
 \geq - c\alpha^\frac32 \|P L_1\| - c\alpha^2 \|\nabla g\| \\
 & \geq  - \epsilon \|P L_1\|^2 - c\alpha \|\nabla g\|^2 - c\alpha^3
 \geq  - \epsilon \delta_4 M(\Proj_1\G) - c\alpha \|\nabla g\|^2 - c\alpha^3 \, .
\end{split}
\end{equation}
Summarizing \eqref{eq:lem7.9-3} and \eqref{eq:lem7.9-5} concludes the proof.
\end{proof}

\begin{lem}\label{lem:10}
For all $\epsilon>0$, there exist $\alpha_0>0$ and $c>0$ such that for all $0<\alpha<\alpha_0$
\begin{equation*}
\begin{split}
 & \Re \bra  (4\sqrt{\alpha} P_f.\Aa + 2\alpha (\Aa)^2 + 2\sqrt{\alpha} \sigma.\Ba) \Proj_{\leq 2} \G, \,
  \Proj_{\leq 2} \G \ket \\
 & \geq -2\alpha\Re\beta_1 \|\Gu (g) \|_*^2 - 2\alpha^2\Re \beta_2\bar\beta_1 \|\Gd (g) \|_*^2
 - c\alpha^2 |\beta_2|\, |1-\beta_1|\, \| g \|^2\\
 & - c\alpha^2 |1-\beta_1|^2\, \|g\|^2 - c\alpha^3 |\beta_2|^2 \|g\|^2 -\epsilon \|H_f^\frac12 \L_2 \|^2
 - \epsilon \| \L_1 \|_*^2 - c\alpha^\frac83\, .
\end{split}
\end{equation*}
\end{lem}
%
%
\begin{proof}
The estimate of
$$
 \Re \bra (4\sqrt{\alpha} P_f.\Aa + 2\alpha (\Aa)^2 + 2\sqrt{\alpha} \sigma . \Ba) \Proj_2 \G,\, \Proj_{\leq 2} \G\ket
$$
is similar to the estimate of
$$
 \Re \bra (4\sqrt{\alpha} P_f.\Aa + 2\alpha (\Aa)^2 + 2\sqrt{\alpha} \sigma . \Ba) \Proj_2 u\alpha\phi,\,
 \Proj_{\leq 2} u_\alpha\phi\ket
$$
done in the proof of Lemma~\ref{lem:7} (see the terms $Q_2$, $Q_4$, $Q_5$ and $Q_6$ therein), and yields
\begin{equation}\label{eq:136-a}
\begin{split}
 & \Re \bra
 (4\sqrt{\alpha} P_f.\Aa + 2\alpha (\Aa)^2 + 2\sqrt{\alpha} \sigma.\Ba) \Proj_2 \G,\, \Proj_{\leq 2} \G \ket \\
 & \geq - 2\alpha^2\Re \beta_2 \bar\beta_1 \|\Gd (g) \|_*^2
 - c\alpha^2 |\beta_2|\, |1-\beta_1|\, \| g \|^2\\
 & - c\alpha^2 |1-\beta_1|^2\, \|g\|^2 - c\alpha^3 |\beta_2|^2 \|g\|^2 -\epsilon \|H_f^\frac12 \L_2 \|^2
 - \epsilon \| \L_1 \|_*^2 - c\alpha^\frac83\, .
\end{split}
\end{equation}

We also have
\begin{equation}\label{eq:137-b}
\begin{split}
 & \Re \bra (4\sqrt{\alpha} P_f.\Aa + 2\alpha (\Aa)^2 + 2\sqrt{\alpha} \sigma .\Ba)\Proj_1 \G,\, \Proj_{\leq 2} \G\ket \\
 & = \Re \bra 2\sqrt{\alpha} \sigma . \Ba \Proj_1 \G,\, \Proj_0 \G\ket
  = -2\alpha \Re\beta_1 \|\Gu(g)\|_*^2\, ,
\end{split}
\end{equation}
using for the last equality $\bra \L_1,\, \Gu(g)\ket_*=0$.

Inequalities \eqref{eq:136-a} and \eqref{eq:137-b} allow to conclude the proof.
\end{proof}

In Lemmata~\ref{lem:8}, \ref{lem:9} and \ref{lem:10}, we estimated all terms occurring in $\bra \Proj_{\leq 2}\G ,\, H \Proj_{\leq 2} \G \ket$. We can therefore write the following result
\begin{lem}\label{lem:11} Let $\delta_3$, $\delta_4$ be as in Lemma~\ref{lem:2} and $\delta_5$
be as in Lemma~\ref{lem:8}. There exist $\delta_1>0$, $\delta_2>0$, such that for all $\epsilon>0$ small enough,
there exist $\alpha_0>0$, $c_0>0$, $c>0$ such that for all $0<\alpha<\alpha_0$ we have
\begin{equation*}
\begin{split}
 & \bra \Proj_{ \leq 2} \G ,\, H \Proj_{ \leq 2} \G \ket \\
 & \geq -\alpha \|\Gu (g)\|_*^2 + 2\alpha^2 \|\Aa\Gu (g)\|^2 - \alpha^2 \| \Gd (g)\|_*^2 -\frac{\alpha^2}{4}
 \| g\|^2 \\
 & + \alpha^2 c_0 |1-\beta_2|^2 \|g\|^2 + \alpha c_0 |1 - \beta_1|^2 \|g\|^2 + \frac{\delta_3}{2} \|L_1\|_*^2 + \frac{\delta_5}{2} \|L_2\|_*^2 \\
 & + \frac{\delta_1}{2} \|\nabla g\|^2 + \frac{\delta_2}{2} \alpha^2 \|g\|^2
 + (1-\epsilon) \delta_4 M(\Proj_1 \G)
 -c\alpha^\frac83\, ,
\end{split}
\end{equation*}
where $M(\Proj_1 \G)$ is defined
in \eqref{eq:def-M}.
\end{lem}
%
%
\begin{proof}
We first write
\begin{equation}\label{eq:diese-7}
\begin{split}
 & \bra H \Proj_{\leq 2} \G,\, \Proj_{\leq 2}\G \ket \\
 & =  \bra H \Proj_0 \G,\, \Proj_0 \G\ket
 + \bra H \Proj_1 \G,\, \Proj_1 \G\ket + \bra H \Proj_2 \G,\, \Proj_2 \G\ket
  - 4\sqrt{\alpha}\Re \bra P.\Aa \Proj_{\leq 2} \G,\,\Proj_{\leq 2}\G \ket \\
 & + \Re \bra 4\sqrt{\alpha} P_f.\Aa + 2\alpha (\Aa)^2 +2\sqrt{\alpha} \sigma.\Ba) \Proj_{\leq 2} \G,\, \Proj_{\leq 2} \G\ket \, .
\end{split}
\end{equation}
Since $g\perp u_\alpha$, there exists $\delta_1>0$ and $\delta_2>0$ such that
\begin{equation}\label{eq:diese-8}
  \bra H \Proj_0 \G,\, \Proj_0 \G\ket = \bra (-\Delta -\frac{\alpha}{|x|}) g,\, g\ket
  \geq -\frac{\alpha^2}{4} \| g\|^2 + \delta_1 \|\nabla g\|^2 + \delta_2 \alpha^2 \| g\|^2 \, .
\end{equation}

Before we proceed further, we introduce two notations. Let $d$ and $\tilde{d}$ be defined by
$$
 d =
 \left\{
  \begin{array}{ll}
    1 & \mbox{if } |\beta_1| < 8(1+\Lambda) \\
    \frac{1}{2(1+\Lambda)} & \mbox{if } |\beta_1|\geq 8(1+\Lambda)
  \end{array}
 \right.\, ,
$$
and
$$
\tilde{d} =
 \left\{
  \begin{array}{ll}
    1 & \mbox{if } |\beta_2| < 8(1+\Lambda) \\
    \frac{1}{2(1+\Lambda)} & \mbox{if } |\beta_2|\geq 8(1+\Lambda)
  \end{array}
 \right.\, .
$$
With these definitions the statements \eqref{eq:itc}-\eqref{eq:itc-2} of Lemma~\ref{lem:2}
can be summarized as: for any $\xi>0$, there exists $\delta_3\in (0,\frac{1}{2(1+\Lambda)})$,
$\delta_4\in(0,\frac12)$ and $c>0$ such that
\begin{equation}\label{added:eq:statement-lem:1}
\begin{split}
 & \bra [(P-P_f)^2 -\frac{\alpha}{|x|} + H_f]\Proj_1 \G,\, \Proj_1\G \ket\\
 & \geq d\alpha |\beta_1|^2 \|\Gu(g)\|_*^2 + \delta_3\|\L_1\|_*^2
 - \frac{\xi}{4}\| \nabla g\|^2 + \delta_4 M(\Proj_1 \G) -c\alpha^3
\end{split}
\end{equation}
and the statements \eqref{eq:statement-a-lem:8}-\eqref{eq:statement-b-lem:8} of Lemma~\ref{lem:8}
can be summarized as: for any $\upsilon>0$, there exists $\delta_5\in (0,\frac{1}{2(1+\Lambda)})$ and $c>0$ such that
\begin{equation}\label{added:eq:statement-lem:2}
\begin{split}
 & \bra [(P-P_f)^2 -\frac{\alpha}{|x|} + H_f]\Proj_2 \G,\, \Proj_2\G \ket\\
 & \geq \tilde{d}\alpha^2 |\beta_2|^2 \|\Gd(g)\|_*^2 + \delta_5\|\L_2\|_*^2
 - \frac{\upsilon}{4}\| \nabla g\|^2 -c\alpha^3
\end{split}
\end{equation}

To estimate the second, third, fourth and fifth term in \eqref{eq:diese-7}, we use respectively \eqref{added:eq:statement-lem:1}, \eqref{added:eq:statement-lem:2}, Lemma~\ref{lem:9} and Lemma~\ref{lem:10}. Together with \eqref{eq:diese-8}, this implies, for $\xi$ and $\upsilon$ small enough with respect to $\delta_1$
\begin{equation}\label{eq:diese-9}
\begin{split}
 & \bra H \Proj_{\leq 2} \G,\, \Proj_{\leq 2}\G \ket \\
 & \geq
 \alpha^2 \tilde{d} |\beta_2|^2 \|\Gd(g)\|_*^2
 + \frac{\delta_3}{2} \|\L_1\|_*^2 +\frac{\delta_5}{2} \|\L_2 \|_*^2
 + (1-\epsilon) \delta_4 M(\Proj_1 \G) \\
 &  -\frac{\alpha^2}{4} \|g\|^2
 +\frac{\delta_1}{2} \|\nabla g\|^2 +\delta_2 \alpha^2 \|g\|^2
 + d \alpha|\beta_1|^2 \|\Gu(g)\|_*^2 - c\alpha^\frac52 |\beta_2| \|g\| \\
 & + 2\alpha^2 |\beta_1|^2 \|\Aa \Gu(g)\|^2 - 2\alpha\Re \beta_1 \|\Gu(g)\|_*^2 - 2\alpha^2 \Re \beta_2 \bar\beta_1 \|\Gd(g)\|_*^2 \\
 & - c\alpha^2 |\beta_2|\, |1-\beta_1|\, \|g\|^2 - c\alpha^2 |1-\beta_1|^2 \|g\|^2 - c\alpha^3 |\beta_2|^2 \|g\|^2 - c\alpha^\frac83\, .
\end{split}
\end{equation}
The sum of the terms $d \alpha|\beta_1|^2 \|\Gu(g)\|_*^2$ and $- 2\alpha\Re \beta_1 \|\Gu(g)\|_*^2$ in \eqref{eq:diese-9} are estimated as
\begin{equation}\label{eq:diese-10}
 d \alpha|\beta_1|^2 \|\Gu(g)\|_*^2 - 2\alpha\Re \beta_1 \|\Gu(g)\|_*^2 \geq -\alpha \|\Gu(g)\|_*^2
 + c_0 \alpha |1-\beta_1|^2 \|g\|^2
\end{equation}

The term $ 2\alpha^2 |\beta_1|^2 \|\Aa \Gu(g)\|^2$ in \eqref{eq:diese-9} equals
\begin{equation}\nonumber
 2\alpha^2 |\beta_1|^2 \|\Aa \Gu(g)\|^2 = 2\alpha^2 \|\Aa \Gu(g)\|^2 + 2\alpha^2 (|\beta_1|^2 -1) \|\Aa \Gu(g)\|^2\, .
\end{equation}
If $|\beta_1|>1$, this term is larger than $2\alpha^2 \|\Aa(g)\|^2$,
and if $|\beta_1| \leq 1$, it is bounded below by
\begin{equation}\nonumber
\begin{split}
 2\alpha^2 |\beta_1|^2 \|\Aa \Gu(g)\|^2
 & \geq 2\alpha^2 \|\Aa \Gu(g)\|^2 - c|1-\beta_1|\alpha^2 \|g\|^2 \\
 & \geq  2\alpha^2 \|\Aa \Gu(g)\|^2  -\epsilon\alpha^2\|g\|^2 - c\alpha^2 |1-\beta_1|^2 \|g\|^2 \\
 & \geq  2\alpha^2 \|\Aa \Gu(g)\|^2  -\epsilon\alpha^2\|g\|^2 - c\alpha^3\, ,
\end{split}
\end{equation}
using from Proposition~\ref{proposition:4} that $|1-\beta_1|\, \|g\| < c\alpha^\frac12$.

Thus one always have
\begin{equation}\label{eq:diese-11}
\begin{split}
 2\alpha^2 |\beta_1|^2 \|\Aa \Gu(g)\|^2
 & \geq  2\alpha^2 \|\Aa \Gu(g)\|^2  -\epsilon\alpha^2\|g\|^2 - c\alpha^3 \, .
\end{split}
\end{equation}

Since $|1-\beta_1|\, \|g\| < c\alpha^\frac12$, for the term $- c\alpha^2 |\beta_2|\, |1-\beta_1|\, \|g\|^2$ in \eqref{eq:diese-9} we have
\begin{equation}\label{eq:diese-12}
- c\alpha^2 |\beta_2|\, |1-\beta_1|\, \|g\|^2  \geq -c\alpha^\frac52 |\beta_2|\, \|g\|
\end{equation}

We next estimate the terms $\alpha^2 \tilde{d} |\beta_2|^2 \|\Gd(g)\|_*^2$
and  $- 2\alpha^2 \Re \beta_2 \bar\beta_1 \|\Gd(g)\|_*^2$ in
\eqref{eq:diese-9} with similar arguments as in \eqref{eq:star-12.1-geq}
for $|\beta_2|\geq 8(1+\Lambda)$ and as in \eqref{eq:star-12.1-<}
for $|\beta_2|\ < 8(1+\Lambda)$. We obtain

\begin{equation}\label{eq:diese-13}
\begin{split}
&  \alpha^2 \tilde{d} |\beta_2|^2 \|\Gd(g)\|_*^2 - 2\alpha^2 \Re \beta_2 \bar\beta_1 \|\Gd(g)\|_*^2\\
& = \alpha^2 \tilde{d} |\beta_2|^2 \|\Gd(g)\|_*^2
- 2\alpha^2 \Re\beta_2 \|\Gd(g)\|_*^2 + 2\alpha^2 \Re \beta_2 (\beta_1 -1) \|\Gd(g)\|_*^2
 \\
& \geq -\alpha^2 \|\Gd(g)\|_*^2 + c_0 \alpha^2  |1-\beta_2|^2 \|g\|^2 - c\alpha^2 |\beta_2|\, |1-\beta_1|\, \|g\|^2 \\
& \geq -\alpha^2 \|\Gd(g)\|_*^2 + c_0 \alpha^2  |1-\beta_2|^2 \|g\|^2 - c\alpha^\frac52 |\beta_2|\, \|g\|\, .
\end{split}
\end{equation}
Using Inequalities \eqref{eq:diese-12} and \eqref{eq:diese-13} we thus obtain
\begin{equation}\label{eq:diese-14}
\begin{split}
& \alpha^2 \tilde{d} |\beta_2|^2 \|\Gd(g)\|_*^2- 2\alpha^2 \Re \beta_2 \bar\beta_1 \|\Gd(g)\|_*^2 + \frac{\delta_2}{2} \|g\|^2
 \\
  & - c\alpha^2 |\beta_2| |1\!-\!\beta_1|\, \|g\|^2
-c\alpha^3 |\beta_2|^2 \|g\|^2 \\
& \geq -\alpha^2 \|\Gd(g)\|_*^2 + c_0\alpha^2 |1-\beta_2|^2 \|g\|^2 + \frac{\delta_2}{2} \|g\|^2 -c\alpha^\frac52 |\beta_2| \|g\| -c\alpha^3 |\beta_2|^2 \|g\|^2 \\
& \geq -\alpha^2 \|\Gd(g)\|_*^2 + \frac{c_0}{2} \alpha^2 |1-\beta_2|^2 \|g\|^2 \, .
\end{split}
\end{equation}

Collecting inequalities \eqref{eq:diese-9}, \eqref{eq:diese-10}, \eqref{eq:diese-11} and \eqref{eq:diese-14} concludes the proof.
\end{proof}

\begin{lem}\label{lem:12}
For all $\epsilon>0$, there exist $\alpha_0>0$ and $c>0$ such that for all $0 < \alpha < \alpha_0$,
 \begin{equation*}
  2\Re \bra \Proj_{\leq 2} \G,\, H \Proj_{\leq 2} u_\alpha\phi \ket
  \geq -\epsilon \| R_1 \|_*^2 - \epsilon \| \L_1 \|_*^2 -c\alpha^3 \, .
 \end{equation*}
\end{lem}
%
%
\begin{proof}
From Proposition~\ref{proposition:4}, we get
\begin{equation}\label{eq:diese-15}
\Re \bra P.P_f \Proj_2 u_\alpha\phi,\, \Proj_2 \G\ket \geq -c \|P u_\alpha\|\, \|H_f^\frac12 \Proj_2\phi\|\, \|H_f^\frac12\Proj_2 \G\| \geq -c\alpha^3\, .
\end{equation}

Since $\bra P.P_f \Guab,\, \Gu(g)\ket = 0$ (see Lemma~\ref{lem:appendix-3}), and using again Proposition~\ref{proposition:4},  we have
\begin{equation}\label{eq:diese-16}
\begin{split}
 & 4 \Re \bra P.P_f \Proj_1 u_\alpha\phi,\, \Proj_1 \G\ket \\
 & =
 4\Re \bra P.P_f (\sqrt{\alpha}\gamma_1 \Guab u_\alpha + R_1 u_\alpha),\, \sqrt{\alpha} \beta_1 \Gu(g) +\L_1\ket \\
 & \geq -c\alpha \|H_f^\frac12 R_1\| \sqrt{\alpha} |\beta_1| \,\|g\| - c\alpha \|H_f^\frac12 R_1\|\, \|H_f^\frac12 \L_1\| - c\alpha^\frac32 |\gamma_1| \|H_f^\frac12 \L_1\| \\
 & \geq - \epsilon \|R_1\|_*^2 - c\alpha^3 |\beta_1|^2 \|g\|^2 - c\alpha^3 - \epsilon \|\L_1\|_*^2  \\
 & \geq -\epsilon\|R_1\|_*^2 - \epsilon\| \L_1\|_*^2 - c\alpha^3\, .
\end{split}
\end{equation}
With the photon number bound in Proposition~\ref{Nf-exp-lemma-1} and Proposition~\ref{proposition:4}, we obtain
\begin{equation}\label{eq:diese-17}
 -4\sqrt{\alpha}\Re \bra P.\Aa \Proj_2 u_\alpha \phi,\, \Proj_1 \G \ket \geq - c\sqrt{\alpha} \|P u_\alpha\|\, \|\Aa\Proj_2 \phi\|\, \|\Proj_1 \G\| \geq -c\alpha^3\, .
\end{equation}
Similarly we have
\begin{equation}\label{eq:diese-18}
 -4\sqrt{\alpha}\Re \bra P.\Aa \Proj_2 \G ,\, \Proj_1 u_\alpha \phi \ket \geq - c\sqrt{\alpha} \|P u_\alpha\|\, \|\Aa\Proj_2 \G \|\, \|\Proj_1 \phi\| \geq -c\alpha^3\, .
\end{equation}
Using $\bra \Guab, \, P.\Ac g\ket =0$ from Lemma~\ref{lem:appendix-4}, boundedness of $\gamma_1$ from Proposition~\ref{proposition:4}, and $\|g\|\leq 1$ implies
\begin{equation}\label{eq:diese-19}
\begin{split}
 &-4\sqrt{\alpha}\Re \bra P.\Aa (\sqrt{\alpha} \gamma_1 \Guab  + R_1)u_\alpha,\, g\ket \\
 & \geq -c \sqrt{\alpha} \|P u_\alpha\|\, \|H_f^\frac12 R_1\|\, \|g\|
 \geq -c\alpha^3 -\epsilon \|H_f^\frac12 R_1\|^2\, .
\end{split}
\end{equation}
Eventually, with similar arguments we have
\begin{equation}\label{eq:diese-20}
\begin{split}
 & -4\sqrt{\alpha}\Re \bra P.\Aa (\sqrt{\alpha} \beta_1 \Gu(g)  + \L_1)u_\alpha,\, \gamma_0 \mab \vac u_\alpha\ket \\
 & \geq -c\alpha^\frac32|\gamma_0|\, \|H_f^\frac12 \L_1\| \geq -\epsilon \|H_f^\frac12 \L_1\|^2 - c\alpha^3\, .
\end{split}
\end{equation}
Collecting\eqref{eq:diese-15}-\eqref{eq:diese-20} concludes the proof.
\end{proof}
We are now ready to write the proof of Proposition~\ref{prop:main-cmain}.
%
%

\textit{Proof of Proposition~\ref{prop:main-cmain}}.
Collecting the results of Lemmata~\ref{lem:5}, \ref{lem:6}, \ref{lem:7}, \ref{lem:11} and \ref{lem:12} yields
\begin{equation}\label{eq:diese-21}
\begin{split}
& \bra H \, \pgs,\, \pgs\ket \\
& \geq -\alpha |\gamma_0|^2 \|\Guab\|_*^2 - \alpha^2 |\gamma_0|^2 \|\Gdab\|_*^2 +2\alpha^2 |\gamma_0|^2 \|\Aa \Guab\|^2 -\frac{\alpha^2}{4} |\gamma_0|^2 \\
& \ \ \ -\alpha \|\Gu(g)\|_*^2 -\alpha^2 \|\Gd(g)\|_*^2 + 2\alpha^2 \|\Aa \Gu(g)\|^2 -\frac{\alpha^4}{4} \|g\|^2
+ \frac14 \|H_f^\frac12 \Proj_{\geq 3} \pgs\|^2 \\
& \ \ \  + \frac18 \|(P-P_f) \Proj_{\geq 3}\pgs\|^2 +\frac12 \| R_1\|_*^2 +\frac12 \| R_2\|_*^2 +\frac{\delta_3}{4} \| \L_1\|_*^2 + \frac{\delta_5}{4} \|\L_2\|_*^2 \\
&
+\frac{\delta_1}{2} \|\nabla g\|^2 + \frac{\delta_2}{2} \alpha^2 \| g\|^2 + c_1 \alpha^2 |\gamma_2 - \gamma_1|^2 + c_1\alpha |\gamma_1 - \gamma_0| + c_0 \alpha^2 |1-\beta_2|^2 \|g\|^2 \\
& + c_0 \alpha |1-\beta_1|^2 \|g\|^2  + (1-\epsilon)\delta_4 M(\Proj_1 \G) -c\alpha^\frac83.
\end{split}
\end{equation}
According to \cite[eq.14, eq.55]{BV4} we have
\begin{equation}\label{eq:diese-21-bis}
 -\alpha \|\Guab\|_*^2 -\alpha^2 \|\Gdab\|_*^2 + 2\alpha^2 \|\Aa \Guab\|^2\, = \Sigma_0 + \mathcal{O}(\alpha^3).
\end{equation}
Similarly, using Lemma~\ref{lem:appendix-6} yields
\begin{equation}\label{eq:diese-22}
\begin{split}
  -\alpha \|\Gu(g)\|_*^2 -\alpha^2 \|\Gd(g)\|_*^2 + 2\alpha^2 \|\Aa \Gu(g)\|^2
 =  \Sigma_0 \|g\|^2 +\mathcal{O}(\alpha^3)\, .
\end{split}
\end{equation}
By the normalization condition $\|\Proj_0 \pgs\|=|\gamma_0|^2 + \|g\|^2 = 1$, we get
\begin{equation}\label{eq:diese-23}
 -\frac{\alpha^2}{4} (|\gamma_0|^2 + \|g\|^2) = -\frac{\alpha^2}{4} = -\frac{\alpha^2}{4} \|\pgs\|^2 (1+\mathcal{O}(\alpha))
 = -\frac{\alpha^2}{4} \|\pgs\|^2 + \mathcal{O}(\alpha^3)\, ,
\end{equation}
since $\|\pgs\|^2 =1+\mathcal{O}(\alpha)$ from Proposition~\ref{Nf-exp-lemma-1}

Inserting \eqref{eq:diese-21-bis}-\eqref{eq:diese-23} in \eqref{eq:diese-21} concludes the proof of Proposition~\ref{prop:main-cmain}.



\section{Proof of the main theorem}\label{S4}

\subsection{Decomposition of the ground state}\label{S-decomposition}

Let as before $\pgs$ be a ground state of $H$ with the normalization condition $\|\Proj_0 \pgs\|=1$.

To prove the main theorem we will use a splitting of $\pgs$ slightly different from the one used in Section~\ref{S-new-7}.

We write $\pgs$ as in \eqref{*-0}, i.e.
\begin{equation}\label{eq:decomposition-again-1}
 \pgs = u_\alpha\phi + \G\, .
\end{equation}

As in \eqref{def:Guab-0}-\eqref{def:Guab-2}, we also define $\mab$, $\gamma_0$, $\gamma_1$, $\Guab$ and $R_1$, which yields the decomposition
\begin{equation}\label{eq:decomposition-again-3}
\begin{split}
 & \Proj_0 \phi = \gamma_0 \mab \vac\, ,\\
 & \Proj_1 \phi = \sqrt{\alpha}\gamma_1 \Guab +R_1, \quad \bra\Guab,\, R_1\ket_*=0\, .
\end{split}
\end{equation}

For $\Proj_2\phi$ we make a splitting as in \eqref{def:pi-2-phi-1}-\eqref{def:pi-2-phi-3}:
$$
\Proj_2 \phi = \alpha\gamma_2\Gdab + R_2,\quad \bra R_2,\, \Gdab\ket_* = 0\, .
$$

For the state $\Proj_1\G$, we then write the decomposition
\begin{equation}\nonumber
  \Proj_1\G = \Fu  + \F_2 + \ell_1\, ,
\end{equation}
which is defined by the conditions
\begin{equation}\label{eq:def:Fun-alpha}
\begin{split}
 \Fu  & := 2\sqrt{\alpha} \sum_{i=1}^3  \kappa_i \chi_{(\alpha,\infty)}(H_f)(H_f+P_f^2)^{-1} (\Ac)^{(i)} \mab \vac \frac{\partial u_\alpha}{\partial x_i} \, , \\
 \F_2 & := 2\sqrt{\alpha} \sum_{i=1}^3  \mu_i (H_f+P_f^2)^{-1}  (P_f)^{(i)} \Guab \frac{\partial u_\alpha}{\partial x_i} \, ,
\end{split}
\end{equation}
and $\kappa_i\neq0$, $\mu_i\neq 0$,
\begin{equation}\label{eq:**-1}
 \bra \ell_1,\, \Fu \ket_*=0, \quad \bra \ell_1,\, \F_2\ket_*=0\, .
\end{equation}
The latter conditions \eqref{eq:**-1} determine $\ell_1$ and the coefficients $\kappa_i$, $\mu_i$ that occur in \eqref{eq:def:Fun-alpha}.
Here  $\chi_{(\alpha,\, \infty)}(H_f)$ is an infrared cutoff and $\chi_{(\alpha,\, \infty)}$ is the
characteristic function of $(\alpha,\, \infty)$. Like in the definition \eqref{def:ups} of $\ups$, this cutoff is necessary to have a state in the Hilbert space $\gH$. In particular, $\|\Fu\|^2 = (\sum_i|\kappa_i|\times \mathcal{O}(\alpha^3|\log\alpha|))$. However, for the computation of the binding energy, this cutoff will be removed.

Note that due to Lemma~\ref{lem:appendix-3}, we also have
$$
 \bra \Fu ,\, \F_2\ket_*=0\, .
$$

For the state $\ell_1$ defined above we write
\begin{equation}\nonumber
\ell_1 = \sqrt{\alpha}\beta_1\Gu(g) + \tilde\ell_1\, ,
\end{equation}
where as already defined in \eqref{def:g} and \eqref{def:Gamma-1-g}
$g=\Proj_0 \G$ and $\Gamma_1(g) = - (H_f+P_f^2)^{-1} \sigma.\Bc g$,
and where $\tilde\ell_1$ and $\beta_1$ are determined by the condition
$$
 \bra \tilde\ell_1,\, \Gu(g)\ket_* =0\, .
$$

Due to Lemmata~\ref{lem:appendix-3} and \ref{lem:appendix-4} we obtain
$$
 \bra \Fu, \, \Gu(g)\ket_* =0,\, \quad \bra\F_2,\, \Gu(g)\ket_*=0\, ,
$$
and thus also
$$
 \bra \Fu, \, \tilde\ell_1\ket_* =0,\, \quad \bra\F_2,\, \tilde\ell_1\ket_*=0\, .
$$

We finally define the state $\cL\in\gH$ by
\begin{equation}\label{def:cL}
\begin{split}
 \cL  : = (\Proj_0 \G ,\,  \ell_1,\, \Proj_2 \G, \Proj_3 \G\, ...)
        =  (g ,\,  \ell_1,\, \Proj_2 \G, \Proj_3 \G\, ...)\, .
 \end{split}
\end{equation}

\subsection{Direct terms}
In this paragraph, we shall compute all direct terms occurring in
$$
 \bra\pgs ,\, H\pgs\ket = \bra u_\alpha \phi + \cL + (\Fu + \F_2), \, H (\, u_\alpha \phi + \cL + (\Fu + \F_2))\ket.
$$

Since $\bra u_\alpha,\, \frac{\partial}{\partial x_i}u_\alpha\ket = 0$, for $i=1,2,3$, a straightforward computation yields
\begin{equation}\label{eq:M1}
\bra u_\alpha\phi,\, H u_\alpha\phi\ket \geq (\Sigma_0 -\frac{\alpha^2}{4}) \|\phi\|^2\, .
\end{equation}
Using Proposition~\ref{prop:orthogonality}, yields
\begin{equation}\label{eq:M2}
 \bra \cL,\, H\cL\ket \geq (\Sigma_0 -\frac{\alpha^2}{4})\|\cL\|^2\, .
\end{equation}
Due to the symmetry of $\frac{\partial u_\alpha}{\partial x_i}$, $i=1,2,3$,
the last direct term to estimate can be written as
\begin{equation}\label{eq:M000}
\begin{split}
& \bra \Fu + \F_2,\, H (\Fu + \F_2) \ket \\
 & = \bra \Fu + \F_2,\, \left((-\Delta -\frac{\alpha}{|x|})  + 2\Re \alpha \Ac.\Aa + (H_f + P_f^2)\right) (\Fu + \F_2)\ket
\end{split}
\end{equation}

We have
\begin{equation}\label{bb-1}
\begin{split}
 & \bra (\Fu + \F_2),\, (-\Delta -\frac{\alpha}{|x|}) \, (\Fu + \F_2) \ket \\
 & \geq -\frac{\alpha^2}{4} \|(\Fu + \F_2)\|^2
 \geq - C\alpha^5 \sum_{i=1}^3 ( |\log\alpha|\, |\kappa_i|^2 + |\mu_i|^2 ) \, .
\end{split}
\end{equation}
We also have
\begin{equation}\label{bb-2}
\begin{split}
 \bra  (\Fu + \F_2),\, \Ac.\Aa (\Fu + \F_2)\ket \geq 0\, ,
\end{split}
\end{equation}
and since $\bra \Fu, \F_2\ket_*=0$ we get
\begin{equation}\label{bb-3}
\begin{split}
 & \bra (\Fu + \F_2),\, (H_f+P_f^2)\, (\Fu + \F_2)\ket = \| \Fu\|_*^2 + \| \F_2\|_*^2 \\
 & = 4\alpha \left\|\frac{\partial u_\alpha}{\partial x_1}\right\|^2
 \Big(\sum_{i=1}^3 |\kappa_i|^2 \| \chi_{(\alpha,\infty)}(H_f)(H_f + P_f^2)^{-1} (\Ac)^{(i)} \mab\vac\|_*^2 \\
 & + \sum_{i=1}^3 |\mu_i|^2 \| (H_f+P_f^2)^{-1} P_f^{(i)} \Guab \|_*^2 \Big)\, .
\end{split}
\end{equation}
Therefore, \eqref{eq:M000}-\eqref{bb-3} yields
\begin{equation}\label{eq:M3}
\begin{split}
 & \bra (\Fu + \F_2),\, H\, (\Fu + \F_2)\ket\\
 & \geq 4\alpha \|\frac{\partial u_\alpha}{\partial x_1}\|^2
 \sum_{i=1}^3 |\kappa_i|^2 \| \chi_{(\alpha,\infty)}(H_f)(H_f + P_f^2)^{-1} (\Ac)^{(i)} \mab\vac\|_*^2 \\
 & + 4\alpha \|\frac{\partial u_\alpha}{\partial x_1}\|^2
 \sum_{i=1}^3 |\mu_i|^2 \| (H_f+P_f^2)^{-1} P_f^{(i)} \Guab \|_*^2 \\
 & - C\alpha^5 \sum_{i=1}^3 ( |\log\alpha|\, |\kappa_i|^2 + |\mu_i|^2 )
\end{split}
\end{equation}

\subsection{Cross terms between $u_\alpha \phi$ and $\Fu + \F_2$}

Due to the orthogonality $\bra u_\alpha,\, \frac{\partial}{\partial x_i}u_\alpha\ket = 0$,
for $i=1,2,3$, only terms involving gradients, i.e. $2P.P_f$ and $-4\Re P.\Aa$, are a priori nonzero. We thus have
\begin{equation}\label{eq:M4}
\begin{split}
 & \bra u_\alpha\phi,\, H(\Fu + \F_2)\ket + \mathrm{h.c.} \\
 & =
 2\Re \bra \Proj_1 u_\alpha\phi,\, (-2P.P_f) (\Fu + \F_2)\ket
 + \Re\bra \Proj_0 u_\alpha\phi,\, (-4\sqrt{\alpha} P.\Aa)(\Fu + \F_2)\ket \\
 & + \Re\bra (\Fu + \F_2),\, (-4\sqrt{\alpha} P.\Aa)\Proj_2 u_\alpha\phi\ket\, ,
\end{split}
\end{equation}
where $\mathrm{h.c.}$ stands for hermitian conjugate.

The first term in the right hand side of \eqref{eq:M4} is
\begin{equation}\label{eq:M5}
\begin{split}
 & 2\Re \bra \Proj_1 u_\alpha\phi,\, (-2P.P_f) (\Fu + \F_2)\ket \\
 & = -4\Re \bra (\sqrt{\alpha} \gamma_1 \Guab + R_1) u_\alpha, P.P_f (\Fu + \F_2)\ket \\
 & \geq -8 \alpha\Re  \left\| \frac{\partial u_\alpha}{\partial x_1}\right\|^2\sum_i
 \gamma_1 \bar\mu_i \|(H_f + P_f^2)^{-1} P_f^{(i)} \Guab\|_*^2 \\
 & -8 \alpha \Re\!  \left\| \frac{\partial u_\alpha}{\partial x_1}\right\|^2\!
 \sum_i\gamma_1\bar\kappa_i\, \bra P_f^{(i)}\Guab\! , \chi_{(\alpha,\infty)}(H_f)(H_f+P_f^2)^{-1} (\Ac)^{(i)}\!\mab\vac\ket\\
 & - c \left\| \frac{\partial u_\alpha}{\partial x_1}\right\|^2
 \sum_i\| P_f^{(i)} R_1\| \sqrt{\alpha} \sum_i (|\log\alpha|^\frac12|\kappa_i| +|\mu_i|) \\
 & \geq  -8 \alpha\Re \sum_i \left\| \frac{\partial u_\alpha}{\partial x_1}\right\|^2
 \gamma_1 \bar\mu_i \|(H_f + P_f^2)^{-1} P_f^{(i)} \Guab\|_*^2 \\
 & - c \alpha^{3+\frac56}|\log\alpha|^\frac12 \sum_i( |\kappa_i| + |\mu_i|)\, ,
\end{split}
\end{equation}
where in the last inequality, we used
$\bra P_f^{(i)} (H_f+P_f^2)^{-1} \Guab,\, \chi_{(\alpha,\infty)}(H_f)
(H_f+P_f^2)^{-1}(\Ac)^{(i)} \mab\vac\ket_*=0$ (see Lemma~\ref{lem:appendix-3}),
and $\|P_f^{(i)} R_1\| = \mathcal{O}(\alpha^\frac43)$ (see Proposition~\ref{IPNE}).

The second term in \eqref{eq:M4} equals
\begin{equation}\nonumber
\begin{split}
 & -4\Re\sqrt{\alpha} \bra \Proj_0 u_\alpha\phi,\, P.\Aa (\Fu + \F_2)\ket \\
 & =
 -8\alpha \left\| \frac{\partial u_\alpha}{\partial x_1}\right\|^2
 \sum_i \gamma_0 \bar\kappa_i \|\chi_{(\alpha,\infty)}(H_f)  (H_f+P_f^2)^{-1}
 (\Ac)^{(i)} \mab\vac\|_*^2\,
\end{split}
\end{equation}
where we applied again Lemma~\ref{lem:appendix-3}.

For the last term in the right hand side of \eqref{eq:M4}, we have
\begin{equation}\label{eq:M7}
\begin{split}
 & -4\sqrt{\alpha} \Re \bra \Fu  + \F_2,\, P.\Aa \Proj_2 u_\alpha\phi\ket  \\
 & \geq -c \alpha \left\| \frac{\partial u_\alpha}{\partial x_1}\right\|
 \sum_i(|\log\alpha|^\frac12|\kappa_i| + |\mu_i|) \|Pu_\alpha\|\, \|\Aa\Proj_2\phi\| \\
 & \geq - c\alpha^4 |\log\alpha|^\frac12\sum_i (|\kappa_i| + |\mu_i|)\, ,
\end{split}
\end{equation}
where in the second inequality we inserted the estimate $\|\Aa\Proj_2\phi\| =
\mathcal{O}(\alpha)$, given by Proposition~\ref{proposition:4}.

Collecting \eqref{eq:M4}-\eqref{eq:M7} yields
\begin{equation}\label{eq:Main1}
\begin{split}
 & \bra u_\alpha\phi, H(\Fu + \F_2)\ket +\mathrm{h.c.} \\
 & \geq
 -8\alpha \left\| \frac{\partial u_\alpha}{\partial x_1}\right\|^2
 \Re\sum_i \gamma_1 \bar\mu_i \| (H_f+P_f^2)^{-1} P_f^{(i)} \Guab\|_*^2 \\
 & - 8\alpha \left\| \frac{\partial u_\alpha}{\partial x_1}\right\|^2
 \Re\sum_i \gamma_0 \bar\kappa_i \| \chi_{(\alpha,\infty)}(H_f)(H_f+P_f^2)^{-1} (\Ac)^{(i)} \mab\vac\|_*^2 \\
 & - c\alpha^{3+\frac56}|\log\alpha| \sum_i (|\mu_i| + |\kappa_i|)\, .
\end{split}
\end{equation}

\subsection{Cross terms between $(\Fu + \F_2)$ and $\cL$}

To compute these terms, we will estimate separately the contribution of each
operator occurring in $H$.

\subsubsection{$(-\Delta-\frac{\alpha}{|x|})$-term}\label{S4.5.1}$\ $
\begin{equation*}
\begin{split}
 & 2\Re \bra \cL, (-\Delta-\frac{\alpha}{|x|}) (\Fu + \F_2)\ket \\
 & \geq
 -c \sum_i\|(-\Delta-\frac{\alpha}{|x|}) \frac{\partial u_\alpha}{\partial x_i}\|
 \sqrt{\alpha} \sum_i(|\log\alpha|^\frac12|\kappa_i| + |\mu_i|) \| \Proj_1 \cL\| \\
 & \geq
 - c\alpha^{4+\frac13}|\log\alpha|^\frac12  \sum_i(|\kappa_i| + |\mu_i|)
 - c\alpha^5|\log\alpha|\sum_i(|\kappa_i|^2 + |\mu_i|^2) \, ,
\end{split}
\end{equation*}
where in the last inequality, we used $\|(-\Delta-\frac{\alpha}{|x|}) \frac{\partial u_\alpha}{\partial x_i}\|=\mathcal{O}(\alpha^3)$, $\|\Proj_1 \G\| = \mathcal{O}(\alpha^\frac56)$ due to Corollary~\ref{Cmain}, and $\|\Proj_1\cL\| =\|\Proj_1 \G -(\Fu + \F_2)\| \leq c\alpha^\frac56 + c\alpha^\frac32 \sum_i(|\log\alpha|^\frac12|\kappa_i| + |\mu_i|)$.

\subsubsection{$(H_f+P_f^2)$-term}
This term is zero since $\Proj_1\cL = \ell_1$ is by definition $*$-orthogonal to $\Fu $ and $\F_2$ (see \eqref{eq:**-1}).

\subsubsection{$(-2 P.P_f)$-term}
We first remark that
$$
\bra (\Fu + \F_2), \, P.P_f (\Fu + \F_2)\ket =0\, ,
$$
by orthogonality of $(\partial^2 u_\alpha)/(\partial x_i \partial x_j)$ and
$(\partial u_\alpha)/(\partial x_i)$.
Therefore, since
$$
\Proj_1 \G = \Proj_1\cL + \Fu  + \F_2 \, ,
$$
we obtain
\begin{equation*}
\begin{split}
   \bra \Proj_1\cL,\, -2P.P_f (\Fu + \F_2)\ket
 = \bra \Proj_1 \G,\, -2P.P_f (\Fu + \F_2) \ket \, .
\end{split}
\end{equation*}
Thus
\begin{equation*}
\begin{split}
  2\Re\bra \Proj_1 \cL,\, (-2P.P_f) (\Fu + \F_2)\ket
 & \geq  -c \|\Proj_1 \G \|\, \|P(\Fu + \F_2)\| \\
 & \geq - c \alpha^{3+\frac13} |\log\alpha|^\frac12\sum_i (|\kappa_i| + |\mu_i|)\, ,
\end{split}
\end{equation*}
using again $\|\Proj_1 \G\| = \mathcal{O}(\alpha^\frac56)$.

\subsubsection{$(-4\sqrt{\alpha}P.\Aa)$-term} From Corollary~\ref{Cmain} we have $\|\Proj_0 \cL\| = \|g\| = \mathcal{O}(\alpha^\frac13)$. Moreover, $\| P.\Aa(\Fu + \F_2) \|= \mathcal{O}(\alpha^\frac52)\sum_i(|\kappa_i| + |\mu_i|)$. Therefore
\begin{equation*}
\begin{split}
 \Re \bra \Proj_0\cL,\, -4\sqrt{\alpha}P.\Aa(\Fu + \F_2)\ket & \geq - 4 \sqrt{\alpha} \|\Proj_0 \cL\| \, \| P.\Aa(\Fu + \F_2) \| \\
 & \geq - c \alpha^{3+\frac13} \sum_i(|\kappa_i| + |\mu_i|)\, .
\end{split}
\end{equation*}
Similarly, using $\| P (\Fu + \F_2) \|= \mathcal{O}(\alpha^\frac52)|\log\alpha|^\frac12\sum_i(|\kappa_i| + |\mu_i|)$ and applying $\|H_f^\frac12\Proj_2\cL\| = \|H_f^\frac12\Proj_2 \G\| = \mathcal{O}(\alpha^\frac43)$ from Corollary~\ref{Cmain} yields
\begin{equation}\label{eq:term-3.4-1}
\begin{split}
 -4\sqrt{\alpha}\Re \bra (\Fu + \F_2),\, P.\Aa \Proj_2 \cL \ket & \geq - c \sqrt{\alpha} \|P(\Fu + \F_2)\| \, \| H_f^\frac12 \Proj_2\cL \| \\
 & \geq - c \alpha^{4+\frac13} |\log\alpha|^\frac12\sum_i(|\kappa_i| + |\mu_i|)\, .
\end{split}
\end{equation}

\subsubsection{$(\sqrt{\alpha}P_f.\Aa)$-term} We have
\begin{equation*}
 4\sqrt{\alpha} \Re\bra \Proj_0\cL,\, P_f.\Aa(\Fu + \F_2)\ket = 0\,
\end{equation*}
since $P_f\Proj_0 \cL=0$. Similarly as for the estimate of \eqref{eq:term-3.4-1} we get
\begin{equation*}
\begin{split}
 4\sqrt{\alpha}\Re \bra (\Fu + \F_2),\, P_f.\Aa \Proj_2\cL\ket
 & \geq - c \sqrt{\alpha} \|P_f(\Fu + \F_2)\| \, \| H_f^\frac12 \Proj_2\cL \| \\
 & \geq - c\alpha^{3+\frac13}  \sum_i (| \kappa_i| + |\mu_i|)\, .
\end{split}
\end{equation*}

\subsubsection{$(2\alpha\Ac.\Aa)$-term}
We have
\begin{equation*}
 \begin{split}
 &      2\Re \bra (\Fu + \F_2),\, 2\alpha \Ac.\Aa \Proj_1\cL\ket
  \geq   - c \alpha\, \| \Aa(\Fu  + \F_2) \| \, \|\Proj_1 \cL\| \\
 & \geq  - c \alpha^{3+\frac13} \sum_i(|\kappa_i| + |\mu_i|)
         - c \alpha^4 |\log\alpha|^\frac12\sum_i (|\kappa_i|^2 + |\mu_i|^2)\,
 \end{split}
\end{equation*}
where we used, as in subsection~\ref{S4.5.1}, $\|\Proj_1\cL\| \leq c\alpha^\frac56 + c\alpha^\frac32| \sum_i(|\log\alpha|^\frac12|\kappa_i| + |\mu_i|)$.

\subsubsection{$(2\alpha\Re(\Aa)^2)$-term}
From Proposition~\ref{IPNE} we have $\|H_f^\frac12 \Proj_3 \cL\| = \mathcal{O}(\alpha^{\frac43})$, thus
\begin{equation*}
\begin{split}
       \Re \bra (\Fu + \F_2) , \, 2\alpha \Aa.\Aa \Proj_3\cL\ket
 & \geq  -c\alpha \| \Fu  + \F_2\|\, \|H_f^\frac12 \Proj_3 \cL\| \\
 & \geq -c \alpha^{3+\frac56} |\log\alpha|^\frac12 \sum_i (|\kappa_i| + |\mu_i|)\, ,
\end{split}
\end{equation*}

\subsubsection{$(2\sqrt{\alpha} \Re\sigma.\Ba)$-term}\label{S4.5.8}
We have two terms to estimate. The first is
\begin{equation*}
\begin{split}
       \Re\bra (\Fu + \F_2) ,\, 2\sqrt{\alpha}\sigma.\Ba \Proj_2 \cL\ket
&  \geq - c \sqrt{\alpha} \| \Fu  + \F_2 \|\, \| H_f^\frac12 \Proj_2\cL\|\\
&  \geq - c \alpha^{3+\frac13} |\log\alpha|^\frac12\sum_i(|\kappa_i| + |\mu_i|)\, ,
\end{split}
\end{equation*}
and the second is
\begin{equation*}
 \Re \bra \Proj_0\cL,\, 2\sqrt{\alpha} \sigma.\Ba (\Fu + \F_2)\ket
 = - 2\sqrt{\alpha} \Re \bra \Gu(g),\, (\Fu + \F_2)\ket_* =0\, ,
\end{equation*}
since Lemma~\ref{lem:appendix-3} yields that $\Fu $ and $\F_2$ are $*$-orthogonal in $\gF$ to $\Gutab$, for all $\tilde a$ and $\tilde b$ in $\C$.

\subsubsection{Collecting all cross terms between $(\Fu + \F_2)$ and $\cL$} With the estimates of subsection~\ref{S4.5.1} to subsection~\ref{S4.5.8} we obtain
\begin{equation}\label{eq:Main2}
\begin{split}
       & \bra \Fu  + \F_2,\, H \cL\ket + \mathrm{h.c.}\\
 & \geq  -c \alpha^4 |\log\alpha|^\frac12 \sum_i (|\kappa_i|^2 + |\mu_i|^2)
       - c\alpha^{3+\frac13}|\log\alpha|^\frac12 \sum_i (|\kappa_i| + |\mu_i|)\, .
\end{split}
\end{equation}

\subsection{Cross terms between $u_\alpha \phi $ and $\cL$} Due to the orthogonalities
  $$
   \forall\mtab, \, \forall k, \quad \bra u_\alpha\mtab, \, \Proj_k \G\ket_{L^2(\R^3,\, \mathrm{d}x)\otimes\C^2} =0\,
  $$
 the only cross terms in $\bra u_\alpha \phi,\, H \, \cL\ket$ are due to the cross terms with $-2\Re P.P_f$ and the cross terms with $-4\sqrt{\alpha} \Re P.\Aa$.

\subsubsection{$(-2\Re P.P_f)$-term} We shall estimate separately the three terms in the right hand side of the equality
\begin{equation}\label{eq:M17}
\begin{split}
  & 2\Re \bra u_\alpha \phi,\, -2 P.P_f \cL\ket \\
  & \!=\! -4\Re \bra \Proj_1 u_\alpha\phi,\, P.P_f \Proj_1 \cL\ket
  - 4\Re \bra \Proj_2 u_\alpha\phi,\, P.P_f \Proj_2 \cL\ket \\
  & - 4\Re \bra \Proj_{\geq 3} u_\alpha\phi,\, P.P_f \Proj_{\geq 3} \cL\ket .
\end{split}
\end{equation}
The third term is
\begin{equation}\label{eq:M18}
     -4\Re \bra\Proj_{\geq 3} u_\alpha\phi,\, P.P_f \Proj_{\geq 3} \cL\ket
\geq - c\|P u_\alpha\|\, \|P_f^\frac12 \Proj_{\geq3} \phi\|\, \|P_f^\frac12 \Proj_{\geq 3} \cL\|
\geq - c\alpha^{3+\frac23}\, ,
\end{equation}
using Corollary~\ref{Cmain}.

For the second term in the right hand side of \eqref{eq:M17}, we write
\begin{equation}\label{eq:M19}
      -4\Re\bra \Proj_2 u_\alpha\phi,\, P.P_f \Proj_2 \cL\ket
 \geq -c \|P u_\alpha\| \, \|P_f^\frac12 \Proj_2 \phi\|\, \|P_f^\frac12 \Proj_2 \cL\|
 \geq - c\alpha^{3+\frac13}\, ,
\end{equation}
using from Proposition~\ref{proposition:4} that $\| H_f^{\frac12} \Proj_2 \phi\| = \mathcal{O}(\alpha)$ and from Corollary~\ref{Cmain} that $\| P_f^{\frac12} \Proj_2\cL\| = \mathcal{O}(\alpha^\frac43)$.

Eventually, the first term in the right hand side of \eqref{eq:M17} is estimated as
\begin{equation}\label{eq:M20}
\begin{split}
 &    - 4\Re \bra \Proj_1 u_\alpha\phi ,\, P.P_f \Proj_1 \cL\ket\\
 & =  - 4\Re \bra \sqrt{\alpha} \gamma_1\Guab u_\alpha,\, P.P_f \Proj_1\cL\ket
      - 4\Re \bra u_\alpha R_1,\, P.P_f \Proj_1\cL\ket \\
 & =  - 4\Re \bra u_\alpha R_1,\, P.P_f\Proj_1 \cL\ket \\
 & \geq - c \|P u_\alpha\|\, \|P_f^\frac12 R_1\|\, \|P_f^\frac12 \Proj_1\cL\| \geq - c\alpha^{3+\frac16}\,
\end{split}
\end{equation}
where in the second equality we used that
 $$
 \bra P.P_f\Guab,\, \Proj_1\cL\ket = \bra (H_f+P_f^2)^{-1} P.P_f \Guab,\, \ell_1\ket_* =0\, ,
 $$
since $\Fu $ is $*$-orthogonal to $\ell_1$, and in the last inequality we used $\| P_f^\frac12 \Proj_1 \cL \| \leq \| \ell_1 \|_* \leq \|\Proj_1 \G\|_* = \mathcal{O}(\alpha^\frac56)$ and $\|P_f^\frac12 R_1\| = \mathcal{O}(\alpha^\frac43)$ due to Corollary~\ref{Cmain}.

Collecting \eqref{eq:M17}-\eqref{eq:M20} yields
\begin{equation}\label{eq:M20.1}
 2\Re\bra u_\alpha\phi,\, -2P.P_f \cL\ket \geq - c\alpha^{3+\frac16}\, .
\end{equation}

\subsubsection{$(-4\sqrt{\alpha} P.\Aa)$-term} We write
\begin{equation}\label{eq:M21}
\begin{split}
 &   \Re \bra u_\alpha \phi,\, -4\sqrt{\alpha}P.\Aa\cL\ket + \Re\bra \cL,\, -4\sqrt{\alpha}P.\Aa u_\alpha\phi\ket \\
 & = \Re\bra \Proj_{\geq 2}\cL,\, -4\sqrt{\alpha} P.\Aa \Proj_{\geq 3}u_\alpha\phi\ket
   + \Re\bra \Proj_{\geq 2} u_\alpha\phi,\, -4\sqrt{\alpha} P.\Aa \Proj_{\geq 3}\cL\ket \\
 & + \Re\bra \Proj_1 \cL,\, -4\sqrt{\alpha} P.\Aa\Proj_2 u_\alpha\phi\ket
   + \Re\bra \Proj_1 u_\alpha\phi,\, -4\sqrt{\alpha} P.\Aa \Proj_2\cL\ket \\
 & + \Re\bra \Proj_0\cL,\, -4\sqrt{\alpha}P.\Aa \Proj_1u_\alpha\phi\ket
   + \Re\bra \Proj_0 u_\alpha\phi,\, -4\sqrt{\alpha}P.\Aa\Proj_1\cL\ket\, .
\end{split}
\end{equation}
For the first term in the right hand side of \eqref{eq:M21}, we write
\begin{equation}\label{eq:M22}
\begin{split}
  \Re \bra \Proj_{\geq 2} \cL,\, -4\sqrt{\alpha} P.\Aa \Proj_{\geq 3} u_\alpha\phi\ket
 & \geq -c\sqrt{\alpha} \|P u_\alpha\|\, \|\Aa \Proj_{\geq 3} \phi\|\, \|\Proj_{\geq 2}\cL\| \\
 & \geq -c\alpha^{3+\frac23}\, .
\end{split}
\end{equation}

For the second term in the right hand side of \eqref{eq:M21} we have
\begin{equation}\label{eq:M23}
\begin{split}
 &       \Re \bra \Proj_{\geq 2} u_\alpha\phi, \, -4\sqrt{\alpha} P.\Aa \Proj_{\geq 3} \cL\ket \\
 & \geq  - c\sqrt{\alpha} \|\Proj_{\geq2} \phi\|\, \|P u_\alpha\| \|\Aa\Proj_{\geq3} \cL\|
   \geq  - c \alpha^{3+\frac23}\, .
\end{split}
\end{equation}

For the third term we have
\begin{equation}\label{eq:M24}
\begin{split}
 &       \Re \bra \Proj_1\cL,\, -4\sqrt{\alpha} P.\Aa \Proj_2 u_\alpha\phi\ket \\
 &  \geq -c\sqrt{\alpha} \|P u_\alpha\|\, \|\Aa\Proj_2\phi\|\, \|\Proj_1 \cL \|
 \geq - c \alpha^{3+\frac13} - c\alpha^4 |\log\alpha|^\frac12\sum_i(|\kappa_i| + |\mu_i|)\, ,
\end{split}
\end{equation}
using the bound $ \|\Proj_1\cL\| \leq c\alpha^\frac56 + c \alpha^\frac32
|\log\alpha|^\frac12 \sum_i (|\kappa_i| + |\mu_i|)$ as in subsection ~\ref{S4.5.1}.

For the fourth term, we have the following estimate
\begin{equation}\label{eq:M25}
\begin{split}
 & \Re \bra \Proj_1 \phi u_\alpha,\, -4\Re \sqrt{\alpha} P.\Aa \Proj_2 \cL\ket \\
 & \geq -c\sqrt{\alpha} \|P u_\alpha\|\, \|\Aa \Proj_2 \cL\| \, \| \Proj_1\phi\|
   \geq - c\alpha^{3+\frac13}\, .
\end{split}
\end{equation}
using $\|H_f^\frac12 \Proj_2\cL\| = \mathcal{O}(\alpha^\frac43)$  from Corollary~\ref{Cmain} and $\|\Proj_1\phi\| = \mathcal{O}(\alpha^\frac12)$ from Proposition~\ref{Nf-exp-lemma-1}.

The fifth term in the right hand side of \eqref{eq:M21} is
\begin{equation}\label{eq:M26}
\begin{split}
 &      \Re \bra g,\, -4\sqrt{\alpha} P.\Aa\Proj_1 u_\alpha\phi\ket \\
 & =    -4\sqrt{\alpha} \Re \bra P.\Ac g, \sqrt{\alpha}\gamma_1\Guab u_\alpha\ket
        -4\sqrt{\alpha} \Re \bra g,\, P.\Aa R_1 u_\alpha\ket \\
 & \geq - c\sqrt{\alpha} \|P u_\alpha\|\, \|g\|\, \|\Aa R_1\| \geq -c\alpha^{3+\frac16}\, ,
\end{split}
\end{equation}
using $\bra P.\Ac g,\, \Guab u_\alpha\ket = 0$ from Lemma~\ref{lem:appendix-4} and $\|g\| = \mathcal{O}(\alpha^\frac13)$, $\|H_f^\frac12 R_1\| = \mathcal{O}(\alpha^\frac43)$ from Corollary~\ref{Cmain}.

The last term in the right hand side of \eqref{eq:M21} is given by
\begin{equation}\label{eq:M27}
\begin{split}
 &   -4 \sqrt{\alpha} \Re \bra \gamma_0 u_\alpha\mab\vac,\, P.\Aa \ell_1\ket \\
 & = -4 \sqrt{\alpha} \Re \gamma_0 \bra (H_f+P_f^2)^{-1} P.\Ac \mab\vac u_\alpha,\, \ell_1 \ket_* = 0\, ,
\end{split}
\end{equation}
since $\ell_1$ is $*$-orthogonal to $\Fu $.

Collecting \eqref{eq:M21}-\eqref{eq:M27} thus yields
\begin{equation*}
\begin{split}
 & \Re \bra u_\alpha \phi,\, -4\sqrt{\alpha}P.\Aa\cL\ket + \Re\bra \cL,\, -4\sqrt{\alpha}P.\Aa u_\alpha\phi\ket\\
 & \geq - c\alpha^{3+\frac16} - \epsilon\alpha^4 |\log\alpha|^\frac12\sum_i (|\kappa_i| + |\mu_i|)\, ,
\end{split}
\end{equation*}
which in turn gives, together with \eqref{eq:M20.1},
\begin{equation}\label{eq:Main3}
 \bra u_\alpha \phi,\, H\, \cL\ket + \mathrm{h.c.}
 \geq
  - c\alpha^{3+\frac16} - c\alpha^4 |\log\alpha|^\frac12\sum_i (|\kappa_i| + |\mu_i|)\, .
\end{equation}

\subsection{Normalization}
Since $\pgs$ is chosen so that $\|\Proj_0 \pgs\|=1$, then $\| \pgs \|\neq 1$. Thus, we need some estimate on its norm in order to conclude the proof of the upper bound for the binding energy. The required estimate is the following
\begin{lem}\label{lem:normalization} We have
\begin{equation*}
\begin{split}
& \left| \| \pgs \|^2 - (\|\phi\|^2 + \|\cL\|^2)\right| \\
& \leq   c \alpha^{2+\frac13}|\log\alpha|^\frac12 \sum_i(|\kappa_i| + |\mu_i|)
       + c \alpha^{3}|\log\alpha| \sum_i(|\kappa_i|^2 + |\mu_i|^2)\, .
\end{split}
\end{equation*}
\end{lem}
\begin{proof}
Since $\pgs= u_\alpha\phi + \G$, $\bra\Proj_k u_\alpha\phi,\, \Proj_k \G\ket =0$ for all $k$, $G=\cL +(\Fu + \F_2)$, and $\|u_\alpha\|=1$,  we have
\begin{equation*}
\begin{split}
  \|\pgs\|^2 -(\|\phi\|^2 + \|\cL\|^2) & = \|\Proj_1\pgs\|^2 - (\|\Proj_1 \phi\|^2
  + \|\Proj_1 \cL\|^2)\\
  & = \|\Proj_1 \G\|^2 - \|\Proj_1\cL\|^2 = \|\Proj_1 \G\|^2 - \|\Proj_1G -\Fu - \F_2\|^2\\
  & = - \| \Fu  + \F_2 \|^2 + 2\Re\bra \Proj_1 \G,\, (\Fu + \F_2)\ket\, .
\end{split}
\end{equation*}
Since $\|\Fu \|^2 \leq c \alpha^3|\log\alpha|\sum_i|\kappa_i|^2$, $\| \F_2\|^2 \leq c \alpha^3\sum_i|\mu_i|^2$ and $\|\Proj_1 \G\| = \mathcal{O}(\alpha^\frac56)$ (see Corollary~\ref{Cmain}), we conclude the proof.
\end{proof}

\subsection{Proof of Theorem~\ref{thm:main}}
The inequality
\begin{equation}\label{eq:M28.1}
   \Sigma_0 -\Sigma \geq  \frac14\alpha^2 \ +\  \left(\dvar^{(1)}
   +  {\dvar}_{\mbox{\scriptsize{Zeeman}}}^{(1)}
   \right) \alpha^3 \  + \  \mathcal{O}(\alpha^{3+\frac16})\, ,
\end{equation}
is proved in Section~\ref{S3}.

To prove the upper bound on the binding energy, we collect the estimates \eqref{eq:M1}, \eqref{eq:M2}, \eqref{eq:M3}, \eqref{eq:Main1}, \eqref{eq:Main2} and \eqref{eq:Main3} of Section~\ref{S4}. This gives
\begin{equation*}
\begin{split}
 &       \bra \pgs,\, H\, \pgs\ket \\
 & \geq  (\Sigma_0-\frac{\alpha^2}{4})(\|\phi\|^2 + \|\cL\|^2)\\
 & + 4\alpha \left\| \frac{\partial u_\alpha}{\partial x_1} \right\|^2
 \sum_i (|\kappa_i|^2 - 2\gamma_0\bar\kappa_i) \|\chi_{(\alpha,\infty)}(H_f)(H_f+P_f^2)^{-1} (\Ac)^{(i)}\mab\vac\|_*^2 \\
 & + 4\alpha \left\| \frac{\partial u_\alpha}{\partial x_1} \right\|^2
 \sum_i (|\mu_i|^2 - 2\gamma_1 \bar\mu_i) \|(H_f+P_f^2)^{-1} (P_f)^{(i)} \Guab\|_*^2\\
 & - c\alpha^{3+\frac16} - c\alpha^{3+\frac13}
 |\log\alpha|^\frac12\sum_i(|\kappa_i| + |\mu_i|) - c\alpha^4|\log\alpha|^\frac12\sum_i(|\kappa_i|^2 + |\mu_i|^2)\, .
\end{split}
\end{equation*}
Replacing $\gamma_0$ by $(\gamma_0 -1) + 1$, $\gamma_1$ by $(\gamma_1-1) + 1$, and using from Corollary~\ref{Cmain} that $|\gamma_0-1| =\mathcal{O}(\alpha^\frac23)$ and $|\gamma_1-1| = \mathcal{O}(\alpha^\frac23)$ implies
\begin{equation}\nonumber
\begin{split}
 &       \bra \pgs,\, H\, \pgs\ket \\
 & \geq  (\Sigma_0\!-\!\frac{\alpha^2}{4})(\|\phi\|^2 \!+\! \|\cL\|^2) \\
 & -4\alpha
 \left\|\frac{\partial u_\alpha}{\partial x_1}\right\|^2\!
 \sum_i \!\|\chi_{(\alpha,\infty)}(H_f)(H_f+P_f^2)^{-1} (\Ac)^{(i)}\!\mab\vac\!\|_*^2 \\
 & -4 \alpha  \left\|\frac{\partial u_\alpha}{\partial x_1}\right\|^2\!
 \sum_i \|(H_f+P_f^2)^{-1} (P_f)^{(i)} \Guab\|_*^2
 + |c_0|^2 \alpha^3 \sum_i( |\kappa_i\!-\!1|^2 \!+\! |\mu_i\!-\!1|^2) \\
 & - c\alpha^{3+\frac16} - c\alpha^{3+\frac13}|\log\alpha|^\frac12
 \sum_i(|\kappa_i| + |\mu_i|) - c\alpha^4|\log\alpha|^\frac12\sum_i(|\kappa_i|^2 + |\mu_i|^2)\, .\\
\end{split}
\end{equation}

Moreover, with $\Sigma_0 = \mathcal{O}(\alpha)$ (see e.g. \cite{BV4} and references therein) and the estimate $\|\pgs\|^2-(\|\psi\|^2 + \|\cL\|^2) = \mathcal{O}(\alpha^{2+\frac13}|\log\alpha|^\frac12)$ from Lemma~\ref{lem:normalization}, we obtain
 \begin{equation}\label{eq:M29-2}
\begin{split}
 &       \bra \pgs,\, H\, \pgs\ket \\
 & \geq  (\Sigma_0\! - \! \frac{\alpha^2}{4})\|\pgs\|^2 \! - \! 4\alpha
 \left\|\frac{\partial u_\alpha}{\partial x_1}\right\|^2
 \!\sum_i \!\|\chi_{(\alpha,\infty)}(H_f)(H_f+P_f^2)^{-1} (\Ac)^{(i)}\!\mab\vac\!\|_*^2 \\
 & -4 \alpha \!\sum_i \! \left\|\frac{\partial u_\alpha}{\partial x_1}\right\|^2
 \|(H_f+P_f^2)^{-1} (P_f)^{(i)} \Guab\|_*^2
 + |c_0|^2 \alpha^3 \sum_i( |\kappa_i\!-\!1|^2 \!+\! |\mu_i\!-\!1|^2) \\
 & - c\alpha^{3+\frac16} - c\alpha^{3+\frac13}|\log\alpha|^\frac12
 \sum_i(|\kappa_i| + |\mu_i|) - c\alpha^4|\log\alpha|\sum_i(|\kappa_i|^2 + |\mu_i|^2)\, .\\
\end{split}
\end{equation}

Thus, if we next remark that according to \eqref{eq:M29-2} $\kappa_i$ and $\mu_i$ ($i=1,2,3$) are necessarily bounded with respect to $\alpha$, we obtain the inequality
\begin{equation}\label{eq:M30}
\begin{split}
  & \bra\pgs,\, H\, \pgs\ket \\
  & \geq (\Sigma_0-\frac{\alpha^2}{4})\|\pgs\|
  - 4\alpha
 \left\|\frac{\partial u_\alpha}{\partial x_1}\right\|^2
 \sum_i \|\chi_{(\alpha,\infty)}(H_f)(H_f+P_f^2)^{-1} (\Ac)^{(i)}\mab\vac\|_*^2 \\
  & -4 \alpha  \left\|\frac{\partial u_\alpha}{\partial x_1}\right\|^2
 \sum_i \|(H_f+P_f^2)^{-1} (P_f)^{(i)} \Guab\|_*^2 +\mathcal{O}(\alpha^{3+\frac16})\\
  & \geq \left( \Sigma_0 - \frac14\alpha^2 \ - \  \left(\dvar^{(1)}
   +  {\dvar}_{\mbox{\scriptsize{Zeeman}}}^{(1)}\right) \alpha^3 \  + \  \mathcal{O}(\alpha^{3+\frac16}) \right) \|\pgs\|^2\, ,
\end{split}
\end{equation}
where the second inequality is obtained with the estimates in Lemma~\ref{lem:appendix-2}.

Inequalities \eqref{eq:M28.1} and \eqref{eq:M30} conclude the proof of Theorem~\ref{thm:main}.

\appendix

\section{Ground state of $T(0)$}\label{appendix1}

For convenience of the reader, we remind here some properties
derived in \cite{BV4} for ground states of the self-energy operator with total momentum zero.
\begin{thm}\label{thm:bv4}
Let
\begin{equation}\label{eq:def-gamma1}
  \Gu := -(H_f + P_f^2)^{-1} \sigma\!\cdot\!B^+ \vac\zo
\end{equation}
and
\begin{equation}\label{eq:def-gamma2}
  \Gd = -(H_f+P_f^2)^{-1} \left(
  \sigma\!\cdot\! \Bc \Gu\  +\  2\Ac \!\cdot\! P_f \Gu\ +\  \Ac\!\cdot\!\Ac \vac\zo
 \right)\, .
\end{equation}
We have
\begin{equation}\label{eq:sigma-zero}
\begin{split}
 & \Sigma_0 = \inf\mathrm{spec} (T(0)) \\
 & =  -\alpha \|\Gu\|_*^2
  \ +\ \alpha^2\left(
 2 \|A^-\Gu\|^2 - \|\Gd\|_*^2 + \|\Gu\|_*^2\, \|\Gu\|^2\right)
 \ +\ \mathcal{O}(\alpha^3)\, .
\end{split}
\end{equation}

In addition, let $\gs$ be the ground state of $T(0)$ such that $\Proj_0\gs=\vac\zo$.
Taking the $\bra\, .\, , \, .\, \ket_*$-orthonormal projections of $\gs$ along the vectors $\Gu$ and $\Gd$, and denoting by $R$ the component in the $\bra\, .\, , \, .\, \ket_*$-orthogonal complement of their span, we get
\begin{equation}\label{eq:def-phi0-1}
 \gs = \vac\zo\ +\ \alpha^\frac12 \gamma_1 \Gu \ + \ \alpha \gamma_2\Gd\ +\ R
\end{equation}
where for $i=1,2$
\begin{equation}\label{eq:def-phi0-2}
 \bra \Gamma_i,\, R\ket_* =0\, \quad\mbox{and}\quad\bra \vac\zo, \, R\ket=0\, .
\end{equation}
Then,  we have
\begin{equation}\label{eq:est-1}
\begin{split}
 &  |\gamma_1-1| = \mathcal{O}(\alpha)\, ,\quad
 |\gamma_2-1| = \mathcal{O}(\alpha^\frac12)\, ,\\
 &  \|R\|_* = \mathcal{O}(\alpha^\frac32)\quad \mbox{and} \quad
 \|R\| = \mathcal{O}(\alpha)\, .
\end{split}
\end{equation}
\end{thm}

\section{Proof of Proposition~\ref{Nf-exp-lemma-1} on photon number bound}\label{appendix-photon-number-bound}


The strategy of the proof of Proposition~\ref{Nf-exp-lemma-1} is similar to the proof of
the photon number bound in \cite{BCVV2}. However, due to the occurrence of the spin-Zeeman term $\sqrt{\alpha} \sigma.B$, the photon number is one order larger in powers of the fine structure constant.

\begin{proof}
If we denote by $\tH\,:=\, :\ttH :$ the operator $K$ (see definition in \eqref{eq:def-op-K}) with normal ordering, we have
\begin{equation}\label{eq:pn-1}
\begin{split}
 & \| (i\nabla - \sqrt{\alpha}A(x))\pgs\|^2 \\
 & =
 \bra \pgs,\, K \pgs\ket + \bra\pgs,\, \frac{\alpha}{|x|}\pgs\ket - \bra\pgs,\, H_f\pgs\ket \\
 & \ \ \ - 2\Re \sqrt{\alpha}\bra\pgs,\, \sigma\cdot B^-(x)\pgs\ket \\
 & \leq \bra \pgs,\,\tH \pgs\ket + \cno\alpha  + \bra\pgs,\, \frac{\alpha}{|x|}\pgs\ket  ,
\end{split}
\end{equation}
since
$$
\Re \bra \phi,\, (i\nabla - \sqrt{\alpha}A(x))^2 - \alpha/|x| + H_f + 2 \sqrt{\alpha}\sigma\cdot
B^-(x) \phi\ket = \bra\phi,\, \tH \phi\ket + \cno\alpha \|\phi\|^2\, ,
$$
and since from \cite[Lemma~A4]{GLL} we have, for $\alpha$ small enough
\begin{equation}\label{eq:photon-number-bound}
\begin{split}
 -\bra\pgs,\, H_f \pgs\ket - 2\Re\sqrt{\alpha}\bra \pgs,\, \sigma\cdot B^-(x) \pgs\ket
 \leq 0\, .
\end{split}
\end{equation}

Now we have, with $P:=i\nabla$, and the inequality $- 2 \Re P.A^-(x) \geq -P^2 - (A^-(x))^2$,
\begin{equation}\label{eq:pn-2}
 \begin{split}
  & 0 \geq \bra \pgs,\, \tH\,\pgs\ket \\
  & = \Re\Bra \pgs, \Big[ -\Delta - 4\sqrt{\alpha} P\cdot A^-(x) +
  \alpha :A(x):^2 \\
  & + 2\sqrt{\alpha}\sigma\cdot B^-(x) + H_f -\frac{\alpha}{|x|} \Big]\pgs\Ket \\
  & \geq - 2 \bra \pgs,\, \frac{\alpha}{|x|} \pgs\ket + \bra \pgs,\, (1 - 8\sqrt{\alpha}) (-\Delta)\pgs\ket \\
  & + 8\sqrt{\alpha}\bra \pgs,\, -\Delta \pgs\ket - 2 \sqrt{\alpha} \bra \pgs,\, P^2 \pgs\ket
  - 2\sqrt{\alpha} \bra \pgs,\, (A^-(x))^2\pgs\ket \\
  & + \frac14 \bra \pgs,\, H_f\pgs\ket + \frac12 \bra \pgs,\, H_f\pgs\ket + 2 \bra \pgs, \alpha:(A^-(x))^2:\pgs\ket \\
  & + \frac14 \bra \pgs,\, H_f\pgs\ket + 2\sqrt{\alpha} \Re\bra\pgs,\, \sigma\cdot B^-(x) \pgs\ket
  + \bra\pgs,\, \frac{\alpha}{|x|}\pgs\ket\, .
\end{split}
\end{equation}
In the right hand side of \eqref{eq:pn-2}, the sum of the first and the second term is bounded below by
$-c\alpha^2$ for some constant $c>0$, the sum of the third and fourth term is positive, and the sum of the
fifth and sixth term is positive using \cite[Lemma~A.4]{GLL}. Again using \cite[Lemma~A.4]{GLL}, we obtain that
the sum of the seventh and eight term
is larger than $-c\alpha^2$ and the sum of the ninth and tenth term is positive. Therefore, we obtain
\begin{equation}\label{eq:pn-3}
 \bra \pgs,\, \frac{\alpha}{|x|} \pgs\ket \leq c\alpha^2\, .
\end{equation}
Now \eqref{eq:pn-3} and \eqref{eq:pn-1} give
\begin{equation}\label{eq:pn-4}
  \| (i\nabla - \sqrt{\alpha} A(x)) \pgs\|^2 \leq c\alpha\, .
\end{equation}
We set
$$
  v := i\nabla - \sqrt{\alpha} A(x)\, .
$$
Using
\begin{equation}\nonumber
 [a_\lambda(k), H_f] = |k|,\quad[a_\lambda(k),v]
 =\frac{\epsilon_\lambda(k)}{2\pi|k|^\frac12} \zeta(|k|) \mathrm{e}^{ik.x}
\end{equation}
and applying the pull-through formula yields
\begin{equation}\label{pull-through-1}
\begin{split}
    & a_\lambda(k) E \pgs  =a_\lambda(k) \ttH  \pgs
    \nonumber\\
    &= \Big[( H_f +|k|)a_\lambda(k) - \frac{1}{|x|}a_\lambda(k)
    + [a_\lambda(k), v]v
    +v [a_\lambda(k), v] + v^2 a_\lambda(k) \\
    & + \sqrt{\alpha}\sigma\cdot B(x)a_\lambda(k)
    - i\sqrt{\alpha} \frac{\zeta(|k|)\,
    \sigma\cdot(\epsilon_{\lambda}(k) \wedge k) }{|k|^\frac12 2\pi}\,\mathrm{e}^{i k.x}\Big] \pgs \;.
\end{split}
\end{equation}
Thus
\begin{equation}\label{alk-psi-GLL-1}
\begin{split}
    & a_\lambda(k) \, \pgs \\
    & =- \, \frac{\sqrt\alpha \zeta(|k|)}{2\pi |k|^\frac12} \,
     \frac{2}{\ttH +|k|-E}  \Bigg( \Big(  i\nabla\!-\!\sqrt\alpha A(x)
      \Big)\cdot
    \epsilon_\lambda(k) \mathrm{e}^{ik.x} \\
    & +  i\sigma\cdot \frac{\epsilon_\lambda(k)\wedge k}{2\pi}\mathrm{e}^{ik.x}
    \Bigg)\, \pgs  \;.
\end{split}
\end{equation}
From \eqref{alk-psi-GLL-1}, we obtain
\begin{equation}\label{a-psi-fund-1}
\begin{split}
    \left\| \, a_\lambda(k)\pgs  \, \right\|
    &\leq
    c\, \frac{ \sqrt{\alpha} \, \zeta(|k|)}{|k|^\frac32 }\,
    \left\| \, \left( \, i\nabla-\sqrt\alpha A(x) \, \right) \pgs   \, \right\|
    +  c\, \frac{ \sqrt{\alpha} \, \zeta(|k|)}{|k|^\frac12 }\,
    \left\|  \pgs   \, \right\| \\
    &\leq
    c \left(  \frac{ \alpha
    }{|k|^\frac32 } +  \frac{\sqrt{\alpha}}{|k|^\frac12 }
     \right)\zeta(|k|) \|\pgs\|\, ,
\end{split}
\end{equation}
where we applied \eqref{eq:pn-4} in the last inequality.

This a priori bound exhibits the $L^2$-critical singularity in
frequency space. It does not take into consideration the
exponential localization of the ground state due to the confining
Coulomb potential.

To account for the latter, following the proof of \cite[Equation (14)]{BCVV2}, we use two results from
the work of Griesemer, Lieb, and Loss, \cite{GLL}. Equation (58)
in \cite{GLL} provides the bound
\begin{equation}\nonumber
    \Big\| \, a_\lambda(k) \pgs   \, \Big\| \; < \;
    \frac{c \, \sqrt\alpha \, \zeta(|k|)}{|k|^{\frac12} }
    \Big\| \, |x|\pgs   \, \Big\|\;.
\end{equation}
Moreover, Lemma 6.2 in \cite{GLL} states that
\begin{equation}\nonumber
    \Big\| \, \exp[\beta|x|] \pgs   \, \Big\|^2 \;
    \leq \; c \, \Big[1+\frac{1}{\Sigma_0-E-\beta^2}\Big]
    \, \| \, \pgs   \, \|^2\;,
\end{equation}
for any
\begin{equation}\nonumber
    \beta^2 \; < \; \Sigma_0-E \; = \; O(\alpha^2) \; .
\end{equation}
For the 1-electron case, $\Sigma_0$ is the infimum of the
self-energy operator, and $E$ is the ground state energy of $\tH
$. Choosing $\beta=O(\alpha)$ yields,
\begin{eqnarray*}
    \|\, |x|\pgs   \, \|&\leq& \||x|^4 \pgs  \|^{\frac14}
    \, \|\pgs  \|^{\frac34}
    \leq \frac{(4!)^{\frac14}}{\beta} \, \Big\|
    \, \exp[\beta|x|]\pgs   \, \Big\|^\frac14 \|\pgs  \|^\frac34
    \nonumber\\
    &\leq&\frac{c}{\beta} \,
    \Big[1+\frac{1}{\Sigma_0-E-\beta^2}\Big]^{\frac18}
    \,\| \, \pgs   \, \|
    \nonumber\\
    &\leq& c_1\alpha^{-\frac54} \;.
\end{eqnarray*}

Thus,
\begin{equation}\label{a-psi-fund-2}
    \Big\| \, a_\lambda(k) \pgs   \, \Big\|<\frac{c \alpha^{-\frac34}
    \zeta(|k|)}{|k|^{\frac12} } \;.
\end{equation}
We see that binding to the Coulomb potential weakens the infrared
singularity by a factor $|k|$, but at the expense of a large
constant factor $\alpha^{-2}$.

Using \eqref{a-psi-fund-1} and \eqref{a-psi-fund-2}, we find
\begin{eqnarray*}
    \bra \, \pgs   \, , \, N_f \, \pgs  \, \ket&=&\int  \, \Big\| \, a_{\lambda}(k)\pgs  \, \Big\|^2 \d k
    \nonumber\\
    &\leq&\int_{|k|<\delta} \, \frac{c \, \alpha^{-\frac32}} {|k| } \d k
    \, + \,
    \int_{\delta\leq |k|\leq\Lambda} \, \left( \frac{c \, \alpha^2}{|k|^3} +  \frac{c \, \alpha}{|k|}\right) \d k
    \nonumber\\
    &\leq& c \alpha^{-\frac32}\delta^2 + c\alpha^2 \log\delta + c\alpha
    \nonumber\\
    &\leq& c\alpha\, ,
\end{eqnarray*}
for $\delta=\alpha^{\frac{5}{4}}$. This proves the result.
\end{proof}


\section{Proof of Proposition~\ref{prop:orthogonality}}\label{appendix-proof-state-orthogonal-u-alpha}

The proof of Proposition~\ref{prop:orthogonality} follows the lines of the proof of Theorem~4.1 in \cite{BCVV2}. The main difference is that here there is no positive term involving $H_f$ in the right hand side of \eqref{eq:orth-u-alpha-0}, since it is of the order $\alpha$, unlike in the spinless case, and is thus much larger than the difference of the first and second level of the Schr\"odinger-Coulomb operator.

\medskip

Let $\gvar:=\gvar_1 + \gvar_2 := \chi(|P| < \frac{p_c}2) \gvar +
\chi(|P| \geq \frac{p_c}2) \gvar $, where $P=i\nabla_x$ is the
total momentum operator (due to the transformation
\eqref{eq-H-Utrsf-1}-\eqref{def:unitary-transform}) and $p_c=\frac13$ is a lower bound on the
norm of the total momentum for which \cite[Theorem 3.2]{Chen2008}
holds.

Since $P$ commutes with the translation invariant operator $T=H
+\frac{\alpha}{|x|}$, we have for all $\epsilon\in (0,1)$,
 \begin{equation}
 \begin{split}\label{eq:eq}
  &\bra H \gvar , \gvar\ket  =  \bra H \gvar_1, \gvar_1 \ket + \bra H \gvar_2,\gvar_2\ket
  - 2\Re \bra \frac{\alpha}{|x|} \gvar_1, \gvar_2\ket \\
  & \geq  \bra H \gvar_1, \gvar_1 \ket + \bra H \gvar_2,\gvar_2\ket
  - \epsilon \bra \frac{\alpha}{|x|} \gvar_1, \gvar_1\ket
  -\epsilon^{-1} \bra \frac{\alpha}{|x|} \gvar_2, \gvar_2\ket .
 \end{split}
 \end{equation}
$\bullet$ First, we have the following estimate
\begin{equation}\nonumber
\begin{split}
 & :(P-P_f -\sqrt{\alpha} A(0))^2: + H_f + 2\Re\sqrt{\alpha}\sigma.\Ba\\
 & = (P-P_f)^2 - 2\Re (P-P_f) . \sqrt{\alpha} A(0) + \alpha :A(0)^2: +
 H_f  + 2\Re\sqrt{\alpha} \sigma .\Ba\\
 & \geq  (1-\sqrt{\alpha})(P-P_f)^2 + (\alpha
 -\sqrt{\alpha}):A(0)^2: + H_f - \cno \sqrt{\alpha} -c\sqrt{\alpha} H_f
   \\
 & \geq  (1-\sqrt{\alpha})(P-P_f)^2 +
 (1- \mathcal{O}(\sqrt{\alpha})) H_f - \mathcal{O}(\sqrt{\alpha})
\end{split}
\end{equation}
where in the last inequality we used \cite[Lemma~A4]{GLL}. Therefore
\begin{equation}\label{eq:lem1-4}
\begin{split}
 \bra (H - \epsilon^{-1} \frac{\alpha}{|x|})\gvar_2, \gvar_2\ket
 \geq \bra (\frac{1-\sqrt{\alpha}}{2} (P-P_f)^2 -
 (1+\epsilon^{-1})\frac{\alpha}{|x|}) \gvar_2, \gvar_2\ket\\
 + \Big\bra \left( \frac{1-\sqrt{\alpha}}{2} (P-P_f)^2
 + (1-\mathcal{O}(\sqrt{\alpha}))H_f - \mathcal{O}(\sqrt{\alpha})\right)
 \gvar_2,\gvar_2
 \Big\ket
\end{split}
\end{equation}
The lowest eigenvalue of the Schr\"odinger operator
$-(1-\mathcal{O}(\sqrt{\alpha})) \frac{\Delta}2 -
\frac{(1+\epsilon^{-1})\alpha}{|x|}$ is larger than
$-c_\epsilon\alpha^2$. Thus, using \eqref{eq:lem1-4} and denoting
 $$
  L:=\frac{1-\sqrt{\alpha}}{2} (P-P_f)^2 +
  (1-\mathcal{O}(\sqrt{\alpha}))H_f - \mathcal{O}(\sqrt{\alpha})
  -c_\epsilon\alpha^2\ ,
 $$
we get
\begin{equation}\label{eq:lem1-5}
 \bra H\gvar_2, \gvar_2 \ket -\epsilon^{-1} \bra \frac{\alpha}{|x|}\gvar_2,\gvar_2\ket
 \geq \bra L\gvar_2, \gvar_2\ket .
\end{equation}
Now we have the following alternative: Either $\bra \gvar_2,\, P_f\gvar_2\ket \leq \frac{p_c}{3}\|\gvar_2\|^2$, in which case, using $\gvar_2 = \chi(|P|>\frac{p_c}{2})\gvar_2$ we get
$(\gvar_2, (P-P_f)^2 \gvar_2) \ge \frac{(\gvar_2, (P-P_f)\gvar_2)^2}{(\gvar_2,\, \gvar_2)} \ge \frac{p_c}{36}\| \gvar_2 \|^2$
 and thus $\bra L\gvar_2,\, \gvar_2\ket \geq (\frac{1}{72} p_c^2 -\mathcal{O}(\alpha^\frac12))\|\gvar_2\|^2$, or $\bra\gvar_2,\, P_f\gvar_2\ket \geq \frac{p_c}{3} \|\gvar_2\|^2$, in which case, using $H_f\geq |P_f|$ we obtain $\bra L\gvar_2,\, \gvar_2\ket \geq (\frac{p_c}{6} - \mathcal{O}(\alpha^\frac12))\|\gvar_2\|^2$.

In both cases, for $\alpha$ small enough, this yields the bound
\begin{equation}\label{eq:lem1-6}
  \bra L\gvar_2, \gvar_2\ket \geq \frac{p_c^2}{144} \|\gvar_2\|^2\geq
  (\Sigma_0 - e_0 + \frac78 (e_0-e_1)) \|\gvar_2\|^2
\end{equation}
since, for $\alpha$ small enough, the right hand side tends to
zero, whereas $p_c$ is a constant independent of $\alpha$.
Inequalities~\eqref{eq:lem1-5} and \eqref{eq:lem1-6} yield
\begin{equation}\label{eq:added2}
 \bra H\gvar_2, \gvar_2\ket - \epsilon^{-1}\bra \frac{\alpha}{|x|}\gvar_2, \gvar_2\ket
 \geq (\Sigma_0 - e_0 + \frac78 (e_0-e_1)) \|\gvar_2\|^2 .
\end{equation}

$\bullet$ For $T(p)$ being the self-energy operator with fixed
total momentum $p\in\R^3$ defined in \eqref{def:T(P)}, we have
from \cite[Theorem~3.1~(B)]{Chen2008}
\begin{equation}\nonumber
 \left|\inf\sigma(T(p)) - {p^2} -\inf\sigma(T(0))\right|
 \leq {c_0 \alpha p^2} .
\end{equation}
Therefore
\begin{equation}\label{eq:lem1-1}
    T(p) \geq (1- o_\alpha(1))p^2 + \Sigma_0  .
\end{equation}
Case 1: If $\|\gvar_2\|^2 \geq 8\|\gvar_1\|^2$, we first do the
following estimate, using \eqref{eq:lem1-1},
\begin{equation}\label{eq:added1}
\begin{split}
& \bra H \gvar_1, \gvar_1\ket -\epsilon \bra \frac{\alpha}{|x|} \gvar_1, \gvar_1\ket \\
& \geq (1- o_\alpha(1))(P^2\gvar_1, \gvar_1\ket - \bra
(1+\epsilon)\frac{\alpha}{|x|}\gvar_1, \gvar_1\ket +  \Sigma_0 \|\gvar_1\|^2\\
& \geq \left(\Sigma_0- (1+ \mathcal{O}(\alpha) +
\mathcal{O}(\epsilon)) e_0\right) \|\gvar_1\|^2 .
\end{split}
\end{equation}
Therefore, together with $\|\gvar_2\|^2 \geq 8\|\gvar_1\|^2$ and
\eqref{eq:added2}, for $\alpha$ and $\epsilon$ small enough this
implies
\begin{equation}\label{eq:eq30}
  \bra H \gvar, \gvar\ket \geq (\Sigma_0-e_0)
  \| \gvar\| ^2 + \frac34 (e_0-e_1)\| \gvar\| ^2 .
\end{equation}

Case 2: If $\|\gvar_2\|^2 < 8\|\gvar_1\|^2$, we write the estimate
\begin{equation}\label{eq:added3}
\begin{split}
 & \bra H \gvar_1, \gvar_1\ket -\epsilon \bra \frac{\alpha}{|x|} \gvar_1, \gvar_1\ket
  \\
 \geq & (1- o_\alpha(1))(P^2\gvar_1, \gvar_1\ket - \bra
 (1+\epsilon)\frac{\alpha}{|x|}\gvar_1, \gvar_1\ket +  \Sigma_0 \|\gvar_1\|^2\\
 \geq & (1\!+\! o_\alpha(1)\!+\!\mathcal{O}(\epsilon))
 \Bigg(- e_0\! \sum_{k=0}^\infty \|\, \bra \Proj^k \gvar_1,
 \Sgsa \ket_{L^2(\R^3, \mathrm{d}x)}\,\|^2 \!-\! e_1(\|\gvar_1\|^2 \\
 & - \! \sum_{k=0}^\infty
 \|\, \bra \Proj^k \gvar_1, \Sgsa \ket_{L^2(\R^3, \mathrm{d}x)}
 \,\|^2) \Bigg)
 +  \Sigma_0 \|\gvar_1\|^2 .
\end{split}
\end{equation}
Now, by orthogonality of $\gvar $  and $\Sgsa $ in the sense that
for all $k$, $\bra \Proj^k \gvar, \Sgsa  \ket_{L^2(\R^3,
\mathrm{d}x)} =0$, we get
\begin{equation}\label{eq:eq31}
\begin{split}
  & \sum_{k=0}^\infty \|\bra \Sgsa , \Proj^k \gvar_1
  \ket_{L^2(\R^3, \mathrm{d}x)}\|^2
  = \sum_{k=0}^\infty
  \|\bra \Sgsa  , \Proj^k
  \gvar_2\ket_{L^2(\R^3, \mathrm{d}x)}\|^2 \\
  & \leq \|\gvar_2\|^2
  \|  \chi(|P| \geq \frac{p_c}2)\Sgsa\|^2
  \leq 8 \|\gvar_1\|^2  \| \chi(|P| \geq \frac{p_c}2)\Sgsa\|^2
   \displaystyle\rightarrow_{\alpha\rightarrow 0} 0\, .
\end{split}
\end{equation}
Thus, for $\alpha$ and $\epsilon$ small enough, \eqref{eq:added3},
\eqref{eq:eq31}, \eqref{eq:added2} and the fact that $\|\chi(|P|>\frac{p_c}{2})u_\alpha\|\to_{\alpha\to 0}0$ imply also \eqref{eq:eq30}
in that case.

This concludes the proof of Proposition~\ref{prop:orthogonality}.

\section{Technical results}\label{appendix2}

\begin{lem}\label{lem:appendix-1}
For $\gr$ defined by \eqref{def:gr} and $\Gu$ defined by \eqref{eq:3.0}, we have for all $\alpha>0$
 $$
  \bra P.\Ac \vac\Sgsa\su,\, \gr\ket = 0
  \quad \mbox{and}\quad
  \bra \gr,\, \ups\ket_* = 0\, .
 $$
\end{lem}

\begin{proof}
This is a straightforward computation.
\end{proof}

\begin{lem}\label{lem:appendix-2}
For $|\mab|=1$, we have
$$
 \sum_{i=1}^3 \| (H_f+P_f^2)^{-1} (P_f)^{(i)} \Guab\|_*^2  = \frac{2}{\pi}\int_0^\infty \frac{t^2\, \zeta^2(t)}{(1+t)^3}\, \d t
$$
and
$$
 \sum_{i=1}^3 \|\chi_{(\alpha,\infty)}(H_f)(H_f+P_f^2)^{-1} (\Ac)^{(i)}\mab\vac\|_*^2 =
 \frac{2}{\pi} \int_0^\infty \frac{\zeta^2(t)}{1+t} \d t \ +\ \mathcal{O}(\alpha)\, .
$$

In particular, for $\gr$ defined by \eqref{def:gr} and $\ups$ defined by \eqref{def:ups}, we have
$$
 \| \gr\|_*^2  = \frac{2\alpha^2}{3\pi} \int_0^\infty \frac{t^2\, \zeta^2(t)}{(1+t)^3}\, \d t\, ,\quad\mbox{and}\quad
 \| \ups\|_*^2 = \frac{2\alpha^2}{3\pi} \int_0^\infty \frac{\zeta^2(t)}{1+t} \d t \ +\ \mathcal{O}(\alpha^3)\, .
$$
\end{lem}
\begin{proof}
We denote by $\uparrow$ and $\downarrow$ the first and second component of a $\C^2$ vector, and by $\lambda_1$ and $\lambda_2$ the two polarizations of transverse photons. Then for $\varphi\in\Proj_1(\C^2\otimes\gF)$ we set
\begin{equation}\label{eq:def-four-components}
 \varphi = :
 \begin{pmatrix}
    \varphi(\uparrow,\, k, \lambda_1,) \\
    \varphi(\uparrow,\, k, \lambda_2,) \\
    \varphi(\downarrow,\, k, \lambda_1,) \\
    \varphi(\downarrow,\, k, \lambda_2)
 \end{pmatrix}\, .
\end{equation}
The proof is thus a straightforward computation using the definition of $\gr$ and $\ups$, the fact that (see e.g. \cite{BLV})
\begin{equation}\nonumber
 \sigma\!\cdot\!\Bc \vac\zo = \frac{-i\, \zeta(|k|)}{ 2\pi |k|^\frac12}
 \!\left(\!\!\begin{array}{l}
 -\sqrt{k_1^2 + k_2^2} \\
 0 \\
 \frac{(k_1+ ik_2)k_3}{\sqrt{k_1^2 +k_2^2}}\\
 \frac{|k|(-k_2 +ik_1)}{\sqrt{k_1^2 + k_2^2}}
 \end{array}\!\!\right)
 \, ,
\end{equation}
and
$$
 (H_f+P_f^2)^{-1}  \Ac \vac  = \frac{\zeta(|k|)}{2\pi|k|^\frac12 (|k|^2 + |k|) \sqrt{k_1^2 + k_2^2}}\left(\!\!
    \begin{array}{l}
        k_2 +\frac{k_1k_3}{|k|} \\
        -k_1 + \frac{k_2k_3}{|k|} \\
        \frac{-(k_1^2 + k_2^2)}{|k|}
    \end{array}\!\!\right)\, .
$$
\end{proof}

\begin{lem}\label{lem:IPNE}
$$
 \sum_{\lambda=1,2} \int_{|k|\leq\alpha^\frac13} \| a_\lambda(k) \Guab \|^2 \mathrm{d}k \leq c\alpha^\frac23\, .
$$
\end{lem}
\begin{proof}
Let $\mab\in\C^2$. With the notation \eqref{eq:def-four-components}, we have the following identity
\begin{equation}\label{eq:explicit-Guab}
\Guab (k,\lambda) =
 \frac{i\, \zeta(|k|)}{2\pi (|k|+|k|^2)}
 \begin{pmatrix}
  \frac{1}{|k|^\frac12} (-a \sqrt{k_1^2 + k_2^2} + b\,\frac{(k_1 - ik_2)k_3}{\sqrt{k_1^2 + k_2^2}}) \\
  b \,  |k|^\frac12 \frac{(-k_2 -i k_1) }{\sqrt{k_1^2 +k_2^2}} \\
  \frac{1}{|k|^\frac12} (b \sqrt{k_1^2 + k_2^2} + a\,\frac{(k_1 + ik_2)k_3}{\sqrt{k_1^2 + k_2^2}}) \\
  a \, |k|^\frac12 \frac{(-k_2 +i k_1) }{\sqrt{k_1^2 +k_2^2}}
 \end{pmatrix}
\end{equation}

A straightforward computation using \eqref{eq:explicit-Guab} yields
 $$
  \| a_\lambda(k) \Guab \| \leq c |k|^{-1}\, ,
 $$
 which implies the result.
\end{proof}

\begin{lem}\label{lem:appendix-3}
For all $\mab$ and $\mtab$ in $\C^2$, for all $i=1,2,3$, we have, for $\Guab$ defined by \eqref{def:Guab}
\begin{equation}\nonumber
\begin{split}
 \left\bra (H_f+P_f^2)^{-1} (\Ac)^{(i)}\vac \mab ,\,
 (H_f+P_f^2)^{-1}  (P_f)^{(i)} \Guab \right\ket_* \, =0 ,
\end{split}
\end{equation}
and
\begin{equation}\nonumber
\begin{split}
     \left\bra  \Gutab ,\, (P_f)^{(i)} \Guab \right\ket
   = \left\bra  \Gutab ,\, (P_f)^{(i)} \Guab \right\ket_* \, =0 \, .
\end{split}
\end{equation}
\end{lem}
\begin{proof}
Straightforward computations using \eqref{eq:explicit-Guab}.
\end{proof}

\begin{lem}\label{lem:appendix-4}
For all $\mab$, $\mtab$ and $i=1,2,3$ we have
\begin{equation}\nonumber
 \bra (\Ac)^{(i)} \mtab\vac, \Guab\ket = 0\, .
\end{equation}
\end{lem}
\begin{proof}
Straightforward computations using \eqref{eq:explicit-Guab}.
\end{proof}

\begin{lem}\label{lem:appendix-6}
For $g\in L^2(\R^3)\otimes\C^2$, and $\mab\!\in\!\C^2$ such that $|a|^2 + |b|^2=1$,
we have, for some constant $c>0$
\begin{flalign}
 & \|\Guab\|_* = \|\Gu^{(1,0)}\|_* \, , & \label{eq:app-1} \\
 & \|\Gu(g)\|_* = \|\Gu^{(1,0)}\|_*\, \|g\|\, ,\quad
   \|\Gu(g)\| = \|\Gu^{(1,0)}\|\, \|g\|\, , & \label{eq:app-2} \\
 & \|\Aa \Gu(g) \| = \|\Aa\Guab\| \, \|g\|\, , & \label{eq:app-4} \\
 & \|\Gd(g)\|_* = \|\Gdab\|_*\, \|g\| + \mathcal{O}(\alpha^3) \, , \label{eq:app-3} \\
 &
 \|\Gd(g)\| \leq c \, \|g\|\, . \label{eq:app-3-bis} &
\end{flalign}

Moreover, we have
\begin{flalign}
 & \|\nabla \Gu(g) \| = \|\Gu^{(1,0)}\|\, \|\nabla g\|\, , \label{eq:app-5} & \\
 &  \|\nabla \Gd(g)\| \leq c \, \|\nabla g\| \label{eq:app-6} \, .&
\end{flalign}
\end{lem}
\begin{proof}
For $g\in L^2(\R^3)\otimes\C^2$, let us denote
$$
 g(x) = : \mabx\, .
$$
A straightforward computation of $\|\Gu(g)\|_*^2$ using \eqref{eq:explicit-Guab} yields
 $$
  \|\Gu(g)\|_*^2 = \sum_{\lambda=1,2} \int |\Gu(k,\lambda)|^2 \mathrm{d} k  \int |a(x)|^2 + |b(x)|^2 \mathrm{d} x\, ,
 $$
which concludes the proof of \eqref{eq:app-1} and \eqref{eq:app-2}.

The proof of \eqref{eq:app-4} is similar.

The proof of  \eqref{eq:app-3} can be achieved by a straightforward but rather long computation. Using the degeneracy of the ground state for $T(0)$ \cite{LMS, HiSp, Chen2008} yields an alternative proof, which is as follows. According to \cite[Theorem~2.1]{BV4}, and using \eqref{eq:app-1}, \eqref{eq:app-2} and \eqref{eq:app-4} yields that for any $\mab\in\C^2$ with $|a|^2 + |b|^2 =1$
\begin{equation}\label{eq:app-*}
  \|\Gdab\|_{*,\C^2\otimes\gF}^2 = \|\Gd^{(1,0)}\|_{*,\C^2\otimes\gF}^2 + \mathcal{O}(\alpha^3)\, .
\end{equation}
where $\|.\|_{*,\C^2\otimes\gF}$ is the $*$-norm in $\C^2\otimes\gF$.
Now, due to \eqref{eq:explicit-Guab}, we have for $\mu\in\C$
 $$
  \Gu^{(\mu a, \mu b)} (k,\lambda) = \mu \Guab (k,\lambda)\, .
 $$
By definition of $\Gdab$ this obviously yields
\begin{equation}\label{eq:app-**}
 \Gd^{(\mu a,\mu b)} ( k,\lambda;\, \tilde k, \tilde\lambda)
 = \mu \Gdab ( k,\lambda;\, \tilde k, \tilde\lambda)\, .
\end{equation}
Thus, with \eqref{eq:app-*} and \eqref{eq:app-**}, we have for any $x\in\R^3$
\begin{equation}\nonumber
\begin{split}
 & \| \Gd(g(x)) \|_{*,\C^2\otimes\gF}^2 \\
 & = \| \sqrt{|a(x)|^2 + |b(x)|^2} \Gd^{(a(x)/\sqrt{|a(x)|^2 + |b(x)|^2} ,b(x)/\sqrt{|a(x)|^2 + |b(x)|^2})}\|_{*,\C^2\otimes\gF}^2  \\
 & =  (|a(x)|^2 + |b(x)|^2) \|\Gd^{(1,0)}\|_{*,\C^2\otimes\gF}^2 +
 (|a(x)|^2 + |b(x)|^2)\mathcal{O}(\alpha^3)\, .
\end{split}
\end{equation}
Therefore we obtain
$$
 \| \Gd(g) \|_*^2 = \| \Gd^{(1,0)} \|_{*,\gF}^2 \|g\|^2  + \mathcal{O}(\alpha^3)\, ,
$$
which concludes the proof of \eqref{eq:app-3}, using again \eqref{eq:app-*}.

To prove \eqref{eq:app-3-bis}, it suffices to remark that from \eqref{eq:explicit-Guab} and the definition of $\Gd(g)$, we have that each term occurring in the components of the vector $\Gd(g(x))$ has a prefactor $a(x)$ or $b(x)$ multiplied by a function independent of $x$. Thus, the computation of $\|\Gd(g(x))\|_{*,\gF}^2$ give a sum of terms of the form $|a(x)|^2 |f(k,\lambda;\tilde k,\tilde\lambda)|^2$ or $|b(x)|^2 |f(k,\lambda;\tilde k,\tilde\lambda)|^2$ or $\Re (a(x) \bar b(x)) |f(k,\lambda;\tilde k,\tilde\lambda)|^2$. Therefore, bounded above $|a(x)|^2$, $|b(x)|^2$ and $\Re (a(x)\bar b(x))$ by $|a(x)|^2 + |b(x)|^2$ yields that for some constant $c$ independent of $x$, we have the bound
$$
 \| \Gd((g(x))\|_{*,\gF}^2\leq c (|a(x)|^2 + |b(x)|^2) \, ,
$$
which proves \eqref{eq:app-3-bis}.

From \eqref{eq:explicit-Guab} we obtain for $i=1,2,3$
 $$
  \frac{\partial}{\partial x_i} \Gu(g) = \Gu(\frac{\partial}{\partial x_i} g)\, ,
 $$
 which in its turn imply \eqref{eq:app-5}.

To prove \eqref{eq:app-6} it suffices to remark that also holds
$$
 \frac{\partial}{\partial_{x_i}} \Gd(g) = \Gd(\frac{\partial}{\partial_{x_i}}g)\, .
$$
\end{proof}


\subsection*{Acknowledgements} J.-M.B. gratefully acknowledges financial
support from Agence Nationale de la Recherche (projects HAM-MARK ANR-09-BLAN-0098-01 and ANR-12-JS01-0008-01). J.-M.B. was also supported by BFHZ/CCUFB with the Franco-Bavarian cooperation project ``Mathematical models of radiative and relativistic effects in atoms and molecules''. S.V. was supported by the
DFG Project SFB TR 12-3. Part of this work was done during the stay of S.V. at the University of Toulon.


\end{document}